\documentclass
[
fontsize = 11 pt,        
american,                
captions = tableheading, 
numbers = noenddot,      
abstracton,
paper = letter,
]
{scrartcl}

\usepackage[utf8]{inputenc}

\usepackage[T1]{fontenc}
\usepackage[utf8]{inputenc}
\usepackage
[
    babel = true, 
]
{microtype}           
\usepackage{csquotes} 

\usepackage[dvipsnames]{xcolor} 

\definecolor{stroke1}{HTML}{2574A9} 

\usepackage{amsmath}
\usepackage{amssymb}
\usepackage{amsthm}
\usepackage{thmtools}
\usepackage{mathtools}
\usepackage{thm-restate}
\usepackage{dsfont}        

\usepackage{url} 
\usepackage{bm}  

\usepackage{graphicx} 
\usepackage
[
    subrefformat = simple, 
    labelformat = simple,  
]
{subcaption}         

\usepackage{array}     
\usepackage{booktabs}  
\usepackage{longtable} 
\usepackage{pdflscape} 

\usepackage[shortlabels]{enumitem} 

\usepackage
[
    ruled,         
    vlined,        
    linesnumbered, 
]
{algorithm2e} 

\usepackage{xspace}   

\usepackage{xparse}    
\usepackage{footnote}  
\usepackage{afterpage} 
\usepackage
[
    textsize = scriptsize, 
]
{todonotes}            

\usepackage
[
    bookmarks = true,                 
    bookmarksopen = false,            
    bookmarksnumbered = true,         
    pdfstartpage = 1,                 
]
{hyperref} 

\usepackage
[
    noabbrev,   
    nameinlink, 
]
{cleveref} 

\usepackage[rightcaption]{sidecap}

\usepackage[numbers]{natbib}

\usepackage{ifthen}
\makeatletter
    \def\IfEmptyTF#1%
    {%
        \if\relax\detokenize{#1}\relax%
            \expandafter\@firstoftwo%
        \else%
            \expandafter\@secondoftwo%
        \fi%
    }
\makeatother

\NewDocumentCommand{\mathOrText}{m}
{%
    \ensuremath{#1}\xspace%
}

\let\originalleft\left
\let\originalright\right
\renewcommand{\left}{\mathopen{}\mathclose\bgroup\originalleft}
\renewcommand{\right}{\aftergroup\egroup\originalright}

\makeatletter
    \DeclareRobustCommand{\bfseries}%
    {%
        \not@math@alphabet\bfseries\mathbf%
        \fontseries\bfdefault\selectfont%
        \boldmath%
    }

\xspaceaddexceptions{]\}}

\allowdisplaybreaks

\crefname{ineq}{inequality}{inequalities}
\creflabelformat{ineq}{#2{\upshape(#1)}#3}

\crefname{term}{term}{terms}
\creflabelformat{term}{#2{\upshape(#1)}#3}

\crefname{cond}{condition}{conditions}
\creflabelformat{cond}{#2{\upshape(#1)}#3}

\crefname{assume}{assumption}{assumptions}
\creflabelformat{assume}{#2{\upshape(#1)}#3}

\let\oldfootnote\footnote

\newlength{\spaceBeforeFootnote} 
\newlength{\spaceAfterFootnote}  

\RenewDocumentCommand{\footnote}{o o o m}%
{%
    \IfNoValueTF{#1}%
    {%
        \oldfootnote{#4}%
    }%
    {%
        \setlength{\spaceBeforeFootnote}{\IfEmptyTF{#1}{0}{#1} em}%
        \IfNoValueTF{#2}%
        {%
            \hspace*{\spaceBeforeFootnote}\oldfootnote{#4}%
        }%
        {%
            \setlength{\spaceAfterFootnote}{\IfEmptyTF{#2}{0}{#2} em}%
            \hspace*{\spaceBeforeFootnote}\IfNoValueTF{#3}{\oldfootnote{#4}}{\oldfootnote[#3]{#4}}\hspace*{\spaceAfterFootnote}%
        }%
    }%
}

\makesavenoteenv{figure}
\makesavenoteenv{table}
\makesavenoteenv{tabular}

\declaretheorem
[
   	name = Claim,
    sharenumber = conjecture,
]
{claim}
\declaretheorem
[
   	name = Lemma,
    sharenumber = conjecture,
]
{lemma}
\declaretheorem
[
   	name = Corollary,
    sharenumber = conjecture,
]
{corollary}
\declaretheorem
[
   	name = Theorem,
    sharenumber = conjecture,
]
{theorem}
\declaretheorem
[
   	name = Definition,
    sharenumber = conjecture,
]
{definition}

\declaretheorem
[
    name = Remark,
    sharenumber = conjecture,
]
{remark}
\declaretheorem
[
    name = Observation,
    sharenumber = conjecture,
]
{observation}

\NewDocumentCommand{\functionTemplate}{m m m m o}%
{%
    \IfNoValueTF{#5}%
    {%
        \mathOrText{#1\left#2{#4}\right#3}%
    }%
    {%
        \mathOrText{#1#5#2{#4}#5#3}%
    }%
}

\newcommand*{\leftBracketType}{(}
\newcommand*{\rightBracketType}{)}

\NewDocumentCommand{\createFunction}{m m o o}%
{%
    \renewcommand*{\leftBracketType}{\IfNoValueTF{#3}{(}{#3}}%
    \renewcommand*{\rightBracketType}{\IfNoValueTF{#4}{)}{#4}}%
    \NewDocumentCommand{#1}{o o}%
    {%
        \IfNoValueTF{##1}%
        {%
            \mathOrText{#2}%
        }%
        {%
            \functionTemplate{#2}{\leftBracketType}{\rightBracketType}{##1}[##2]%
        }%
    }%
}

\DeclareDocumentCommand{\probabilisticFunctionTemplate}{m m O{} O{} o}
{%
    \functionTemplate{#1\IfEmptyTF{#4}{}{_{#4}}}%
    {\lbrack}%
    {\rbrack}%
    {#2\IfEmptyTF{#3}{}{\ \IfNoValueTF{#5}{\left}{#5}\vert\ \vphantom{#2}#3\IfNoValueTF{#5}{\right.}{}}}%
    [#5]%
}

\newcommand*{\N}{\mathOrText{\mathds{N}}}
\newcommand*{\Z}{\mathOrText{\mathds{Z}}}

\newcommand*{\R}{\mathOrText{\mathds{R}}}

\newcommand*{\indicatorFunctionSymbol}{\mathds{1}}

\RenewDocumentCommand{\Pr}{m O{} O{} o}%
{%
    \probabilisticFunctionTemplate{\mathds{P}}{#1}[#2][#3][#4]%
}

\NewDocumentCommand{\E}{m O{} O{} O{} o}%
{%
    \IfEmptyTF{#3}
    {
    	\probabilisticFunctionTemplate{\mathds{E}}{#1}[#2][#4][#5]
    }
    {
    	\IfEmptyTF{#2}
    	{\probabilisticFunctionTemplate{\mathds{E}}{#1}[#3][#4][#5]}
    	{\probabilisticFunctionTemplate{\mathds{E}}{#1}[#2;#3][#4][#5]}
    }
}

\NewDocumentCommand{\Var}{m O{} O{} o}%
{%
    \probabilisticFunctionTemplate{\mathrm{Var}}{#1}[#2][#3][#4]%
}

\DeclareDocumentCommand{\bigO}{m o}%
{%
    \functionTemplate{\mathrm{O}}{(}{)}{#1}[#2]%
}

\DeclareDocumentCommand{\bigOTilde}{m o}%
{%
	\functionTemplate{\widetilde{\mathrm{O}}}{(}{)}{#1}[#2]%
}

\DeclareDocumentCommand{\smallO}{m o}%
{%
    \functionTemplate{\mathrm{o}}{(}{)}{#1}[#2]%
}

\DeclareDocumentCommand{\bigTheta}{m o}%
{%
    \functionTemplate{\Theta}{(}{)}{#1}[#2]%
}

\DeclareDocumentCommand{\bigOmega}{m o}%
{%
    \functionTemplate{\upOmega}{(}{)}{#1}[#2]%
}

\DeclareDocumentCommand{\smallOmega}{m o}%
{%
    \functionTemplate{\upomega}{(}{)}{#1}[#2]%
}

\DeclareDocumentCommand{\eulerE}{o}%
{%
    \mathOrText{\mathrm{e}\IfNoValueTF{#1}{}{^{#1}}}%
}

\DeclareDocumentCommand{\poly}{m o}%
{%
    \functionTemplate{\mathrm{poly}}{(}{)}{#1}[#2]%
}

\DeclareDocumentCommand{\exponential}{m o}%
{%
	\functionTemplate{\mathrm{exp}}{(}{)}{#1}[#2]%
}

\createFunction{\id}{\mathrm{id}}

\NewDocumentCommand{\ind}{m o o}%
{%
    \IfNoValueTF{#2}%
    {%
        \mathOrText{\indicatorFunctionSymbol_{#1}}%
    }%
    {%
        \functionTemplate{\indicatorFunctionSymbol_{#1}}{(}{)}{#2}[#3]%
    }%
}

\DeclareDocumentCommand{\dom}{m o}%
{%
    \functionTemplate{\mathrm{dom}}{(}{)}{#1}[#2]%
}

\DeclareDocumentCommand{\rng}{m o}%
{%
    \functionTemplate{\mathrm{rng}}{(}{)}{#1}[#2]%
}

\DeclareDocumentCommand{\d}{o}%
{%
    \mathrm{d}\IfNoValueTF{#1}{}{^{#1}}%
}

\DeclareDocumentCommand{\set}{m O{} o}%
{%
    \mathOrText{
    	\IfNoValueTF{#3}{\left}{#3}\{
    		#1 
    		\IfEmptyTF{#2}{}{\ \IfNoValueTF{#3}{\left}{#3}\vert\ \vphantom{#1}#2\IfNoValueTF{#3}{\right.}{}}
    	\IfNoValueTF{#3}{\right}{#3}\}}%
}

\newcommand*{\size}[1]{\mathOrText{\left\vert #1 \right\vert}}

\newcommand*{\absolute}[1]{\mathOrText{\left\vert #1 \right\vert}}

\newcommand*{\powerset}[1]{\mathOrText{2^{#1}}}

\newcommand{\borel}{\mathcal{B}}
\newcommand{\boundedBorel}{\borel_{b}}
\newcommand{\comp}{c}

\DeclareDocumentCommand{\potential}{o o}%
{%
	\mathOrText{%
		\phi \IfNoValueF{#2}{\left(#1, #2\right)}
	}%
}

\DeclareDocumentCommand{\hamiltonian}{o}%
{%
	\mathOrText{%
		H \IfNoValueF{#1}{\left(#1\right)}
	}%
}

\newcommand*{\activity}{\mathOrText{\lambda}}

\DeclareDocumentCommand{\activityFunction}{o}%
{%
	\mathOrText{%
		\myMathbold{\lambda} \IfNoValueF{#1}{\left(#1\right)}
	}%
}

\newcommand{\projection}[1]{\pi_{#1}}

\DeclareDocumentCommand{\gibbs}{o o}%
{%
	\mathOrText{%
		\mu \IfNoValueTF{#2}
        {\IfNoValueF{#1}{_{#1}}}
        {_{#1}^{#2}}
	}%
}

\DeclareDocumentCommand{\partitionFunction}{o o}%
{%
	\mathOrText{%
		Z \IfNoValueTF{#2}
        {\IfNoValueF{#1}{_{#1}}}
        {_{#1}^{#2}}
	}%
}

\newcommand*{\dtv}[2]{\left\lvert #1 - #2\right\rvert_{TV}}

\newcommand{\dtvProjected}[3]{\left\lvert #1 - #2\right\rvert_{#3}}

\NewDocumentCommand{\PrSymbol}{o}%
{%
	\mathOrText{\mathds{P}\IfNoValueF{#1}{_{#1}}}
}

\newcommand*{\symmDiff}{\mathOrText{\ominus}}

\newcommand*{\norm}[2]{\mathOrText{\left\lVert #1 \right\rVert_{#2}}}

\newcommand*{\myMathbold}[1]{\mathOrText{\pmb{#1}}}

\newcommand*{\intD}{\mathOrText{\, \text{d}}}

\newcommand*{\vectorize}[1]{\mathOrText{\myMathbold{#1}}}

\newcommand*{\setFromTuple}[1]{\mathOrText{\eta_{#1}}}

\newcommand*{\complementOf}[1]{\mathOrText{(#1)^{\mathrm{c}}}}

\newcommand*{\supp}[1]{\mathOrText{\mathrm{supp}\left(#1\right)}}

\newcommand*{\distSymbol}{\mathOrText{\mathrm{dist}}}

\newcommand*{\dist}[2]{\mathOrText{\mathrm{\distSymbol}\left(#1, #2\right)}}

\newcommand*{\volume}[1]{\mathOrText{\size{#1}}}

\newcommand*{\dimensions}{\mathOrText{d}}

\newcommand*{\range}{\mathOrText{r}}

\newcommand*{\sidelength}{\mathOrText{L}}

\newcommand*{\region}{\mathOrText{\Lambda}}

\newcommand*{\subregion}{\mathOrText{{\Lambda'}}}

\newcommand*{\boxRegion}[1]{\mathOrText{\Lambda_{#1}}}

\DeclareDocumentCommand{\pointsets}{o}%
{%
	\mathOrText{%
		\mathcal{N} \IfNoValueF{#1}{_{#1}}
	}%
}

\DeclareDocumentCommand{\pointsetEvents}{o}%
{%
	\mathOrText{%
		\mathfrak{R} \IfNoValueF{#1}{_{#1}}
	}%
}

\DeclareDocumentCommand{\PPP}{m o}%
{%
	\mathOrText{P_{#1 \IfNoValueF{#2}{, #2}}}
}

\DeclareDocumentCommand{\potentialHS}{o o}%
{%
	\mathOrText{%
		\IfNoValueTF{#2}{\potential}{\potential[#1][#2]}
	}%
}

\DeclareDocumentCommand{\partitionFunctionApprox}{o o o}%
{%
	\mathOrText{%
		\IfEmptyTF{#3}{\hat{Z}}{\hat{Z}\left(#1, #2, #3\right)}
	}%
}

\DeclareDocumentCommand{\validHS}{o o}%
{%
	\mathOrText{%
		\IfNoValueTF{#1}{D}{D\left(#1 \IfNoValueF{#2}{\mid #2}\right)}
	}%
}

\DeclareDocumentCommand{\density}{o O{}}%
{%
	\mathOrText{%
		\rho \IfNoValueF{#1}{_{#1}} \IfEmptyTF{#2}{}{\left(#2\right)}
	}%
}

\newcommand*{\numberBoxes}{\mathOrText{N}}

\newcommand*{\boxIds}{\mathOrText{\mathcal{V}}}

\newcommand*{\ball}[2]{\mathOrText{\mathds{B}_{#2}\left(#1\right)}}

\DeclareDocumentCommand{\ball}{m m o}%
{%
	\mathOrText{%
		\mathds{B}\IfNoValueF{#3}{^{(#3)}}_{#2}\left(#1\right) 
	}%
}

\newcommand*{\iterationNumber}{\mathOrText{T}}

\newcommand*{\iterationTime}[1]{\mathOrText{R_{#1}}}

\newcommand*{\boxesToFix}[1]{\mathOrText{\mathcal{U}_{#1}}}

\newcommand*{\boxChosen}[1]{\mathOrText{\vectorize{u}_{#1}}}

\DeclareDocumentCommand{\boxesToUpdate}{o o o o}
{
	\mathOrText{B \IfNoValueF{#3}{\left(#1, #2, #3\right)}}
}

\DeclareDocumentCommand{\filterProbability}{o o o o o o}
{
	\mathOrText{p \IfNoValueF{#6}{\left(#1, #2, #3, #4, #5, #6\right)}}
}

\DeclareDocumentCommand{\bayesFilterCorrection}{o o o}
{
	\mathOrText{C \IfNoValueF{#3}{\left(#1, #2, #3\right)}}
}

\DeclareDocumentCommand{\bayesFilterCorrectionHS}{m o o o}
{
	\mathOrText{C_{#1} \IfNoValueF{#4}{\left(#2, #3, #4\right)}}
}

\DeclareDocumentCommand{\bayesFilterCorrectionR}{m m o o o}
{
	\mathOrText{C_{#1, #2} \IfNoValueF{#5}{\left(#3, #4, #5\right)}}
}

\newcommand*{\currentPointset}[1]{\mathOrText{X_{#1}}}

\newcommand*{\resampledPointset}{\mathOrText{Y}}

\newcommand*{\bayesFilter}[1]{\mathOrText{F_{#1}}}

\newcommand{\updateRadius}{
	\mathOrText{\ell}
}

\DeclareDocumentCommand{\correction}{o o}
{
	\mathOrText{\delta \IfNoValueF{#2}{\left(#1, #2\right)}}
}

\newcommand*{\statespace}{\mathOrText{\Omega}}

\newcommand*{\sigmafield}{\mathOrText{\mathcal{A}}}

\newcommand*{\Ber}[1]{\mathOrText{\mathrm{Ber}\left(#1\right)}}

\SetKwBlock{Loop}{Loop}{EndLoop}

\title{Perfect Sampling for Hard Spheres from Strong Spatial Mixing}

\newboolean{showAuthors}
\setboolean{showAuthors}{true} 

\ifshowAuthors
\author{Konrad Anand\thanks{School of Informatics, University of Edinburgh}, Andreas G\"obel\thanks{Hasso Plattner Institute, University of Potsdam}, Marcus Pappik$^\dagger$\!\!, Will Perkins\thanks{School of Computer Science, Georgia Institute of Technology}}
\date{\today}

\else
\author{Anonymous Author(s)}
\date{}
\fi

\begin{document}

\maketitle

\begin{abstract}
    We provide a perfect sampling algorithm for the hard-sphere model on subsets of $\R^d$ with expected running time linear in the volume under the assumption of strong spatial mixing.   A large number of perfect and approximate sampling algorithms have been devised to sample from the hard-sphere model, and  our perfect sampling algorithm is efficient for a range of parameters for which only efficient approximate samplers were previously known and is faster than these known approximate approaches.  Our methods also extend to the more general setting of Gibbs point processes interacting via finite-range, repulsive potentials.
\end{abstract}

\section{Introduction}
\label{secIntro}

Gibbs point processes, or classical gases, are mathematical models of interacting particles.  In statistical physics they are used to model gases, fluids, and crystals, while in other fields they are used to model spatial phenomena such as the growth of trees in a forest, the distribution of stars in the universe, or the location of cities on a map (see e.g.~\cite{ruelle1999statistical,moller2003statistical,van2000markov,dereudre2019introduction}). 

Perhaps the longest and most intensively studied Gibbs point process is the hard-sphere model: a  model of a gas in which the only interaction between particles is a hard-core exclusion in a given radius around each particle. That is, it is a model of a  random packing of equal-sized spheres.   Despite the simplicity of its definition, the hard-sphere model is expected to exhibit the  qualitative behavior of a real gas~\cite{alder1957phase}, and in particular exhibits gas, liquid, and solid phases, thus giving evidence for the hypothesis, dating back to at least Boltzmann, that the macroscopic properties of a gas or fluid are determined by its microscopic interactions. This rich behavior exhibited by the hard-sphere model is  very difficult to analyze rigorously, and the most fundamental questions about phase transitions in this model are open mathematical problems~\cite{ruelle1999statistical,lowen2000fun}.

In studying the hard-sphere model (or Gibbs point processes more generally), a fundamental task is to sample from the model.  Sampling is used to estimate statistics, observe evidence of phase transitions, and perform statistical tests on data.  A wide variety of methods have been proposed to sample from these distributions; for instance, the Markov chain Monte Carlo (MCMC) method was first proposed by Metropolis, Rosenbluth, Rosenbluth, Teller, and Teller~\cite{metropolis1953equation} to sample from the two-dimensional hard-sphere model.  Understanding sampling methods for point processes in theory and in practice is a major area of study~\cite{moller2001review,moller2003statistical,engel2013hard,isobe2016hard,li2022hard}, and advances in sampling techniques have led to advances in the understanding of the physics of these models~\cite{metropolis1953equation,alder1957phase,lowen2000fun,bernard2009event,bernard2011two,engel2013hard}.

In this paper we will be concerned with provably efficient sampling from the hard-sphere model.  Rigorous guarantees for sampling algorithms come in several different varieties.  One question is what notion of `efficient' to use; another is what guarantee we insist on for the output.  In this paper we will provide an efficient sampling algorithm under the strictest possible terms with respect to both running time and accuracy of the output: a linear-time, \textit{perfect} sampling algorithm.

For simplicity we focus on sampling from the hard-sphere model defined on finite boxes in $\R^d$. For fixed parameter values of the model, the typical number of points appearing in such a region is linear in the volume, and so any  sampling algorithm will require at least this much time.  

As for guarantees on the output, there are two main types of guarantees.  The first type is an \textit{approximate sampler}: the output of such an algorithm must be distributed within $\varepsilon$ total variation distance of the desired target distribution.  Perhaps the main approach to efficient sampling from distributions normalized by intractable normalizing constants is the MCMC method.  In this approach, one devises a Markov chain with the target distribution as the stationary distribution and runs a given number steps of the chain from a chosen starting configuration; if the number of steps is at least the $\varepsilon$-mixing time, then the final state has distribution within $\varepsilon$ total variation distance of the target~\cite{jerrum1996markov,randall2006rapidly,diaconis2009markov}.  In general, however, computing or bounding the mixing time can be a very challenging problem.

The second type of guarantee is that of a \textit{perfect sampler}~\cite{propp1996exact}.  Such an algorithm has a running time that is random, but the distribution of the output is guaranteed to be \textit{exactly} that of the target distribution.  The main advantage of perfect sampling algorithms -- and the primary reason they are studied and used in practice -- is that one need not prove a theorem or understand the mixing time of a Markov chain to run the algorithm and get an accurate sample; one can simply run the algorithm and know that the output has the correct distribution.  The drawback is that the running time may be very large, depending on the specific algorithm and on the parameter regime.  Some naive sampling methods such as rejection sampling return perfect samples but are inefficient on large instances (exponential expected running time in the volume).  The breakthrough of Propp and Wilson in introducing `coupling from the past'~\cite{propp1996exact,propp1998get} was to devise a procedure for using a Markov chain transition matrix to design perfect sampling algorithms which, under some conditions, could run in time polylogarithmic in the size of a discrete state space (polynomial-time in the size of the graph of a spin system), matching the efficiency of fast mixing Markov chains which only return approximate samples 
(see also~\cite{asmussen1992stationarity,lovasz1995exact} for precedents in perfect sampling).  The work of Propp and Wilson led to numerous constructions of  perfect sampling algorithms for problems with both discrete and continuous state spaces including~\cite{felsner1997markov,haggstrom1998exact,kendall1998perfect,murdoch1998exact,haggstrom1999characterization,ferrari2002perfect,kendall2000perfect,moller2001review,garcia2000perfect}.  Notably, many of the first applications of Propp and Wilson's technique were  in designing perfect sampling algorithms for Gibbs point processes (though often without rigorous guarantees on the efficiency of the algorithms).

Perfect sampling continues to be a very active area of research today, with a special focus on improving the range of parameters for which perfect sampling algorithms can (provably) run in expected linear or polynomial  time~\cite{bhandari2020improved,jain2021perfectly,he2021perfect}

In this paper we design a perfect sampling algorithm for the hard-sphere model (and Gibbs point processes interacting with a finite-range, repulsive pair potential more generally) that is guaranteed to run in linear expected time for activity parameters up to the best known bound for efficient approximate sampling via MCMC.  


What is this bound and how do we design the algorithm?    One central theme in the analysis of discrete spin systems is the relationship between spatial mixing (correlation decay properties) and temporal mixing (mixing times of Markov chains)~\cite{holley1985possible,aizenman1987rapid,stroock1992logarithmic,martinelli1999lectures,dyer2004mixing}.  At a high level, these works show that for discrete lattice systems a strong correlation decay property (\textit{strong spatial mixing}) implies a near-optimal convergence rate for local-update Markov chains like the Glauber dynamics.  Recently it has been showed that strong spatial mixing in a discrete lattice model also implies the existence of efficient \textit{perfect} sampling algorithms~\cite{feng2022perfect,anand2022perfect}.  In parallel, there has been work establishing the connection between strong spatial mixing and optimal temporal mixing for Markov chains in the setting of the hard-sphere model and Gibbs point processes~\cite{helmuth2020correlation,michelen2020analyticity,michelen2022strong}.  At a high level, our aim is to combine these threads to show that strong spatial mixing for Gibbs point processes implies the existence of an efficient perfect sampler.  One challenge is that the approaches of~\cite{feng2022perfect,anand2022perfect} are inherently discrete in that key steps of the algorithms involve enumerating over all possible configurations in a subregion, something that is not possible in the continuum.  To overcome this we make essential use of Bernoulli factories -- a method for perfect simulation of a coin flip with a bias $f(p)$ given access to coin flips of bias $p$.  Bernoulli factories have recently been used in perfect sampling algorithms for solutions to constraint satisfaction problems in~\cite{he2022sampling,he2022improved}.

\subsection{The hard-sphere model, strong spatial mixing, and perfect sampling}\label{sec:intro_hard-sphere}

The hard-sphere model is defined on a bounded, measurable subset $\region$ of $\R^d$ with an activity parameter $\lambda \ge 0$ that governs the density of the model and a parameter $r >0$ that governs the range of interaction (though by re-scaling there is really only one meaningful parameter, and we could take $r=1$ without loss of generality).  In words, the hard-sphere model is the distribution of finite point sets in $\region$ obtained by taking a Poisson point process of activity $\lambda $ on $\region$ and conditioning on the event that all pairs of points are at distance at least $r$ from each other; in other words, on the event that spheres of radius $r/2$ centered at the given points form a sphere packing.  

We can equivalently define the model more explicitly, and in doing so, introduce objects and notation we work with throughout the paper.  
Fix the number of dimensions $\dimensions \in \N$, and denote by $\borel$ the Borel $\sigma$-field on $\R^{\dimensions}$ and by $\boundedBorel$ all bounded sets in $\borel$.
A point process on $\R^{\dimensions}$ is a probability measure on the set of locally finite point sets $\pointsets = \{\eta \subset \R^{\dimensions} \mid \forall \region \in \boundedBorel: \size{\eta \cap \region} < \infty\}$, equipped with the $\sigma$-field $\pointsetEvents$ that is generated by the maps $\set{\pointsets \to \N_{0}, \eta \mapsto \size{\eta \cap \region}}[\region \in \borel][\big]$.
The idea behind modeling gases via point processes is to represent them as random point configurations $\eta \in \pointsets$, where each point $x \in \eta$ indicates the (random) location of a particle.  

Throughout this paper, we focus on gases that are confined in a bounded region of space.
To this end, for $\region \in \boundedBorel$, we write $\pointsets_{\region}$ for the set of point configurations $\eta \in \pointsets$ with $\eta \cap \region^{\comp} = \emptyset$, and we write $\pointsetEvents_{\region}$ for the trace of $\pointsets_{\region}$ in $\pointsetEvents$. 
Note that every configuration in $\pointsets_{\region}$ contains only finitely many points.
The hard-sphere model (or in fact any Gibbs point process) on a bounded region $\region \in \pointsets[\region]$ is a point process $\gibbs[\lambda, \region]$ that is only supported on $\pointsets[\region]$.

Define for every $x_1, \dots, x_k \in \R^{\dimensions}$  the indicator that the points are centers of non-overlapping spheres of radius $r/2$; that is,
\begin{align*}
	\validHS[x_1, \dots, x_k] &= \prod_{\set{i, j} \in \binom{[k]}{2}} \ind{\dist{x_i}{x_j} \ge \range}  \, .
 \end{align*}
 Then define the hard-sphere partition function on $\region \in \boundedBorel$ at activity $\activity \in \R_{\ge 0}$ as
\begin{equation*}
    \partitionFunction[\region](\lambda) = \sum_{k \ge 0}  
   \frac{ \lambda^k}{k!}\int_{\region^k} \validHS[x_1, \dots, x_k]   \intD x_1 \dots \intD x_k \,.
\end{equation*}
For an event $A \in \pointsetEvents$, the hard-sphere model on $\region$ with activity $\activity$ assigns the probability
\begin{equation}
\label{eqHSGibbsDefn}
    \gibbs[\activity, \region](A) = \frac{1}{\partitionFunction[\region](\lambda)} \sum_{k \ge 0} \frac{\lambda^k}{k!} \int_{\region^k} \ind{\set{x_1, \dots, x_k} \in A} \validHS[x_1, \dots, x_k] \intD x_1 \dots \intD x_k \,.
\end{equation}

A very useful generalization of this model is to allow for bounded, measurable activity functions $\activityFunction: \R^{\dimensions} \to \R_{\ge 0}$ instead just constant activities.  
Here the model is a Poisson process with inhomogenous activity $\activityFunction$ conditioned on the points forming the centers of a sphere packing; the partition function is now
\begin{equation*}
    \partitionFunction[\region](\activityFunction) = \sum_{k \ge 0}  
   \frac{ 1}{k!}\int_{\region^k} \prod_{i=1}^k \activityFunction(x_i) \validHS[x_1, \dots, x_k]   \intD x_1 \dots \intD x_k 
\end{equation*}
and the measure $\gibbs[\activityFunction, \region]$ is defined analogously to~\eqref{eqHSGibbsDefn}. 
This generalization allows modeling of inhomogenous spaces and generalizes the concept of imposing boundary conditions on the model (i.e., the effect of placing particles at fixed locations in space).  
To see the latter, suppose we fix a point configuration $\eta \in \pointsets$ as boundary condition, meaning that no points are allowed to be placed in a ball of radius $r$ around the points in $\eta$, then we can simply model this by setting $\activityFunction(x) = 0$ for every point $x \in \region$ such that $\dist{x}{y} < r$ for any $y \in \eta$. 
We defer a detailed discussion to the more general setting of repulsive point processes,
and proceed by using activity functions to define \textit{strong spatial mixing}, the condition under which we can guarantee the efficiency of our perfect sampling algorithm. 

To define the concept of strong spatial mixing, we write $\projection{\region}: \pointsets \to \pointsets$ for the projection $\eta \mapsto \eta \cap \region$ to some region $\region \in \borel$.
Moreover, for any two point processes $P, Q$ on $\R^{\dimensions}$, we write $\dtv{P}{Q}$ for their total variation distance, and we write $\dtvProjected{P}{Q}{\region} \coloneqq \dtv{P \circ \projection{\region}^{-1}}{Q \circ \projection{\region}^{-1}}$ for the total variation distance between the projections of $P$ and $Q$ to $\region$ (i.e., of their pushforward measures under $\projection{\region}$). 

Strong spatial mixing asserts that, for any bounded region $\region \in \boundedBorel$ and any suitable pair of activity functions $\activityFunction, \activityFunction'$, the distributions $\gibbs[\activityFunction, \region]$ and $\gibbs[\activityFunction', \region]$ are similar on any region $\region' \in \boundedBorel$ such that $\activityFunction$ and $\activityFunction'$ only differ far away from $\region'$; i.e., $\dtvProjected{\gibbs[\activityFunction, \region]}{\gibbs[\activityFunction', \region]}{\region'}$ vanishes as $\dist{\subregion}{\supp{\activityFunction - \activityFunction'}}$ increases. 
Writing $\volume{\subregion}$ for the volume of $\subregion$, strong spatial mixing with exponential decay is defined as follows.
\begin{definition}\label{def:ssmHS}
	Given $a, b \in \R_{>0}$, the hard-sphere model on $\R^{\dimensions}$ exhibits \textbf{$(a, b)$-strong spatial mixing} up to activity $\activity \in \R_{>0}$ if for all bounded regions $\region, \region' \in \boundedBorel$ and all activity functions $\activityFunction, \activityFunction' \le \activity$ it holds that
	\[
	   \dtvProjected{\gibbs[\activityFunction, \region]}{\gibbs[\activityFunction', \region]}{\region'} \le a \volume{\subregion} \eulerE^{-b \cdot \dist{\subregion}{\supp{\activityFunction - \activityFunction'}}} .
	\]
\end{definition}

This definition of strong spatial mixing comes from~\cite{michelen2022strong}, which in turn adapted similar notions from discrete spin systems~\cite{dyer2004mixing,Weitz}.  Strong spatial mixing has proved to be an essential definition in the analysis, both probabilistic and algorithmic, of spin systems on graphs, and many recent works are focused on either proving strong spatial mixing for a particular model, range of parameters, and class of graphs (e.g.~\cite{Weitz,gamarnik2015strong,lu2013improved,sinclair2017spatial,regts2023absence,chen2023strong})  or deriving consequences of strong spatial mixing (e.g.~\cite{spinka2020finitary,feng2018local,liu2022correlation,feng2022perfect, anand2022perfect}).

Our main result is a linear expected-time perfect sampling algorithm for the hard-sphere model under the assumption of strong spatial mixing.

\begin{theorem}
    \label{mainThmHardSphere}
    There is a perfect sampling algorithm for the hard-sphere model on finite boxes $\region \subset \R^d$ with the property that if the hard-sphere model exhibits $(a,b)$-strong spatial mixing up to $\lambda$, then the expected running time of the algorithm at activity $\lambda$ is $O(| \region|)$, where the implied constant is a function of $a,b$, and $\activity$. 
\end{theorem}

In particular, one can run the algorithm for any value of $\activity$ (without knowing whether or not strong spatial mixing holds) and the algorithm will terminate in finite time with an output distributed exactly as $\gibbs[\activity, \region]$; under the assumption of strong spatial mixing the expected running time is guaranteed to be linear in the volume.

Using bounds from~\cite{michelen2022strong} on strong spatial mixing in the hard-sphere model, we obtain the following explicit bounds on the activities for which the algorithm is efficient.

\begin{corollary}
    \label{corHardSphereBound}
    The above perfect sampling algorithm runs in expected time $O(| \region|)$ when $\lambda < \frac{\eulerE}{v_d(r)}$, where $v_d(r)$ is the volume of the ball of radius $r$ in $\R^d$.
\end{corollary}

In comparison, near-linear time MCMC-based approximate samplers  were given in~\cite{michelen2022strong} for the same range of parameters (following results for more restricted ranges in~\cite{kannan2003rapid,helmuth2020correlation}).  For perfect sampling from the hard-sphere model, linear expected time algorithms were given in~\cite{huber2012spatial,jerrum2019perfect} for more restrictive ranges of parameters.

\subsection{Gibbs point processes with finite-range repulsive potentials}\label{sec:intro_gpp}

We now give a closely related result in the  more general setting of Gibbs point processes interacting via finite-range, repulsive pair potentials.

Gibbs point processes are defined via a density against an underlying Poisson point process.  In general, this density is the exponential of (the negative of) an energy function on point sets that captures the interactions between points.  In many of the most studied cases, this energy function takes a special form: it is the sum of potentials over pairs of points in a configuration.  

A \textit{pair potential} is a measurable symmetric function $\potential: \R^d \times \R^d \to \R \cup \set{\infty}$.  For a bounded, measurable activity function $\activityFunction$ on $\region$ the Gibbs point process with pair potential $\potential$ on $\region$ is defined via the partition function 
\[
    \partitionFunction[\region](\activityFunction) = \sum_{k \ge0 } \frac{1}{k!} \int_{\region^k} \left(\prod_{i \in [k]} \activityFunction[x_i] \right) \eulerE^{- \hamiltonian[x_1, \dots, x_k]} \intD x_1 \dots \intD x_k
\]
where
\[
    \hamiltonian[x_1, \dots, x_k] = \sum_{\set{i, j} \in \binom{[k]}{2}} \potential[x_i][x_j]  \,,
\]
denotes the Hamiltonian given by the potential $\potential$.
Again the corresponding probability measure  $\gibbs[\activityFunction, \region]$ is obtained as in~\eqref{eqHSGibbsDefn}. 
Note that the hard-sphere model discussed earlier is obtained by setting $\potential(x, y) = \infty$ if $\dist{x}{y} < r$ and $\potential(x, y) = 0$ otherwise.
Making sure that the potential in question is always clear from the context, we allow ourselves to omit the dependency on the potential from the notation.
For the same reason, we use the same notation for the general model and the hard-sphere model.

A pair potential $\phi$ is \textit{repulsive} if $\phi(x,y) \ge 0$ for all $x,y$. 
It is of \textit{finite-range} if there exists $\range\ge0$ so that $\potential[x][y] =0$ whenever $\dist{x}{y} \ge \range$.  
Typical example of models interacting via a finite-range, repulsive pair potential are the hard-sphere model (as illustrated above) or the Strauss process~\cite{strauss1975model,kelly1976note}.

As with the hard-sphere model, we can use the activity function to encode the influence of boundary conditions.
To this end, given a repulsive potential $\potential$ and an activity function $\activityFunction$, we account for the impact of a boundary condition $\eta \in \pointsets$ by defining the modified activity function $\activityFunction_{\eta}: y \mapsto \activityFunction[y] \eulerE^{- \sum_{x \in \eta} \potential[x][y]}$.
We then define the partition function and Gibbs point process on $\region \in \boundedBorel$ with activity function $\activityFunction$ and boundary condition $\eta$ by $\partitionFunction[\region][\eta](\activityFunction) \coloneqq \partitionFunction[\region](\activityFunction_{\eta})$ and $\gibbs[\activityFunction, \region][\eta] \coloneqq \gibbs[\activityFunction_{\eta}, \region]$.
In the case of constant activity functions $\activityFunction \equiv \activity \in \R_{\ge 0}$, our notation simplifies to $\partitionFunction[\region][\eta](\activity)$ and $\gibbs[\activity, \region][\eta]$ respectively.

We proceed by defining strong spatial mixing for a Gibbs point process exactly as in Definition~\ref{def:ssmHS}.
Our next result is a near-linear expected time perfect sampling algorithm for Gibbs point processes interacting via finite-range, repulsive potentials under the assumption  of strong spatial mixing.
\begin{theorem}
	\label{thmMain}
	Suppose $\phi$ is a finite-range, repulsive potential on $\R^d$ and suppose $\phi$ exhibits $(a,b)$-strong spatial mixing up to $\lambda$ for some constants $a, b >0$. Then there is a perfect sampling algorithm for the Gibbs point process defined by $\phi$ and activity bounded by $\lambda$ on boxes $\region $ in $ \R^d$ with expected running time $O \left( \volume{\region} \log ^{O(1)} \volume{\region} \right)$.
\end{theorem}

One difference between this algorithm and the hard-sphere algorithm of Theorem~\ref{mainThmHardSphere} is that this algorithm needs knowledge of the constants $a,b$ in the assumption of strong spatial mixing, whereas the hard-sphere algorithm does not.  

Using the results of~\cite{michelen2022strong}, we can get explicit bounds for the existence of efficient perfect sampling algorithms in terms of the \textit{temperedness constant} of the potential defined by
\begin{equation}
	C_\phi := \sup_{x \in \R^d} \int_{\R^d} | 1- e^{-\phi(x,y)} |\, dy \,.
\end{equation}
Under the assumption that $\phi$ is repulsive and of finite range $r$, we have $0 \le C_\phi \le v_d(r)$.
\begin{corollary}
	\label{corBounds}
	The above perfect sampling algorithm runs in expected time $O \left( \volume{\region} \log ^{O(1)} \volume{\region} \right)$ when $\lambda < \frac{\eulerE}{C_{\phi}}$.
\end{corollary}
\begin{remark}
In fact, using the results of Michelen and Perkins~\cite{mp-CC}, one can push the bound for strong spatial mixing up to $e/ \Delta_\phi$, where $\Delta_\phi\le C_{\phi}$ is the \textit{potential-weighted connective constant} defined therein; our perfect sampling algorithm is efficient up to that point. 
\end{remark}

\subsection{Related work and future directions}

\subsubsection*{Related work}
In recent years there has been a moderate flurry of activity around proving rigorous results for Gibbs point processes in both the setting of statistical physics and probability theory and in the setting of provably efficient sampling algorithms.  

Work on provably efficient approximate sampling methods for the hard-sphere model begins with the seminal paper of Kannan, Mahoney, and Montenegro~\cite{kannan2003rapid}, who used techniques from the analysis of discrete spin systems to prove mixing time bounds for Markov chains for the hard-sphere model.  Improvements to the range of parameters for which fast mixing holds came in~\cite{hayes2014lower,helmuth2020correlation}, before Michelen and Perkins proved the bound $e/v_d(r)$ in~\cite{michelen2022strong}, which we match with a perfect sampling algorithm in Corollary~\ref{corHardSphereBound}. 

Perfect sampling algorithms for the hard sphere model have been considered in~\cite{haggstrom1998exact,kendall2000perfect,ferrari2002perfect,jerrum2019perfect,huber2013bounds}.  In terms of rigorous guarantees of efficiency,  Huber proved a bound of $2/v_d(r)$ for a near-linear expected time perfect sampler in~\cite{huber2012spatial}.  The perfect sampling algorithm of Guo and Jerrum in~\cite{jerrum2019perfect} does not match this bound, but the algorithm, based on `partial rejection sampling'~\cite{guo2019uniform} is novel and particularly simple.  Several of these approaches also apply for finite-range, repulsive potentials or can be extended to that setting (e.g.~\cite{moka2020perfect}).

In parallel, there has been much work on proving bounds on the range of activities for which no phase transition can occur in the hard-sphere model; and, in recent years in particular, the techniques used have close connections to algorithms and the study of Markov chains.  The classic approach to proving absence of phase transition is by proving convergence of the cluster expansion; the original bound here is $1/(e v_d(r))$ due to Groeneveld~\cite{groeneveld1962two}. In small dimensions (most significantly in dimension $2$) improvements to the radius of convergence can be obtained~\cite{fernandez2007analyticity}.  On the other hand, this approach is inherently limited by the presence of non-physical singularities on the negative real axis.  Alternative approaches avoiding this obstruction include using the equivalence of spatial and temporal mixing~\cite{helmuth2020correlation,michelen2022strong}; or disagreement percolation~\cite{hofer2015disagreement,hofer2019disagreement,betsch2021uniqueness}.  The best current bound for absence of phase transition for the hard-sphere model and for repulsive pair potentials is the bound of $e/ C_{\potential}$ (and $e/ \Delta_{\potential}$) obtained by Michelen and Perkins~\cite{michelen2020analyticity,michelen2022strong,mp-CC}.  Theorem~\ref{thmMain} brings the bound for efficient perfect sampling up to this bound.

On a technical level, the most relevant past work is~\cite{feng2022perfect}, in which the authors prove that for discrete spin systems, strong spatial mixing and subexponential volume growth of a sequence of graphs imply the existence of an efficient perfect sampling algorithm. We take their approach as a starting point but need new ideas to replace their exhaustive enumeration of configurations. 

A key step in  our algorithm is the use of a Bernoulli factory to implement a Bayes filter.  Bernoulli factories are algorithms by which a Bernoulli random variable with success probability $f(p)$ can be simulated (perfectly) by an algorithm with access to independent Bernoulli $p$ random variables, where the algorithm does not know the value $p$.  Whether a Bernoulli factory exists (and how efficient it can be) depends on the function $f(\cdot)$ and a priori bounds on the possible values $p$.  Bernoulli factories have been studied in~\cite{nacu2005fast,MR3506427,dughmi2021bernoulli} and recently used in the design of perfect sampling algorithms for CSP solutions in~\cite{he2022sampling, he2022improved}.

\subsubsection*{Future directions}
There are a number of extensions and improvements to these results one could pursue.  
 Perhaps most straightforward would be to relax the notion of strong spatial mixing from exponential decay to decay faster than the volume growth of $\R^d$ and to extend the results to repulsive potentials of unbounded range but finite temperedness constant $C_{\potential}$.  Moreover, it would be nice to upgrade the guarantees of the  algorithm in Theorem~\ref{thmMain} to that of Theorem~\ref{mainThmHardSphere}: that the algorithm does not need prior knowledge of the strong spatial mixing constants $a,b$ to run correctly.  

An ambitious and exciting direction would be to remove the assumption of a repulsive potential and find efficient perfect sampling algorithms for the class of \textit{stable} potentials (see e.g.~\cite{penrose1963convergence,ruelle1963correlation,ruelle1999statistical} for a definition).  A stable potential is repulsive at short ranges but can include a weak attractive part; such potentials include the physically realistic Lenard-Jones potential among others~\cite{wood1958recent}. 
This would require some very new ideas, as much of the recent probabilistic and algorithmic work on Gibbs point processes (e.g.~\cite{michelen2020analyticity,michelen2022strong,betsch2021uniqueness,mp-CC}) has used repulsiveness as an essential ingredient (for one, repulsiveness of the potential implies stochastic domination by the underlying Poisson point process). 
As a notable exception, a deterministic approximation algorithm for partition functions of finite-range stable potentials based on cluster expansion was recently proposed in~\cite{jenssen2022quasipolynomial}. 

\subsection{Outline of the paper}

In \Cref{secIntution}, we describe the high-level idea and intuition behind the algorithm.  In \Cref{secPrelim} we introduce  some notation and present some preliminary results that we will use throughout the paper.  In \Cref{secTheAlgorithm} we present the algorithm that we will apply to both hard spheres and more general processes.  In \Cref{secCorrectness} we prove correctness of the algorithm. In \Cref{sec:point_densities} we prove a technical lemma that will be crucial for showing efficiency of our algorithm under the assumption of strong spatial mixing. In \Cref{secHardSphere} we specialize to the hard-sphere model to complete the proof of \Cref{mainThmHardSphere}.  In \Cref{secRepulsive} we work with finite-range, repulsive potentials to complete the proof of \Cref{thmMain}. In \Cref{sec:BernoulliFactories} we prove the running time bound for the Bernoulli factory used by our algorithm.  The appendix contains some technical lemmas on measure theory and stochastic processes.

\section{Intuitive idea behind the algorithm}
\label{secIntution}

Our algorithm is an adaptation of the work by Feng, Guo, and Yin \cite{feng2022perfect} on perfect sampling from discrete spin systems to continuum models. 
We mimic their setting of a spin system on a graph $G = (V,E)$ by considering a graphical structure on sub-regions of our continuous space.


Let $\region = [0, \sidelength)^{\dimensions} \subset \R^{\dimensions}$ be the region considered,  $\activity>0$ the activity, and let  $\potential$ be a repulsive potential of range $\range > 0$. 
The main idea is to subdivide $\region$ into boxes of (roughly) side length $\range$, indexed by $\boxIds = \set{0, \dots, \numberBoxes-1}^{\dimensions}$ for $\numberBoxes = \left\lceil \sidelength/\range \right\rceil$.
Each box index $\vectorize{v} = (v_1, \dots, v_{\dimensions}) \in \boxIds$ is associated with the sub-region 
\[
    \boxRegion{\vectorize{v}} = \big([v_1 \range, (v_1+1) \range) \times \dots \times [v_{\dimensions} \range, (v_{\dimensions}+1) \range)\big) \cap \region.
\]
We extend this notation to sets of box indices $S \subseteq \boxIds$ by setting $\boxRegion{S} = \bigcup_{\vectorize{v} \in S} \boxRegion{\vectorize{v}}$.
Further, for $\vectorize{v} \in \boxIds$, we write $\ball{\vectorize{v}}{k}$ for the set of boxes $\vectorize{w} \in \boxIds$ with $\norm{\vectorize{v} - \vectorize{w}}{\infty} \le k$, and we denote by $\partial S = (\bigcup_{\vectorize{v} \in S} \ball{\vectorize{v}}{1}) \setminus S$ for the outer boundary of $S \subseteq \boxIds$.
For readers familiar with the work of Feng, Guo, and Yin \cite{feng2022perfect} on discrete spin systems, it will be helpful to think of $\boxIds$ as the vertices of a graph, where two vertices $\vectorize{v}, \vectorize{w} \in \boxIds$ are adjacent if $\norm{\vectorize{v} - \vectorize{w}}{\infty} = 1$. 

Often, it will be convenient to not differentiate between a set of box indices $S \subseteq \boxIds$ and the associated region $\boxRegion{S}$.
More precisely, we write $\pointsets_{S}$ for the point sets in $\pointsets_{\boxRegion{S}}$, we denote by $\partitionFunction[S](\activity)$ and $\gibbs[\activity, S]$ the partition function and the Gibbs point process on $\boxRegion{S}$, and, for a point set $\eta \in \pointsets$, we write $\eta \cap S$ and $\eta \setminus S$ for $\eta \cap \boxRegion{S}$ and $\eta \setminus \boxRegion{S}$.
A notable exception from this abuse of notation is that we always write $\size{S}$ for the cardinality of the set $S$ and $\volume{\boxRegion{S}}$ for the total volume of the boxes indicated by $S$.
Moreover, the set complement $S^{\comp}$ should be understood as $\boxIds \setminus S$, which is then associated with the region $\region \setminus \boxRegion{S}$ (opposed to $\R^{\dimensions} \setminus \boxRegion{S}$).

Our algorithm runs iteratively, keeping track of two random variables: a point configuration $\currentPointset{t} \in \pointsets[\region]$ with $\currentPointset{0} = \emptyset,$ and a set of `incorrect' boxes $\boxesToFix{t} \subseteq \boxIds$ with $\boxesToFix{0} = \boxIds$. 
With each iteration $t$ we maintain the following \emph{invariant}: the partial configuration $\currentPointset{t} \cap \boxesToFix{t}^{\comp}$ is distributed according to the projection of $\gibbs[\activity, \boxesToFix{t}^{\comp}][\currentPointset{t} \cap \boxesToFix{t}]$ (i.e., the Gibbs point process on $\region \setminus \boxRegion{\boxesToFix{t}}$ with boundary condition $\currentPointset{t} \cap \boxesToFix{t}$).
It follows that $\currentPointset{t}$ is distributed according to $\gibbs[\activity, \region]$ once we reach the state $\boxesToFix{t} = \emptyset$.

In every iteration, the algorithm tries to update the point configuration $\currentPointset{t}$ on a subset of boxes $\boxesToUpdate \subseteq \boxIds$.
To this end, given an \emph{update radius} $\updateRadius \in \N$, we define for every $S \subseteq \boxIds$ and $\vectorize{v} \in S$ the set of box indices
\begin{align}
	\boxesToUpdate[S][\vectorize{v}][\updateRadius] \coloneqq \left(\ball{\vectorize{v}}{\updateRadius} \setminus  S\right) \cup \vectorize{v}  .
	\label{eq:boxesToUpdate}
\end{align}
We proceed by sketching an iteration of the algorithm. 
An example for the involved subregions is given in \Cref{fig:regions}.
Each iteration runs as follows:
\begin{enumerate}[1.]
	\item We choose $\boxChosen{t} \in \boxesToFix{t}$ uniformly at random and attempt to `repair' it by updating $\currentPointset{t}$ on a neighborhood of boxes $\boxesToUpdate = \boxesToUpdate[\boxesToFix{t}][\boxChosen{t}][\updateRadius]$ as given in \eqref{eq:boxesToUpdate}.
	\item We sample a Bayes filter $\bayesFilter{t}$ (i.e., a Bernoulli random variable) with probability depending on the potential $\potential$, the activity $\activity$, and the current point configuration $\currentPointset{t}$ on $\boxRegion{\boxChosen{t}}$ and $\boxRegion{\partial \boxesToUpdate}$.
	\item \begin{enumerate}[a)]
		\item If $\bayesFilter{t} = 1$, we set $\boxesToFix{t+1} = \boxesToFix{t} \setminus \boxChosen{t}$, and we obtain $\currentPointset{t+1}$ by updating $\currentPointset{t}$ on $\boxRegion{\boxesToUpdate}$ according to a sample from $\gibbs[\activity, \boxesToUpdate][\currentPointset{t} \cap \boxesToUpdate^{\comp}]$ (i.e., the Gibbs point process on $\boxRegion{\boxesToUpdate}$ with boundary condition $\currentPointset{t} \cap \boxesToUpdate^{\comp}$). 
		\item If $\bayesFilter{t} = 0$, the configuration is unchanged and we add the boundary boxes to our `incorrect' list, i.e., $\currentPointset{t+1} = \currentPointset{t}$ and $\boxesToFix{t+1} = \boxesToFix{t} \cup \partial \boxesToUpdate$. 
	\end{enumerate}
\end{enumerate}

\sidecaptionvpos{figure}{c}
\begin{SCfigure}[50][h]
	\centering
	\includegraphics[width=0.35\textwidth]{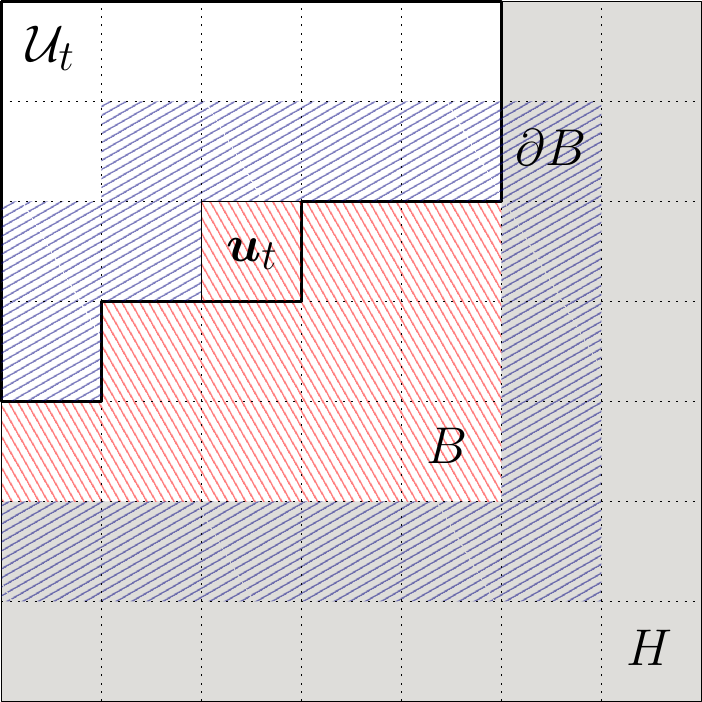}
	\caption{{\small The box-shaped region $\region \subset \R^{2}$ is divided into boxes of side length $\range$ (dotted lines). The boxes $\boxesToFix{t}$ are bordered by bold black lines. For $\boxChosen{t}$ as given and update radius $\updateRadius=2$, the corresponding set~$\boxesToUpdate$ of boxes to be updated  is indicated by the red hatched area (falling left to right). Its boundary boxes $\partial \boxesToUpdate$ are shown as blue hatched area (rising left to right). The boxes in $H = \complementOf{\boxesToFix{t} \cup \boxesToUpdate}$ are shown with gray background.}}
	\label{fig:regions}
\end{SCfigure}

We use the Bayes filter, as in~\cite{feng2022perfect}, to remove bias from the resulting distribution. To give some intuition for its role, suppose we run a naive version of the algorithm where we always update $\currentPointset{t}$ on $\boxRegion{\boxesToUpdate}$ as in step 3.a) above.
Assuming the desired invariant holds after $t$ iterations, this naive algorithm gives a bias to the distribution of $\currentPointset{t+1}$ proportional to $
	\frac{ \partitionFunction[\boxesToUpdate \setminus \boxChosen{t}][\currentPointset{t} \cap (\partial \boxesToUpdate \cup \boxChosen{t})](\activity)}{\partitionFunction[\boxesToUpdate][\currentPointset{t} \cap \partial \boxesToUpdate](\activity)}
$.
We choose the Bayes filter such that, conditioned on $\bayesFilter{t} = 1$, the bias term gets canceled.
This suggests the choice
\begin{align} \label{eq:filter_probability}
	\Pr{\bayesFilter{t} = 1}[\currentPointset{t}, \boxesToFix{t}, \boxChosen{t}] = \bayesFilterCorrection[\boxesToFix{t}][\boxChosen{t}][\currentPointset{t}] \cdot
	\frac{
		\partitionFunction[\boxesToUpdate][\currentPointset{t} \cap \partial \boxesToUpdate](\activity) 
	}{ 
		\partitionFunction[\boxesToUpdate \setminus \boxChosen{t}][\currentPointset{t} \cap (\partial \boxesToUpdate \cup \boxChosen{t})](\activity)
	},
\end{align}
where scaling $\bayesFilterCorrection[\boxesToFix{t}][\boxChosen{t}][\currentPointset{t}]$   serves three main purposes.

First, it must guarantee that the right-hand side of \eqref{eq:filter_probability} is a probability.
To achieve this we need, for $H = \complementOf{\boxesToFix{t} \cup \boxesToUpdate}$ and almost all realizations of $\currentPointset{t}$, $\boxesToFix{t}$ and $\boxChosen{t}$, that
\begin{align} \label{eq:infimum_bound}
	\bayesFilterCorrection[\boxesToFix{t}][\boxChosen{t}][\currentPointset{t}] 
	\le \inf_{\substack{\xi \in \pointsets[H]}}  \frac{\partitionFunction[\boxesToUpdate \setminus \boxChosen{t}][\xi \cup (\currentPointset{t} \cap \boxesToFix{t})](\activity)}{\partitionFunction[\boxesToUpdate][\xi \cup (\currentPointset{t} \cap (\boxesToFix{t} \setminus \boxChosen{t}))](\activity)}.
\end{align}
Second, $\bayesFilterCorrection[\boxesToFix{t}][\boxChosen{t}][\currentPointset{t}]$ must introduce no new bias.
Carrying out the necessary calculations, it can be shown that this is guaranteed if $\bayesFilterCorrection[\boxesToFix{t}][\boxChosen{t}][\currentPointset{t}]$ only depends on $\currentPointset{t} \cap \boxesToFix{t}$. 
Finally, it must ensure that the algorithm terminates almost surely.
It suffices to ensure $\bayesFilterCorrection[\boxesToFix{t}][\boxChosen{t}][\currentPointset{t}]$ is uniformly bounded away from $0$ for almost all realizations of $\currentPointset{t}$,
implying that the same holds for the right-hand side of \eqref{eq:filter_probability}.
We refer to  a function $\bayesFilterCorrection(\cdot)$ satisfying these requirements as a \textit{Bayes filter correction}. 

If we use a Bayes filter as given in \eqref{eq:filter_probability}, keeping $\currentPointset{t}$ and $\boxesToFix{t}$ unchanged whenever $\bayesFilter{t} = 0$ introduces new bias.
To prevent this, we set $\boxesToFix{t+1} = \boxesToFix{t} \cup \partial \boxesToUpdate$ in step 3.b), effectively deleting the part of the configuration that was revealed by the filter.
Since the algorithm only terminates once $\boxesToFix{t} = \emptyset$, we further require the Bayes filter correction to ensure that the probability of $\bayesFilter{t} = 0$ is small to guarantee efficiency. 

Constructing a Bayes filter correction that satisfies the requirements above and allows for efficient sampling of $\bayesFilter{t}$ is a non-trivial task.
In the next subsections, we present two approaches for this, the first specialized to the hard-sphere model without requirements, and the second one for more general potentials under the assumption of strong spatial mixing. 
Crucially, assuming strong spatial mixing, both constructions allow us to control the success probability of the Bayes filter via the update radius $\updateRadius$ in the construction of the updated set of boxes $\boxesToUpdate$ (see step 1 and \eqref{eq:boxesToUpdate}).

\subsection{Bayes filter for the hard-sphere model}
To construct a Bayes filter for the hard-sphere model, we efficiently approximate the right-hand side of \eqref{eq:infimum_bound}.
To approximate the infimum over the uncountable set of configurations $\xi \in \pointsets[H]$ we take the minimum over a finite, but sufficiently rich set of configurations, balancing the quality of approximation with the computation required. In fact the number of configurations needed will depend only on the volume of $\boxRegion{\boxesToUpdate \cup \partial \boxesToUpdate}$.
We approximate the fraction of partition functions in \eqref{eq:infimum_bound} with running time only depending on the volume of $\boxRegion{\boxesToUpdate \cup \partial \boxesToUpdate}$.
As a result, we efficiently compute a Bayes filter correction $\bayesFilterCorrectionHS{\varepsilon}(\cdot)$, with the parameter $\varepsilon > 0$ controlling how much $\bayesFilterCorrectionHS{\varepsilon}[\boxesToFix{t}][\boxChosen{t}][\currentPointset{t}]$ deviates from the right-hand side of \eqref{eq:infimum_bound}.

While our construction of $\bayesFilterCorrectionHS{\varepsilon}(\cdot)$ guarantees correctness of the sampling algorithm for any $\varepsilon > 0$, proving efficiency requires additional assumptions.
To this end, we show that strong spatial mixing allows us to choose $\varepsilon$ so that the probability that $\bayesFilter{t} = 0$ is uniformly bounded above, ensuring $\bigO{\volume{\region}}$ iterations of the algorithm in expectation.

It remains to argue that we can efficiently sample $\bayesFilter{t}$, using the Bayes filter correction $\bayesFilterCorrectionHS{\varepsilon}(\cdot)$.
Explicitly computing the success probability of $\bayesFilter{t}$ as in \eqref{eq:filter_probability} would require computing the fraction of partition functions on the right-hand side exactly, while approximating these partition functions would require that the approximation error only depends on $\currentPointset{t} \cap \boxesToFix{t}$, to avoid new bias.

It is unclear how to implement these approaches, so instead we use Bernoulli factories to sample $\bayesFilter{t}$ without knowing the success probability.
To do so, we observe that the fraction of partition functions can be written as a ratio of probabilities for drawing the empty set from a conditional hard-sphere model on $\boxRegion{\boxesToUpdate}$ and $\boxRegion{\boxesToUpdate \setminus \boxChosen{t}}$.
Since both regions have constant volume, rejection sampling gives Bernoulli random variables with these success probabilities in constant time.
Hence, we obtain a Bernoulli factory for $\bayesFilter{t}$ with constant expected running time.
Wald's identity  yields a total expected running time $\bigO{\volume{\region}}$ for the algorithm.

\subsection{Bayes filter for general potentials} 
We now consider the case of general bounded-range, repulsive potentials.
Unlike the hard sphere model, it is not clear here how to approximate the infimum in \eqref{eq:infimum_bound} from a finite set of boundary configurations.
However, given constants $a, b > 0$ such that $\potential$ satisfies $(a, b)$-strong spatial mixing, we can explicitly compute a function $\correction[a][b]$ so that 
\[
	\bayesFilterCorrectionR{a}{b}[\boxesToFix{t}][\boxChosen{t}][\currentPointset{t}] = \correction[a][b] \cdot \frac{\partitionFunction[\boxesToUpdate \setminus \boxChosen{t}][\currentPointset{t} \cap \boxesToFix{t}](\activity)}{\partitionFunction[\boxesToUpdate][\currentPointset{t} \cap (\boxesToFix{t} \setminus \boxChosen{t})](\activity)}
\] 
is a Bayes filter correction.
With strong spatial mixing, we use $\bayesFilterCorrectionR{a}{b}(\cdot)$ to construct a Bayes filter such that probability that $\bayesFilter{t} = 0$ is bounded above, again implying a bound of $\bigO{\volume{\region}}$ on the expected number of iterations of the algorithm.

Note that in this setting, we require spatial mixing for both correctness and efficiency, while for the hard-sphere model we only need it for efficiency.
Another crucial difference is that, while we can explicitly compute $\correction[a][b]$, the same does not hold for $\bayesFilterCorrectionR{a}{b}(\cdot)$ due to the fraction of partition functions involved.
Again we circumvent this by rewriting the success probability of the Bayes filter in a suitable way and applying a Bernoulli factory for sampling $\bayesFilter{t}$.
Finally, we  point out that we do not obtain a constant bound for the expected running time of each iteration, but instead the bound depends on the number of points in $\currentPointset{t} \cap \partial \boxesToUpdate$.
Possible dependencies between the configuration $\currentPointset{t}$ and the number of iterations prevent us from bounding the total expected running time using Wald's identity.
Instead, we provide tail bounds on the number of iterations and the running time of each iteration, allowing us to derive an expected total running time that is linear in the volume of $\region$ up to polylogarithmic factors.

\section{Preliminaries}
\label{secPrelim}

Throughout the paper, we write $\N$ for the set of strictly positive integers, and we write $\N_{0} = \N \cup \set{0}$.
For any $k \in \N$, we denote by $[k]$ the set $[1, k] \cap \N$.

For a point configuration $\eta \in \pointsets$, we write $\size{\eta} \in \N_{0} \cup \{\infty\}$ for the number of points in $\eta$.
Note that this notation is the same that as the one we use for the volume of a region. The particular meaning will be clear from the context.
Moreover, for $k \in \N$, we write $\binom{\eta}{k}$ for the set $\set{\eta' \subseteq \eta}[\size{\eta'} = k]$.

\subsection{Gibbs point processes}
Throughout the paper, we use the definitions and notation for point sets and Gibbs point processes introduced in \Cref{sec:intro_hard-sphere,sec:intro_gpp}.
However, we will allow ourselves some notational shortcuts. 
Firstly, when dealing with a tuple $(x_1, \dots, x_k) \in (\R^{\dimensions})^{k}$ we frequently denote it by the corresponding bold letter $\vectorize{x}$.
Based on this, we write $\intD \vectorize{x}$ for $\intD x_1 \dots \intD x_k$ and $\hamiltonian[\vectorize{x}]$ for $\hamiltonian[x_1, \dots, x_k]$. 
Moreover, for any $k \in \N_{0}$ and $\vectorize{x} = (x_1, \dots, x_k) \in (\R^{\dimensions})^{k}$ we write $\setFromTuple{\vectorize{x}}$ for the set $\set{x_1, \dots, x_k}$, where the case $k = 0$ results in $\setFromTuple{\vectorize{x}} = \emptyset$.
Finally, for $\vectorize{x} \in \region^{k}$ we write $\activityFunction^{\vectorize{x}}$ for $\prod_{i \in [k]} \activityFunction[x_i]$.
This allows us to write the partition function on $\region \in \boundedBorel$ with activity function $\activityFunction$ as
\[
    \partitionFunction[\region](\activityFunction) = \sum_{k \ge 0} \frac{1}{k!} \int_{\region^k} \activityFunction^{\vectorize{x}} \eulerE^{- \hamiltonian[\vectorize{x}]} \intD \vectorize{x}
\]
and the probability of $A \in \pointsetEvents$ under $\gibbs[\activityFunction, \region]$ as
\[
	\gibbs[\activityFunction, \region](A) = \frac{1}{\partitionFunction[\region](\activityFunction)} \sum_{k \ge 0} \frac{1}{k!} \int_{\region^k} \ind{\setFromTuple{\vectorize{x}} \in A} \activityFunction^{\vectorize{x}} \eulerE^{- \hamiltonian[\vectorize{x}]} \intD \vectorize{x}.
\]
As discussed in the introduction, we express the impact of a boundary conditions $\eta \in \pointsets$ by considering the modified activity function $\activityFunction_{\eta}: y \mapsto \activityFunction(y) \eulerE^{-\sum_{x \in \eta} \potential(x, y)}$, and we write $\partitionFunction[\region][\eta](\activityFunction) \coloneqq \partitionFunction[\region](\activityFunction_{\eta})$ and $\gibbs[\activityFunction, \region][\eta] \coloneqq \gibbs[\activityFunction_{\eta}, \region]$ for the respective partition function and Gibbs point process with boundary condition $\eta$.
Moreover, if the activity function is constant $\activityFunction \equiv \activity \in \R_{\ge 0}$, our notation simplifies to $\partitionFunction[\region][\eta](\activity)$ and $\gibbs[\activity, \region][\eta]$ respectively.
Finally,, if the $\activity$ is clear from the context, we omit it and write $\partitionFunction[\region][\eta]$ and $\gibbs[\region][\eta]$.

We introduce further concepts related to Gibbs point processes, such a point density functions when they are required.
Moreover, various useful properties of Gibbs point processes are given in \Cref{appendix:gpps}.

\subsection{Bernoulli factories}
In designing our  sampling algorithm, it will be useful to consider the following Bernoulli factory problem. We are given access to a sampler for $\Ber{p}$ and for $\Ber{q}$, that is samplers of Bernoulli random variables with parameters $p$ and $q$ respectively, where we further assume $p < q$. We want to sample a random variable $Z \sim \Ber{\frac{p}{q}}$. 

Most work on Bernoulli factories studies their running time in terms of the number of coin flips required. In our setting, the time needed to generate each of these coin flips is random variable. Fortunately, suitable independence assumptions hold in our setting allowing us to prove the following lemma.

\begin{restatable}{lemma}{bernoulliFrac} \label{lemma:bernoulli_frac}
	Fix some $p, q \in [0, 1]$ such that $q - p \ge \epsilon$ for some $\epsilon > 0$. Further assume that we have oracle access to a sampler from $\Ber{p}$ and $\Ber{q}$ in the following sense:
	\begin{enumerate}
		\item every sample  from $\Ber{p}$ (resp. $\Ber{q}$) is independent from all previous samples;
		\item the expected running time for obtaining a sample from $\Ber{p}$ (resp. $\Ber{q}$), conditioned on previously obtained samples, is uniformly bounded by some $t \in \R_{\ge 0}$. 
	\end{enumerate}
	Then we can sample from $\Ber{\frac{p}{q}}$ in $\bigO{t\epsilon^{-2}}$ expected time.
\end{restatable}

\Cref{lemma:bernoulli_frac}, will play a key role in bounding the expected running time of our algorithm. To proceed with the formal description of our algorithm, we defer the proof of this lemma to \Cref{sec:BernoulliFactories}.

\section{The algorithm}
\label{secTheAlgorithm}
Let $\region = [0, \sidelength)^{\dimensions}$ and consider a Gibbs point processes on $\region$ with uniform activity $\activityFunction(x) \equiv \activity$ for some $\activity \in \R_{> 0}$
and repulsive potential $\potential$ with finite range $\range \in \R_{>0}$.
Throughout the analysis of our algorithm, it will be useful to focus on configurations $\eta \in \pointsets[\region]$ such that $\potential[x][y] < \infty$ for all ${\set{x, y} \in \binom{\eta}{2}}$, in which case we call $\eta$ a \emph{feasible configuration}.

We use the method of splitting $\region$ into smaller boxes that we introduced in \Cref{secIntution}, along with the same definitions and notation.
As discussed earlier, our algorithm runs in multiple iterations, and the update steps in every iteration $t$ depends on the outcome of a Bernoulli random variable $\bayesFilter{t}$, called the \emph{Bayes filter}.
The construction of this Bayes filter is closely tied to the following  definition.
\begin{definition} \label{def:bayes_filter}
	Fix a repulsive potential $\potential$ of range $\range \in \R_{>0}$, an activity $\activity \in \R_{>0}$ and some $\updateRadius \in \N$.
	We call a function $\bayesFilterCorrection: \powerset{\boxIds} \times \boxIds \times \pointsets \to [0, 1]$ a \emph{Bayes filter correction} if, 
	for all non-empty $S \subseteq \boxIds$ and $\vectorize{v} \in S$, it holds that
	\begin{enumerate}
		\item The map $\bayesFilterCorrection[S][\vectorize{v}][\cdot]$ is $\pointsetEvents$-measurable and satisfies $\bayesFilterCorrection[S][\vectorize{v}][\eta] = \bayesFilterCorrection[S][\vectorize{v}][\eta \cap S]$ for all $\eta \in \pointsets$,
		\item there is some $\varepsilon > 0$ such that for $\boxesToUpdate = \boxesToUpdate[S][\vectorize{v}][\updateRadius]$, $H = (S \cup \boxesToUpdate)^{\comp}$ and all feasible $\eta \in \pointsets_{\region}$ it holds that
		\[
			\varepsilon \le  \bayesFilterCorrection[S][\vectorize{v}][\eta] 
			\le \inf_{\substack{\xi \in \pointsets[H]\\ \xi \cup (\eta \cap S) \text{ is feasible}}}  \left\{\frac{\partitionFunction[\boxesToUpdate \setminus \vectorize{v}][\xi \cup (\eta \cap S)]}{\partitionFunction[\boxesToUpdate][\xi \cup (\eta \cap (S \setminus \vectorize{v}))]}\right\} .
		\]
	\end{enumerate} 
    Note that by 1. it holds that  $\bayesFilterCorrection[S][\vectorize{v}][\cdot]$ is fully characterized by its behavior on $\pointsets_{\region}$.
\end{definition}

Our perfect sampling procedure is stated in \Cref{algo:sampling}.
\begin{algorithm}[h!]
	\KwData{region $\region = [0, \sidelength)^{\dimensions}$, repulsive potential $\potential$ of range at most $\range \in \R_{>0}$, activity $\activity \in \R_{>0}$, update radius $\updateRadius \in \N$}
	set  $t = 0$, $\boxesToFix{t} = \boxIds$, $\currentPointset{t} = \emptyset$\\
	\While{$\boxesToFix{t} \neq \emptyset$}{
		draw $\boxChosen{t} \in \boxesToFix{t}$ uniformly at random\\
		set $\boxesToUpdate = \boxesToUpdate[\boxesToFix{t}][\boxChosen{t}][\updateRadius]$ as defined in \eqref{eq:boxesToUpdate}\\
		draw $\bayesFilter{t}$ from $\Ber{
		\bayesFilterCorrection[\boxesToFix{t}][\boxChosen{t}][\currentPointset{t}] \cdot
		\frac{
			\partitionFunction[\boxesToUpdate][\currentPointset{t} \cap \partial \boxesToUpdate] 
		}{ 
		      \partitionFunction[\boxesToUpdate \setminus \boxChosen{t}][\currentPointset{t} \cap (\partial \boxesToUpdate \cup \boxChosen{t})]
		}}$ where $\bayesFilterCorrection$ is a Bayes filter correction as in \Cref{def:bayes_filter}\\ \label{algo:sampling:filter}
		\eIf{$\bayesFilter{t}$ = 1}{
			draw $\resampledPointset$ from $\gibbs[\boxesToUpdate][\currentPointset{t} \cap \boxesToUpdate^{\comp}]$\\ \label{algo:sampling:update}
			set $\currentPointset{t+1} = \left(\currentPointset{t} \setminus \boxesToUpdate\right) \cup \resampledPointset$\\
			set $\boxesToFix{t+1} = \boxesToFix{t} \setminus \boxChosen{t}$\\
		}{
			
			set $\boxesToFix{t+1} = \boxesToFix{t} \cup \partial \boxesToUpdate$\\
		}
		increase $t$ by $1$\\
	}
	\Return $\currentPointset{t}$
	\caption{Perfect sampling algorithm for repulsive Gibbs point processes}
	\label{algo:sampling}
\end{algorithm}

To analyze \Cref{algo:sampling}, it will help to think of it as a Markov chain $(\currentPointset{t}, \boxesToFix{t}, \bayesFilter{t}, \boxChosen{t})_{t \in \N_{0}}$, which we set to remain constant once it hits a state with $\boxesToFix{t} = \emptyset$.
We write $\statespace = (\pointsets \times 2^{\boxIds} \times \set{0, 1} \times \boxIds)^{\N_{0}}$ for the state space of all trajectories of that Markov chain, which we equip the $\sigma$-field $\sigmafield =(\pointsetEvents \otimes 2^{2^{\boxIds}} \otimes 2^{\set{0, 1}} \otimes 2^{\boxIds})^{\otimes\N_{0}}$, and we denote by $\PrSymbol$ the distribution on $(\statespace, \sigmafield)$ induced by \Cref{algo:sampling}.
Note that in particular $\Pr{\currentPointset{0} = \emptyset, \boxesToFix{0} = \boxIds} = 1$.

Before we analyze the correctness and running time of \Cref{algo:sampling}, we first argue that each update step is well-defined.
In particular, we need to show that the success probability of the Bayes filter that we require in line~\ref{algo:sampling:filter} is indeed a probability.
Moreover, we convince ourselves that the algorithm terminates almost surely after finitely many iterations.
For this, we use the following lemma.
\begin{lemma} \label{lemma:bayes_filter_bounds}
	Suppose $\bayesFilterCorrection$ is a Bayes filter correction.
	Let $S \subseteq \boxIds$ be non-empty, $\vectorize{v} \in S$ and $\boxesToUpdate = \boxesToUpdate[S][\vectorize{v}][\updateRadius]$.
	There is some $\varepsilon > 0$ such that, for all feasible $\eta \in \pointsets[\region]$, it holds that
	\[
		\varepsilon \le
		\bayesFilterCorrection[S][\vectorize{v}][\eta] \cdot
		\frac{
			\partitionFunction[\boxesToUpdate][\eta \cap \partial \boxesToUpdate] 
		}{ 
		      \partitionFunction[\boxesToUpdate \setminus \vectorize{v}][\eta \cap (\partial \boxesToUpdate \cup \vectorize{v})]
		}
		\le 1 . \qedhere
	\]
\end{lemma}
\begin{proof} 
	Fix $S$ and $\vectorize{v}$.
	For the lower bound, note that for all $\eta \in \pointsets[\region]$
	\[
	 	\partitionFunction[\boxesToUpdate][\eta \cap \partial \boxesToUpdate]  
	 	\ge 
        \partitionFunction[\boxesToUpdate \setminus \vectorize{v}][\eta \cap \partial \boxesToUpdate]
	 	\ge \partitionFunction[\boxesToUpdate \setminus \vectorize{v}][\eta \cap (\partial \boxesToUpdate \cup \vectorize{v})].
	\]
	Thus, by the definition of a Bayes filter correction, there is some $\varepsilon > 0$ such that for all feasible $\eta \in \pointsets[\region]$
	\[
		\bayesFilterCorrection[S][\vectorize{v}][\eta] \cdot
		\frac{
			\partitionFunction[\boxesToUpdate][\eta \cap \partial \boxesToUpdate] 
		}{ 
			\partitionFunction[\boxesToUpdate \setminus \vectorize{v}][\eta \cap (\partial \boxesToUpdate \cup \vectorize{v})]
		}
		\ge \bayesFilterCorrection[S][\vectorize{v}][\eta]
		\ge \varepsilon.
	\]
	To derive the upper bound, note that by \Cref{lemma:spatial_markov}  	$\partitionFunction[\boxesToUpdate][\eta \cap \partial \boxesToUpdate] = \partitionFunction[\boxesToUpdate][\eta \cap \boxesToUpdate^{\comp}]$ and $\partitionFunction[\boxesToUpdate \setminus \vectorize{v}][\eta \cap (\partial \boxesToUpdate \cup \vectorize{v})] = \partitionFunction[\boxesToUpdate \setminus \vectorize{v}][\eta \cap (\boxesToUpdate \setminus \vectorize{v})^{\comp}]$.
	Next, set $H = (S \cup \boxesToUpdate)^{\comp}$ and note that, for feasible $\eta \in \pointsets[\region]$, it holds that $(\eta \cap H) \cup (\eta \cap S)$ is feasible as well.
	By the definition of a Bayes filter correction, this implies
	\[
		\bayesFilterCorrection[S][\vectorize{v}][\eta] \le \frac{\partitionFunction[\boxesToUpdate \setminus \vectorize{v}][\eta \cap (\boxesToUpdate \setminus \vectorize{v})^{\comp}]}{\partitionFunction[\boxesToUpdate][\eta \cap \boxesToUpdate^{\comp}]} .
	\]
	Consequently, it holds that
	\[
		\bayesFilterCorrection[S][\vectorize{v}][\eta] \cdot
		\frac{
			\partitionFunction[\boxesToUpdate][\eta \cap \partial \boxesToUpdate] 
		}{ 
			\partitionFunction[\boxesToUpdate \setminus \vectorize{v}][\eta \cap (\partial \boxesToUpdate \cup \vectorize{v})]
		}
		= \bayesFilterCorrection[S][\vectorize{v}][\eta] \cdot
		\frac{
			\partitionFunction[\boxesToUpdate][\eta \cap \boxesToUpdate^{\comp}]
		}{
			\partitionFunction[\boxesToUpdate \setminus \vectorize{v}][\eta \cap (\boxesToUpdate \setminus \vectorize{v})^{\comp}]
		}
		\le 1,
	\]
	which proves the claim
\end{proof}

We use the previous lemma to derive the following statement, which will guarantee that \Cref{algo:sampling} is well-defined and terminates almost surely.
\begin{lemma} \label{lemma:feasible_bounded}	
 The following holds throughout \Cref{algo:sampling}:
	\begin{enumerate}[1)]
		\item For every $t \in \N_0$, $\currentPointset{t}$ is almost surely feasible.
		\label{lemma:feasible_bounded:1}
		\item There is some $\varepsilon > 0$ such that for all $t \in \N_0$ it holds that almost surely $\boxesToFix{t} = \emptyset$ or 
		\[
			\varepsilon \le
			\bayesFilterCorrection[\boxesToFix{t}][\boxChosen{t}][\currentPointset{t}] \cdot
			\frac{
				\partitionFunction[\boxesToUpdate][\currentPointset{t} \cap \partial \boxesToUpdate] 
			}{ 
				\partitionFunction[\boxesToUpdate \setminus \boxChosen{t}][\currentPointset{t} \cap (\partial \boxesToUpdate \cup \boxChosen{t})]
			}
			\le 1 .
		\]
		\label{lemma:feasible_bounded:2} 
	\end{enumerate}
\end{lemma}

\begin{proof}
	We prove this statement via induction over the iteration $t \in \N_{0}$.
	For $t = 0$, note that $\currentPointset{0} = \emptyset$.
	Thus, \ref{lemma:feasible_bounded:1} is trivially true.
	Moreover, \ref{lemma:feasible_bounded:2} follows from applying \Cref{lemma:bayes_filter_bounds} to $\bayesFilterCorrection[\boxIds][\vectorize{v}][\emptyset]$ for every $\vectorize{v} \in \boxIds$.
	
	Now, suppose our claim holds at some iteration $t \in N$.
	We start by showing that \ref{lemma:feasible_bounded:1} holds in iteration $t+1$.
	First, note that if $\boxesToFix{t} = \emptyset$, then there is nothing to prove.
	Thus, we may assume $\boxesToFix{t} \neq \emptyset$.
	If $\bayesFilter{t} = 0$, then $\currentPointset{t+1} = \currentPointset{t}$.
	Thus, in this case, $\currentPointset{t+1}$ is feasible if and only if $\currentPointset{t}$ was feasible, which holds almost surely by the induction hypothesis. 
	Next, consider the case $\bayesFilter{t} = 1$ and set $\boxesToUpdate = \boxesToUpdate[\boxesToFix{t}][\boxChosen{t}][\updateRadius]$.
	By the induction hypothesis, we have that $\currentPointset{t} \cap \boxesToUpdate^{\comp} \subseteq \currentPointset{t}$ is almost surely feasible.
    Further, note that $ \gibbs[\boxesToUpdate][\currentPointset{t} \cap \boxesToUpdate^{\comp}]$ is only supported on $\eta \in \pointsets_{\boxesToUpdate}$ such that $\eta \cup (\currentPointset{t} \cap \boxesToUpdate^{\comp})$ is feasible.
    Hence, we have that $\currentPointset{t+1} = (\currentPointset{t+1} \cap \boxesToUpdate) \cup (\currentPointset{t} \cap \boxesToUpdate^{\comp})$ is almost surely feasible, proving \ref{lemma:feasible_bounded:1}.
	For \ref{lemma:feasible_bounded:2}, assume that $\boxesToFix{t+1} \neq \emptyset$.
	Applying \Cref{lemma:bayes_filter_bounds} for every non-empty $S \subseteq \boxIds$ and $\vectorize{v} \in S$ yields the desired bounds on $\bayesFilterCorrection[\boxesToFix{t+1}][\boxChosen{t+1}][\currentPointset{t+1}]$ whenever $\currentPointset{t+1}$ is feasible.
	As we have just shown, this is the case almost surely, which concludes the proof.
\end{proof}

Considering \Cref{algo:sampling}, an immediate question is how to construct the Bayes filter correction in line~\ref{algo:sampling:filter}, and in particular, how to do so in such a way that the Bayes filter $\bayesFilter{t}$ can be sampled efficiently.
However, we will defer this question for now and first prove that \Cref{algo:sampling}  produces the correct output distribution.

\section{Proof of correctness}
\label{secCorrectness}
In this section we prove that \Cref{algo:sampling} produces the correct output distribution.
That is for $\iterationNumber = \inf\set{t \in \N_{0}}[\boxesToFix{t} = \emptyset]$ it holds that $\currentPointset{\iterationNumber} \sim \gibbs[\region]$.  
We first show that the number of iterations $\iterationNumber$ is finite almost surely.
This directly follows as a corollary of \Cref{lemma:feasible_bounded}.
\begin{corollary} \label{cor:finite_iterations}
	\Cref{algo:sampling} terminates almost surely after finitely many iterations.
	That is, for $\iterationNumber = \inf\set{t \in \N_{0}}[\boxesToFix{t} = \emptyset]$ we have $\Pr{\iterationNumber < \infty} = 1$.
\end{corollary} 

\begin{proof}
	By \Cref{lemma:feasible_bounded}, the probability that $\size{\boxesToFix{t}}$ decreases in each step is uniformly bounded away from $0$.
	Thus, there is a positive probability of going from any $\boxesToFix{t}$ to the empty set in $\size{\boxesToFix{t}}$ steps.
	Since $\size{\boxesToFix{t}} \le \size{\boxIds}$, this means for every $k \in \N$, it holds that the probability that $\boxesToFix{k \cdot \size{\boxIds}} = \emptyset$ is bounded away from $0$ uniformly in $k$.
	Thus, $\iterationNumber$ is dominated by a geometric random variable with strictly positive success probability, which proves the claim. 
\end{proof}

Before going into the technical part of proving correctness, a few remarks about our notation are in place.
Firstly, for any $A \in \sigmafield$ with $\Pr{A} > 0$ we write $\PrSymbol[A]$ as shorthand for the probability measure $\Pr{\,\cdot}[][A] = \Pr{\,\cdot}[A]$.
Note that for all events $A, B \in \sigmafield$ with $\Pr{A \cap B} > 0$ it holds that $\PrSymbol[A \cap B] = (\PrSymbol[A])_{B}$.
Throughout our proofs, we use conditional expectations to make conditioning on partial point configurations $\currentPointset{t} \cap \subregion$ rigorous.
In particular, we frequently condition on a sub-$\sigma$-field $\mathcal{F} \subseteq \sigmafield$ and an event $A \in \sigmafield$ with $\Pr{A} > 0$ at the same time.
Formally, for a measurable function $f: \statespace \to \R$ we write $\E{f}[\mathcal{F}][A]$ for the conditional expectation of $f$ given $\mathcal{F}$ under the conditional measure $\PrSymbol[A]$.
Note that any identity involving $\E{f}[\mathcal{F}][A]$ should be understood to hold $\PrSymbol[A]$-almost surely. 
Moreover, if $f = \ind{B}$ for some event $B \in \mathcal{F}$, we write the conditional expectation as $\Pr{B}[\mathcal{F}; A]$.
More details on conditional expectations can be found in \Cref{appendix:measures}.
Moreover, we often make use of the concept of regular condition distributions.
For more details, see \Cref{appendix:rcd}.
Lastly, for every bounded measurable region $\region' \in \boundedBorel$, we write $\projection{\region'}: \pointsets \to \pointsets$ for the projection $\eta \mapsto \eta \cap \region'$.

Our main result in this section is the following statement.
\begin{theorem}
	\label{thm:correctness}
	For all $t \in \N_{0}$ with $\Pr{\boxesToFix{t} = \emptyset} > 0$ and all $A \in \pointsetEvents$, it holds that 
	\[
		\Pr{\currentPointset{t} \in A}[\boxesToFix{t} = \emptyset] = \gibbs[\region](A).
	\]
\end{theorem}
Since the algorithm terminates when $\boxesToFix{t} = \emptyset$, this implies that the output of \Cref{algo:sampling} follows the distribution $\gibbs[\region](A)$.
We  deduce \Cref{thm:correctness} from the following invariant. 

\begin{lemma} \label{lemma:invariant}
	For all $t \in \N_{0}$, $S \subseteq \boxIds$ with $\Pr{\boxesToFix{t} = S} > 0$ and $A \in \pointsetEvents$ it holds that 
	\[
		\Pr{\currentPointset{t} \cap S^{\comp} \in A}[\currentPointset{t} \cap S, \boxesToFix{t} = S] = \gibbs[S^{\comp}][\currentPointset{t} \cap S](A) .
	\] 
	In particular, the map 
	\[
		(\omega, A) \mapsto \gibbs[S^{\comp}][\currentPointset{t}(\omega) \cap S](A) \quad \text{$\omega \in \statespace$, $A \in \pointsetEvents$}
	\] 
	is a regular conditional distribution of $\currentPointset{t} \cap S^{\comp}$ given $\sigma(\currentPointset{t} \cap S)$ under the probability measure $\PrSymbol[\set{\boxesToFix{t} = S}]$. 
\end{lemma} 

Before we get into proving \Cref{lemma:invariant}, we first show how \Cref{thm:correctness} follows from it.

\begin{proof}[Proof of \Cref{thm:correctness}]
	Note that $\currentPointset{t} \cap \boxIds = \currentPointset{t}$ and $\sigma(\currentPointset{t} \cap \emptyset) = \set{\emptyset, \statespace}$ for all $t \in \N_{0}$. 
	Using \Cref{lemma:invariant} for $S = \emptyset$, we obtain for all $A \in \pointsetEvents$
	\begin{align*}
		\Pr{\currentPointset{t} \in A}[\boxesToFix{t} = \emptyset] 
		= \Pr{\currentPointset{t} \cap \boxIds \in A}[\currentPointset{t} \cap \emptyset, \boxesToFix{t} = \emptyset] = \gibbs[\region](A) ,
	\end{align*}
	which proves the theorem.
\end{proof}

We proceed by stating and proving several lemmas that we will use to prove \Cref{lemma:invariant}.

\begin{lemma}\label{lemma:invariant_projection}
	Fix  $t \in \N_{0}$, and assume that for all $S \subseteq \boxIds$ with $\Pr{\boxesToFix{t} = S} > 0$ and all $A \in \pointsetEvents$ it holds that 
	\[
		\Pr{\currentPointset{t} \cap S^{\comp} \in A}[\currentPointset{t} \cap S, \boxesToFix{t} = S] = \gibbs[S^{\comp}][\currentPointset{t} \cap S](A) .
	\] 
	Let $E = \set{\boxesToFix{t} = S, \boxChosen{t} = \vectorize{v}}$ for some $S \in \powerset{\boxIds} \setminus \set{\emptyset}$ and $\vectorize{v} \in S$ such that $\Pr{E} > 0$.
	For any measurable region $\region' \subseteq \complementOf{\boxRegion{S}}$ and any event $A \in \pointsetEvents$ it holds that
	\[
		\Pr{\currentPointset{t} \cap \region' \in A}[\currentPointset{t} \cap S; E]  
		= \gibbs[S^{\comp}][\currentPointset{t} \cap S] \circ \projection{\region'}^{-1}(A) .
	\]
	In particular, 
	\[
		(\omega, A) \mapsto \gibbs[S^{\comp}][\currentPointset{t}(\omega) \cap S] \circ \projection{\region'}^{-1}(A) \quad \text{$\omega \in \statespace$, $A \in \pointsetEvents$}
	\] 
	is a regular conditional distribution of $\currentPointset{t} \cap \region'$ given $\sigma(\currentPointset{t} \cap S)$ under $\PrSymbol[E]$.
\end{lemma}

\begin{proof}
	Fix some measurable region $\subregion \subseteq \complementOf{\boxRegion{S}}$ and note that  $\gibbs[S^{\comp}][\currentPointset{t} \cap S] \circ \projection{\region'}^{-1}$ is a probability distribution on $(\pointsets, \pointsetEvents)$.
	Thus, it suffices to show that  $\gibbs[S^{\comp}][\currentPointset{t} \cap S] \circ \projection{\region'}^{-1}(A)$ is also a version of the conditional expectation $\Pr{\currentPointset{t} \cap \region' \in A}[\currentPointset{t} \cap S; E]$ for all events $A \in \pointsetEvents$.
	By the assumptions of the lemma, we have
	\begin{align*}
		\Pr{\currentPointset{t} \cap \region' \in A}[\currentPointset{t} \cap S; \boxesToFix{t} = S] 
		&= \Pr{\currentPointset{t} \cap S^{\comp} \in \projection{\region'}^{-1}(A)}[\currentPointset{t} \cap S; \boxesToFix{t} = S] \\
		&= \gibbs[S^{\comp}][\currentPointset{t} \cap S] \circ \projection{\region'}^{-1}(A) .
	\end{align*}
	Next, we use \Cref{lemma:adding_conditions} part \ref{lemma:adding_conditions:2} to argue that this still holds if we change the probability measure from $\PrSymbol[\set{\boxesToFix{t} = S}]$ to $\PrSymbol[E]$.
	Note that, given $\boxesToFix{t} = S$, $\boxChosen{t}$ is chosen uniformly from $S$ independent of $\currentPointset{t}$. 
	Therefore, we have
	\begin{align*}
		\Pr{\boxChosen{t} = \vectorize{v}}[\currentPointset{t} \cap S; \boxesToFix{t} = S] 
		&= 
		\Pr{\boxChosen{t} = \vectorize{v}}[ \boxesToFix{t} = S] \\
		&= \Pr{\boxChosen{t} = \vectorize{v}}[\currentPointset{t} \cap S, \currentPointset{t} \cap \region';\boxesToFix{t} = S],
	\end{align*}
    and applying \Cref{lemma:adding_conditions} part \ref{lemma:adding_conditions:2} proves the claim.
\end{proof}

We use the next lemma to prove \Cref{lemma:invariant} in the case that $\bayesFilter{t} = 1$. 
\begin{lemma} \label{lemma:invariant_case1}
	Under the assumptions of \Cref{lemma:invariant_projection}, let $E_1 = \set{\boxesToFix{t} = S, \boxChosen{t} = \vectorize{v}, \bayesFilter{t} = 1}$ for some $S \in \powerset{\boxIds} \setminus \set{\emptyset}$ and $\vectorize{v} \in S$ such that $\Pr{E_1} > 0$, and set $R = S \setminus \vectorize{v}$.
	For all $A \in \pointsetEvents$
	it holds that
	\[
		\Pr{\currentPointset{t+1} \cap R^{\comp} \in A}[\currentPointset{t+1} \cap R; E_1] = \gibbs[R^{\comp}][\currentPointset{t+1} \cap R](A) .
	\]
\end{lemma} 

\begin{proof} 
    An example that might help to keep track of the relevant regions throughout this proof is given in \Cref{fig:regions_F=1}. 
    \begin{figure}[h]
	   \centering
	   \includegraphics[width=0.7\textwidth]{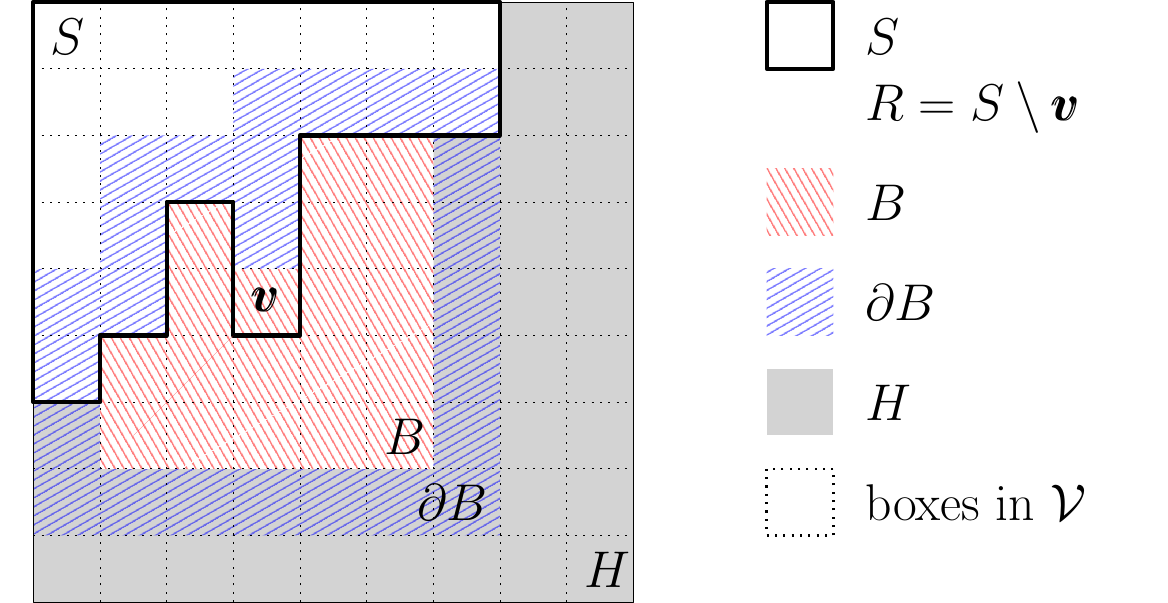}
	   \caption{Example for the regions considered throughout the proof of \Cref{lemma:invariant_case1} for $\updateRadius = 2$.}
	   \label{fig:regions_F=1}
    \end{figure} 

	Let $E = \set{\boxesToFix{t} = S, \boxChosen{t} = \vectorize{v}}$ for $S$ and $\vectorize{v}$ as in the definition of $E_1$, and set $\boxesToUpdate = \boxesToUpdate[S][\vectorize{v}][\updateRadius]$ and $H = \boxIds \setminus (S \cup \boxesToUpdate)$.
    Note that it suffices to consider $A \in \pointsetEvents_{R^{\comp}}$.
    Further, observe that $R^{\comp}$ is the disjoint union of $H$ and $\boxesToUpdate$.
    Hence, by \Cref{lemma:pi_system}, it suffices to prove our claim for events of the form $\{\currentPointset{t+1} \cap H \in A_{H}, \currentPointset{t+1} \cap \boxesToUpdate \in A_{\boxesToUpdate}\}$ for pairs $A_H \in \pointsetEvents[H]$, $A_{\boxesToUpdate} \in \pointsetEvents[\boxesToUpdate]$.
    Formally, we need to show
    \begin{align}
        \Pr{\currentPointset{t+1} \cap H \in A_H, \currentPointset{t+1} \cap \boxesToUpdate \in A_{\boxesToUpdate}}[\currentPointset{t+1} \cap S; E_1] = \gibbs[R^{\comp}][\currentPointset{t} \cap R](\projection{H}^{-1}(A_H) \cap \projection{\boxesToUpdate}^{-1}(A_{\boxesToUpdate})). \label{lemma:invariant:eq:4}
    \end{align}
	
    To establish \eqref{lemma:invariant:eq:4}, we derive a suitable version of the conditional expectation 
    \begin{align}
        \Pr{\currentPointset{t} \cap H \in A_H, \currentPointset{t+1} \cap \boxesToUpdate \in A_{\boxesToUpdate}, \bayesFilter{t} = 1}[\currentPointset{t} \cap S; E]. \label{lemma:invariant:eq6}
    \end{align}
    Our first step is to obtain an expression for $\Pr{\currentPointset{t+1} \cap \boxesToUpdate \in A_{\boxesToUpdate},\bayesFilter{t} = 1}[\currentPointset{t} \cap (\boxesToUpdate \setminus \vectorize{v})^{\comp}; E]$.
    This is give by the following claim.
    \begin{claim} \label{claim:invariant_case1:step1}
        For
        \[
            g: \omega \mapsto \gibbs[\boxesToUpdate][\currentPointset{t}(\omega) \cap \boxesToUpdate^{\comp}](A_{\boxesToUpdate})
            \,\text{ and }\\
            h: \omega \mapsto \bayesFilterCorrection[S][\vectorize{v}][\currentPointset{t}(\omega) \cap S] \cdot
        		  \frac{
        			 \partitionFunction[\boxesToUpdate][\currentPointset{t}(\omega) \cap \partial \boxesToUpdate] 
        		}{ 
        			 \partitionFunction[\boxesToUpdate \setminus \vectorize{v}][\currentPointset{t}(\omega) \cap (\partial \boxesToUpdate \cup \vectorize{v})]
        		} 
        \]
        it holds that
    	\[
    		\Pr{\currentPointset{t+1} \cap \boxesToUpdate \in A_{\boxesToUpdate}, \bayesFilter{t} = 1}[\currentPointset{t} \cap (\boxesToUpdate \setminus \vectorize{v})^{\comp}; E](\omega) = g(\omega) \cdot h(\omega)
    	\]
	    for $\PrSymbol[E]$-almost all $\omega \in \statespace$.
    \end{claim}
    
	Our next step is to derive an expression for the conditional expectation in \eqref{lemma:invariant:eq6} for every $A_{H} \in \pointsetEvents[H]$.
	To this end, note that by \Cref{claim:invariant_case1:step1} we have
	\begin{align}
		&\Pr{\currentPointset{t} \cap H \in A_H, \currentPointset{t+1} \cap \boxesToUpdate \in A_{\boxesToUpdate}, \bayesFilter{t} = 1}[\currentPointset{t} \cap S; E] \nonumber\\
		&\quad\quad= \E{\ind{\currentPointset{t} \cap H \in A_H} \cdot \Pr{\currentPointset{t+1} \cap \boxesToUpdate \in A_{\boxesToUpdate}, \bayesFilter{t} = 1}[\currentPointset{t} \cap (\boxesToUpdate \setminus \vectorize{v})^{\comp}; E]}[\currentPointset{t} \cap S][E] \nonumber\\
		&\quad\quad= \E{\ind{\currentPointset{t} \cap H \in A_H} \cdot g \cdot h}[\currentPointset{t} \cap S][E]  \label{lemma:invariant:eq:1},
	\end{align}
	where the first equality comes from the fact that $\sigma(\currentPointset{t} \cap S) \subseteq \sigma(\currentPointset{t} \cap (\boxesToUpdate \setminus \vectorize{v})^{\comp})$.
    Our goal is now to compute \eqref{lemma:invariant:eq:1} by using \Cref{thm:integration_RCP}, which tells us that the conditional expectation can be calculated by integrating $g \cdot h$ over $A_{H}$ against a regular conditional distribution for $\currentPointset{t} \cap H$ given $\sigma(\currentPointset{t} \cap S)$ under the measure $\PrSymbol[E]$.
	By \Cref{lemma:invariant_projection}, the map
	\[
		(\omega, A) \mapsto 
		\gibbs[S^{\comp}][\currentPointset{t}(\omega) \cap S] \circ \projection{H}^{-1}(A)  
		\quad \text{ for $\omega \in \statespace$, $A \in \pointsetEvents$}
	\] 
	is such a regular conditional distribution.
    We obtain the following claim.
    \begin{claim} \label{claim:invariant_case1:step2}
        It holds that 
        \begin{align*}
    		&\Pr{\currentPointset{t} \cap H \in A_H, \currentPointset{t+1} \cap \boxesToUpdate \in A_{\boxesToUpdate}, \bayesFilter{t} = 1}[\currentPointset{t} \cap S; E] \\
    		&\quad\quad= \bayesFilterCorrection[S][\vectorize{v}][\currentPointset{t} \cap S] \cdot \frac{
    			\partitionFunction[R^{\comp}][\currentPointset{t} \cap R] 
    		}{
    			\partitionFunction[S^{\comp}][\currentPointset{t} \cap S]
    		} \cdot \gibbs[R^{\comp}][\currentPointset{t} \cap R](\projection{H}^{-1}(A_H) \cap \projection{\boxesToUpdate}^{-1}(A_{\boxesToUpdate}))
	\end{align*}
    \end{claim}
    
	Given \Cref{claim:invariant_case1:step2}, we proceed by using \Cref{lemma:conditioning_on_event} to switch from $\PrSymbol[E]$ to $\PrSymbol[E_1]$.
	To this end, note that
	\begin{align*}
		\Pr{\bayesFilter{t} = 1}[\currentPointset{t} \cap S; E]
		&= \Pr{\currentPointset{t} \cap H \in \pointsets_{H}, \currentPointset{t+1} \cap \boxesToUpdate \in \pointsets_{\boxesToUpdate}, \bayesFilter{t} = 1}[\currentPointset{t} \cap S; E] \\
		&= \bayesFilterCorrection[S][\vectorize{v}][\currentPointset{t} \cap S] \cdot \frac{
			\partitionFunction[R^{\comp}][\currentPointset{t} \cap R] 
		}{
			\partitionFunction[S^{\comp}][\currentPointset{t} \cap S]
		}.
	\end{align*}
	Thus, by \Cref{lemma:conditioning_on_event} it holds that
	\begin{align*}
		&\Pr{\currentPointset{t} \cap H \in A_H, \currentPointset{t+1} \cap \boxesToUpdate \in A_{\boxesToUpdate}}[\currentPointset{t} \cap S; E_1] \\
		& \quad \quad = \frac{\Pr{\currentPointset{t} \cap H \in A_H, \currentPointset{t+1} \cap \boxesToUpdate \in A_{\boxesToUpdate}, \bayesFilter{t} = 1}[\currentPointset{t} \cap S; E]}{\Pr{\bayesFilter{t} = 1}[\currentPointset{t} \cap S; E]} \\
		& \quad \quad = \gibbs[R^{\comp}][\currentPointset{t} \cap R](\projection{H}^{-1}(A_H) \cap \projection{\boxesToUpdate}^{-1}(A_{\boxesToUpdate})) .
	\end{align*} 
	As the right-hand side is $\sigma(\currentPointset{t} \cap R)$-measurable and $\sigma(\currentPointset{t} \cap R) \subseteq \sigma(\currentPointset{t} \cap S)$ since $R \subseteq S$, we have
    \[
        \Pr{\currentPointset{t} \cap H \in A_H, \currentPointset{t+1} \cap \boxRegion{\boxesToUpdate} \in A_{\boxesToUpdate}}[\currentPointset{t} \cap R; E_1] 
        = \gibbs[R^{\comp}][\currentPointset{t} \cap R](\projection{H}^{-1}(A_H) \cap \projection{\boxesToUpdate}^{-1}(A_{\boxesToUpdate}))
    \]
    Moreover, given $E_1$, it holds that $\currentPointset{t+1} \cap H = \currentPointset{t} \cap H$ and $\currentPointset{t+1} \cap R = \currentPointset{t} \cap R$, which proves \eqref{lemma:invariant:eq:4}.
    Finally, applying \Cref{lemma:pi_system} as discussed earlier proves the lemma.
\end{proof}

Before we proceed with our main argument, we prove \Cref{claim:invariant_case1:step1} and \Cref{claim:invariant_case1:step2}.
\begin{proof}[Proof of \Cref{claim:invariant_case1:step1}]
    Note that given the event $E_1$, the point set $\currentPointset{t+1} \cap \boxesToUpdate$ is sampled from $\gibbs[\boxesToUpdate][\currentPointset{t} \cap \boxesToUpdate^{\comp}]$.
	Therefore we have for all $A_{\boxesToUpdate} \in \pointsetEvents[\boxesToUpdate]$ that
	\[
		\Pr{\currentPointset{t+1} \cap \boxesToUpdate \in A_{\boxesToUpdate}}[\currentPointset{t} \cap \boxesToUpdate^{\comp}; E_1] 
		= \gibbs[\boxesToUpdate][\currentPointset{t} \cap \boxesToUpdate^{\comp}](A_{\boxesToUpdate}) = g .
	\]
	Since the above is true independently of the partial configuration $\currentPointset{t} \cap \vectorize{v}$, 
 we obtain
	\begin{align*}	
		\Pr{\currentPointset{t+1} \cap \boxesToUpdate \in A_{\boxesToUpdate}}[\currentPointset{t} \cap (\boxesToUpdate \setminus \vectorize{v})^{\comp}; E_1]
		&= \Pr{\currentPointset{t+1} \cap \boxesToUpdate \in A_{\boxesToUpdate}}[\currentPointset{t} \cap \boxesToUpdate^{\comp}, \currentPointset{t} \cap \vectorize{v}; E_1] \\
        &= \Pr{\currentPointset{t+1} \cap \boxesToUpdate \in A_{\boxesToUpdate}}[\currentPointset{t} \cap \boxesToUpdate^{\comp}; E_1] \\
		&= g 
	\end{align*}
	Moreover, given $E$ and $\currentPointset{t} \cap \boxesToUpdate^{\comp}$, $\bayesFilter{t}$ is drawn as a Bernoulli random variable with success probability
	\[
		\bayesFilterCorrection[S][\vectorize{v}][\currentPointset{t}] \cdot
		\frac{
			\partitionFunction[\boxesToUpdate][\currentPointset{t} \cap \partial \boxesToUpdate]
		}{ 
			\partitionFunction[\boxesToUpdate \setminus \vectorize{v}][\currentPointset{t} \cap (\partial \boxesToUpdate \cup \vectorize{v})]
		} ,
	\]
	where $\bayesFilterCorrection$ is a Bayes filter correction as in \Cref{def:bayes_filter}.
	Thus, it holds that
	\begin{align*}
		\Pr{\bayesFilter{t} = 1}[\currentPointset{t} \cap (\boxesToUpdate \setminus \vectorize{v})^{\comp};E] 
		&= \bayesFilterCorrection[S][\vectorize{v}][\currentPointset{t}] \cdot 
		\frac{
			\partitionFunction[\boxesToUpdate][\currentPointset{t} \cap \partial \boxesToUpdate] 
		}{ 
			\partitionFunction[\boxesToUpdate \setminus \vectorize{v}][\currentPointset{t} \cap (\partial \boxesToUpdate \cup \vectorize{v})]
		} \\
		&=  \bayesFilterCorrection[S][\vectorize{v}][\currentPointset{t} \cap S] \cdot
		\frac{
			 \partitionFunction[\boxesToUpdate][\currentPointset{t} \cap \partial \boxesToUpdate] 
		}{ 
			 \partitionFunction[\boxesToUpdate \setminus \vectorize{v}][\currentPointset{t} \cap (\partial \boxesToUpdate \cup \vectorize{v})]
		} ,
	\end{align*}
	where the second equality holds since $\bayesFilterCorrection[S][\vectorize{v}][\cdot] = \bayesFilterCorrection[S][\vectorize{v}][\cdot \cap S]$ by definition.
	We note that the right-hand side of the last equality is precisely $h$.
	By \Cref{lemma:conditioning_on_event}, we then have
	\[
		\Pr{\currentPointset{t+1} \cap \boxesToUpdate \in A_{\boxesToUpdate}, \bayesFilter{t} = 1}[\currentPointset{t} \cap (\boxesToUpdate \setminus \vectorize{v})^{\comp}; E](\omega) = g(\omega) \cdot h(\omega)
	\]
	for $\PrSymbol[E]$-almost all $\omega \in \statespace$.
\end{proof}

\begin{proof}[Proof of \Cref{claim:invariant_case1:step2}]
   Recall that by \Cref{claim:invariant_case1:step1} and the fact that $\sigma(\currentPointset{t} \cap S) \subseteq \sigma(\currentPointset{t} \cap (\boxesToUpdate \setminus \vectorize{v})^{\comp})$ we have
	\begin{align*}
		&\Pr{\currentPointset{t} \cap H \in A_H, \currentPointset{t+1} \cap \boxesToUpdate \in A_{\boxesToUpdate}, \bayesFilter{t} = 1}[\currentPointset{t} \cap S; E] \nonumber\\
		&\quad\quad= \E{\ind{\currentPointset{t} \cap H \in A_H} \cdot \Pr{\currentPointset{t+1} \cap \boxesToUpdate \in A_{\boxesToUpdate}, \bayesFilter{t} = 1}[\currentPointset{t} \cap (\boxesToUpdate \setminus \vectorize{v})^{\comp}; E]}[\currentPointset{t} \cap S][E] \nonumber\\
		&\quad\quad= \E{\ind{\currentPointset{t} \cap H \in A_H} \cdot g \cdot h}[\currentPointset{t} \cap S][E]  .
	\end{align*}
    Next, we define functions $\hat{g}, \hat{h}: \pointsets[H] \times \pointsets[S] \to \R_{\ge 0}$ via
	\begin{align*}
		\hat{g}(\eta_1, \eta_2) 
		&= \gibbs[\boxesToUpdate][(\eta_1 \cup \eta_2) \cap \boxesToUpdate^{\comp}](A_{\boxesToUpdate})\\ 
		\hat{h}(\eta_1, \eta_2) 
		&= \bayesFilterCorrection[S][\vectorize{v}][\eta_2] \cdot \frac{
			\partitionFunction[\boxesToUpdate][(\eta_1 \cup \eta_2) \cap \partial \boxesToUpdate] 
		}{ 
			\partitionFunction[\boxesToUpdate \setminus \vectorize{v}][(\eta_1 \cup \eta_2) \cap (\partial \boxesToUpdate \cup \vectorize{v})]
		} ,
	\end{align*}
	and observe that for all $\omega \in \statespace$ it holds that $\hat{g}(\currentPointset{t}(\omega) \cap H, \currentPointset{t}(\omega) \cap S) = g(\omega)$ and $\hat{h}(\currentPointset{t}(\omega) \cap H, \currentPointset{t}(\omega) \cap S) = h(\omega)$.
	Thus, we have
    \begin{align*}
        &\E{\ind{\currentPointset{t} \cap H \in A_H} \cdot g \cdot h}[\currentPointset{t} \cap S][E] \\
		&= \E{\ind{\currentPointset{t} \cap H \in A_H} \cdot \hat{g}(\currentPointset{t} \cap H, \currentPointset{t} \cap S) \cdot \hat{h}(\currentPointset{t} \cap H, \currentPointset{t} \cap S)}[\currentPointset{t} \cap S][E] .
    \end{align*}
    Further, by \Cref{lemma:invariant_projection}, the map
	\[
		(\omega, A) \mapsto 
		\gibbs[S^{\comp}][\currentPointset{t}(\omega) \cap S] \circ \projection{H}^{-1}(A)  
		\quad \text{ for $\omega \in \statespace$, $A \in \pointsetEvents$}
	\] 
	is a regular conditional distribution for $\currentPointset{t} \cap H$ given $\currentPointset{t} \cap S$ under the measure $\PrSymbol_{E}$.
    Using \Cref{thm:integration_RCP} then yields
	\begin{align}
		&\E{\ind{\currentPointset{t} \cap H \in A_H} \cdot g \cdot h}[\currentPointset{t} \cap S][E] (\omega) \nonumber\\
		&= \int_{A_H} \hat{g}(\eta, \currentPointset{t}(\omega) \cap S) \cdot \hat{h}(\eta, \currentPointset{t}(\omega) \cap S) \,\gibbs[S^{\comp}][\currentPointset{t}(\omega) \cap S] \circ \projection{H}^{-1} (\intD \eta) \nonumber\\
		&= \frac{1}{\partitionFunction[S^{\comp}][\currentPointset{t}(\omega) \cap S]} \sum_{n \ge 0} \frac{1}{n!} \int_{\boxRegion{H}^n} \ind{\setFromTuple{\vectorize{x}} \in A_H} \cdot \hat{g}(\setFromTuple{\vectorize{x}}, \currentPointset{t}(\omega) \cap S) \cdot  \hat{h}(\setFromTuple{\vectorize{x}}, \currentPointset{t}(\omega) \cap S)\nonumber\\
		& \hspace*{10em}\cdot (\activity_{\currentPointset{t}(\omega) \cap S})^{\vectorize{x}} \cdot \eulerE^{-\hamiltonian[\vectorize{x}]} \cdot \partitionFunction[S^{\comp} \cap H^{\comp}][\setFromTuple{\vectorize{x}} \cup (\currentPointset{t}(\omega) \cap S)] \intD \vectorize{x} \nonumber\\
		&= \frac{1}{\partitionFunction[S^{\comp}][\currentPointset{t}(\omega) \cap S]} \sum_{n \ge 0} \frac{1}{n!} \int_{\boxRegion{H}^n} \ind{\setFromTuple{\vectorize{x}} \in A_H} \cdot \hat{g}(\setFromTuple{\vectorize{x}}, \currentPointset{t}(\omega) \cap S) \cdot  \hat{h}(\setFromTuple{\vectorize{x}}, \currentPointset{t}(\omega) \cap S) \nonumber\\
		& \hspace*{10em}\cdot (\activity_{\currentPointset{t}(\omega) \cap S})^{\vectorize{x}} \cdot \eulerE^{-\hamiltonian[\vectorize{x}]} \cdot \partitionFunction[\boxesToUpdate \setminus \vectorize{v}][\setFromTuple{\vectorize{x}} \cup (\currentPointset{t}(\omega) \cap S)] \intD \vectorize{x}. \label{lemma:invariant:eq:2}
	\end{align}
	
	We proceed by simplifying (\ref{lemma:invariant:eq:2}). 
	Observe that 
	\begin{align*}
		\hat{h}(\setFromTuple{\vectorize{x}}, \currentPointset{t} \cap S) 
		&= \bayesFilterCorrection[S][\vectorize{v}][\currentPointset{t} \cap S] \cdot \frac{
			\partitionFunction[\boxesToUpdate][(\setFromTuple{\vectorize{x}} \cup (\currentPointset{t} \cap S)) \cap \partial \boxesToUpdate] 
		}{ 
			\partitionFunction[\boxesToUpdate \setminus \vectorize{v}][(\setFromTuple{\vectorize{x}} \cup (\currentPointset{t} \cap S)) \cap (\partial \boxesToUpdate \cup \vectorize{v})]
		} \\
		&= \bayesFilterCorrection[S][\vectorize{v}][\currentPointset{t} \cap S] \cdot \frac{
			\partitionFunction[\boxesToUpdate][(\setFromTuple{\vectorize{x}} \cup (\currentPointset{t} \cap S)) \cap \boxesToUpdate^{\comp}]
		}{ 
			\partitionFunction[\boxesToUpdate \setminus \vectorize{v}][(\setFromTuple{\vectorize{x}} \cup (\currentPointset{t} \cap S)) \cap (\boxesToUpdate \setminus \vectorize{v})^{\comp}]
		} \\
		&= \bayesFilterCorrection[S][\vectorize{v}][\currentPointset{t} \cap S] \cdot \frac{
			\partitionFunction[\boxesToUpdate][\setFromTuple{\vectorize{x}} \cup (\currentPointset{t} \cap R)] 
		}{ 
			\partitionFunction[\boxesToUpdate \setminus \vectorize{v}][\setFromTuple{\vectorize{x}} \cup (\currentPointset{t} \cap S)]
		}.
	\end{align*}
	Here, the second equality follows from \Cref{lemma:spatial_markov} and the fact that, for all $\eta \in \pointsets[\region]$, it holds that the distance between $\boxRegion{\boxesToUpdate}^{\comp} \setminus \boxRegion{\partial \boxesToUpdate}$ and $\boxRegion{\boxesToUpdate}$, and the distance between $\boxRegion{\boxesToUpdate \setminus \vectorize{v}}^{\comp} \setminus \boxRegion{\partial \boxesToUpdate \cup \vectorize{v}}$ and $\boxRegion{\boxesToUpdate \setminus \vectorize{v}}$ are at least $\range$ (the range of the potential).
	Moreover, the last equality follows from the fact that $S \cap \boxesToUpdate^{\comp} = S \setminus \vectorize{v} = R$, $\setFromTuple{\vectorize{x}} \subseteq \boxRegion{H} \subseteq \boxRegion{B}^{\comp} \subseteq \boxRegion{B \setminus \vectorize{v}}^{\comp}$ and $S \subseteq (\boxesToUpdate \setminus \vectorize{v})^{\comp}$. 
	Further, note that
	\begin{align*}
		\hat{g}(\setFromTuple{\vectorize{x}}, \currentPointset{t} \cap S) 
		&= \gibbs[\boxesToUpdate][(\setFromTuple{\vectorize{x}} \cup (\currentPointset{t} \cap S)) \cap \boxesToUpdate^{\comp}](A_{\boxesToUpdate}) 
        = \gibbs[\boxesToUpdate][\setFromTuple{\vectorize{x}} \cup (\currentPointset{t} \cap R)](A_{\boxesToUpdate}) \\
		&= \frac{\sum_{m \ge 0} \frac{1}{m!} \int_{\boxRegion{\boxesToUpdate}^m} \ind{\setFromTuple{\vectorize{y}} \in A_{\boxesToUpdate}} \cdot (\activity_{\setFromTuple{\vectorize{x}} \cup (\currentPointset{t} \cap R)})^{\vectorize{y}} \cdot \eulerE^{- \hamiltonian[\vectorize{y}]} \intD \vectorize{y}}{\partitionFunction[\boxesToUpdate][\setFromTuple{\vectorize{x}} \cup (\currentPointset{t} \cap R)] } .
	\end{align*}
	Substituting both back into (\ref{lemma:invariant:eq:2}) and canceling $\partitionFunction[\boxesToUpdate][\setFromTuple{\vectorize{x}} \cup (\currentPointset{t} \cap R)] $ and $\partitionFunction[\boxesToUpdate \setminus \vectorize{v}][\setFromTuple{\vectorize{x}} \cup (\currentPointset{t} \cap S)]$ yields
	\begin{align*}
		&\Pr{\currentPointset{t} \cap H \in A_H, \currentPointset{t+1} \cap \boxesToUpdate \in A_{\boxesToUpdate}, \bayesFilter{t} = 1}[\currentPointset{t} \cap S; E] \\
		&\hspace*{5em}= 
		\frac{
			\bayesFilterCorrection[S][\vectorize{v}][\currentPointset{t} \cap S]
		}{
			\partitionFunction[S^{\comp}][\currentPointset{t} \cap S]
		}
		\sum_{n \ge 0} \frac{1}{n!} \int_{\boxRegion{H}^n} \ind{\setFromTuple{\vectorize{x}} \in A_H} \cdot (\activity_{\currentPointset{t} \cap S})^{\vectorize{x}} \cdot \eulerE^{-\hamiltonian[\vectorize{x}]} \\
		&\hspace*{10em}\cdot \left[\sum_{m \ge 0} \frac{1}{m!} \int_{\boxRegion{\boxesToUpdate}^m} \ind{\setFromTuple{\vectorize{y}} \in A_{\boxesToUpdate}} \cdot (\activity_{\setFromTuple{\vectorize{x}} \cup (\currentPointset{t} \cap R)})^{\vectorize{y}} \cdot \eulerE^{- \hamiltonian[\vectorize{y}]} \intD \vectorize{y}\right] \intD \vectorize{x} .
	\end{align*}
	Moreover, note that $\dist{\boxRegion{\vectorize{v}}}{\boxRegion{H}} \ge \range$. 
	Thus, it holds for $n \in \N_{0}$ and $\vectorize{x} = (x_1, \dots, x_n) \in \boxRegion{H}^n$ that
	\begin{align*}
		(\activity_{\currentPointset{t} \cap S})^{\vectorize{x}} 
		&= \activity^{n} \eulerE^{- \sum_{i=1}^{n} \sum_{z \in \currentPointset{t} \cap S} \potential[x_i][z]} \\
		&= \activity^{n} \eulerE^{- \sum_{i=1}^{n} \sum_{z \in \currentPointset{t} \cap R} \potential[x_i][z]} \\
		&= (\activity_{\currentPointset{t} \cap R})^{\vectorize{x}}.
	\end{align*}
	Consequently, we have
	\begin{align*}
		&\Pr{\currentPointset{t} \cap H \in A_H, \currentPointset{t+1} \cap \boxesToUpdate \in A_{\boxesToUpdate}, \bayesFilter{t} = 1}[\currentPointset{t} \cap S; E] \\
		&\quad= \frac{
			\bayesFilterCorrection[S][\vectorize{v}][\currentPointset{t} \cap S]
		}{
			\partitionFunction[S^{\comp}][\currentPointset{t} \cap S]
		} 
		\sum_{n \ge 0} \frac{1}{n!} \int_{\boxRegion{H}^n} \ind{\setFromTuple{\vectorize{x}} \in A_H} \cdot (\activity_{\currentPointset{t} \cap R})^{\vectorize{x}} \cdot \eulerE^{-\hamiltonian[\vectorize{x}]} \\
		&\hspace*{10em}\cdot \left[\sum_{m \ge 0} \frac{1}{m!} \int_{\boxRegion{\boxesToUpdate}^m} \ind{\setFromTuple{\vectorize{y}} \in A_{\boxesToUpdate}} \cdot (\activity_{\setFromTuple{\vectorize{x}} \cup (\currentPointset{t} \cap R)})^{\vectorize{y}} \cdot \eulerE^{- \hamiltonian[\vectorize{y}]} \intD \vectorize{y}\right] \intD \vectorize{x}.
	\end{align*}
	Next, note that, since $H \cup \boxesToUpdate = \complementOf{R}$, we it holds that
	\begin{align*}
		&\sum_{n \ge 0} \frac{1}{n!} \int_{\boxRegion{H}^n} \ind{\setFromTuple{\vectorize{x}} \in A_H} \cdot (\activity_{\currentPointset{t} \cap R})^{\vectorize{x}} \eulerE^{-\hamiltonian[\vectorize{x}]} \left[\sum_{m \ge 0} \frac{1}{m!} \int_{\boxRegion{\boxesToUpdate}^m} \ind{\setFromTuple{\vectorize{y}} \in A_{\boxesToUpdate}} (\activity_{\setFromTuple{\vectorize{x}} \cup (\currentPointset{t} \cap R)})^{\vectorize{y}} \eulerE^{- \hamiltonian[\vectorize{y}]} \intD \vectorize{y}\right] \intD \vectorize{x} \\
		&\quad\quad = \partitionFunction[R^{\comp}][\currentPointset{t} \cap R] 
		\cdot \gibbs[R^{\comp}][\currentPointset{t} \cap R](\projection{H}^{-1}(A_H) \cap \projection{\boxesToUpdate}^{-1}(A_{\boxesToUpdate}))
	\end{align*}
	This finally yields the expression 
	\begin{align*}
		&\Pr{\currentPointset{t} \cap H \in A_H, \currentPointset{t+1} \cap \boxesToUpdate \in A_{\boxesToUpdate}, \bayesFilter{t} = 1}[\currentPointset{t} \cap S; E] \\
		&\quad\quad= \bayesFilterCorrection[S][\vectorize{v}][\currentPointset{t} \cap S] \cdot \frac{
			\partitionFunction[R^{\comp}][\currentPointset{t} \cap R] 
		}{
			\partitionFunction[S^{\comp}][\currentPointset{t} \cap S]
		} \cdot \gibbs[R^{\comp}][\currentPointset{t} \cap R](\projection{H}^{-1}(A_H) \cap \projection{\boxesToUpdate}^{-1}(A_{\boxesToUpdate}))
	\end{align*}
\end{proof}

We now continue with our main argument.
The next lemma is the counterpart of \Cref{lemma:invariant_case1} for the case $\bayesFilter{t} = 0$.

\begin{lemma} \label{lemma:invariant_case0}
	Under the assumptions of \Cref{lemma:invariant_projection}, let $E_0 = \set{\boxesToFix{t} = S, \boxChosen{t} = \vectorize{v}, \bayesFilter{t} = 0}$ for some $S \in \powerset{\boxIds} \setminus \set{\emptyset}$ and $\vectorize{v} \in S$ such that $\Pr{E_0} > 0$, and set $\boxesToUpdate = \boxesToUpdate[S][\vectorize{v}][\updateRadius]$ and $R = S \cup \partial \boxesToUpdate$.
	For all $A \in \pointsetEvents$
	it holds that
	\[
		\Pr{\currentPointset{t+1} \cap R^{\comp} \in A}[\currentPointset{t+1} \cap R; E_0] 
		= \gibbs[R^{\comp}][\currentPointset{t+1} \cap R](A) . 
	\]
\end{lemma} 

\begin{proof}
    An example that might help to keep track of the relevant regions throughout this proof is given in \Cref{fig:regions_F=0}. 
    \begin{figure}[h]
	   \centering
	   \includegraphics[width=0.7\textwidth]{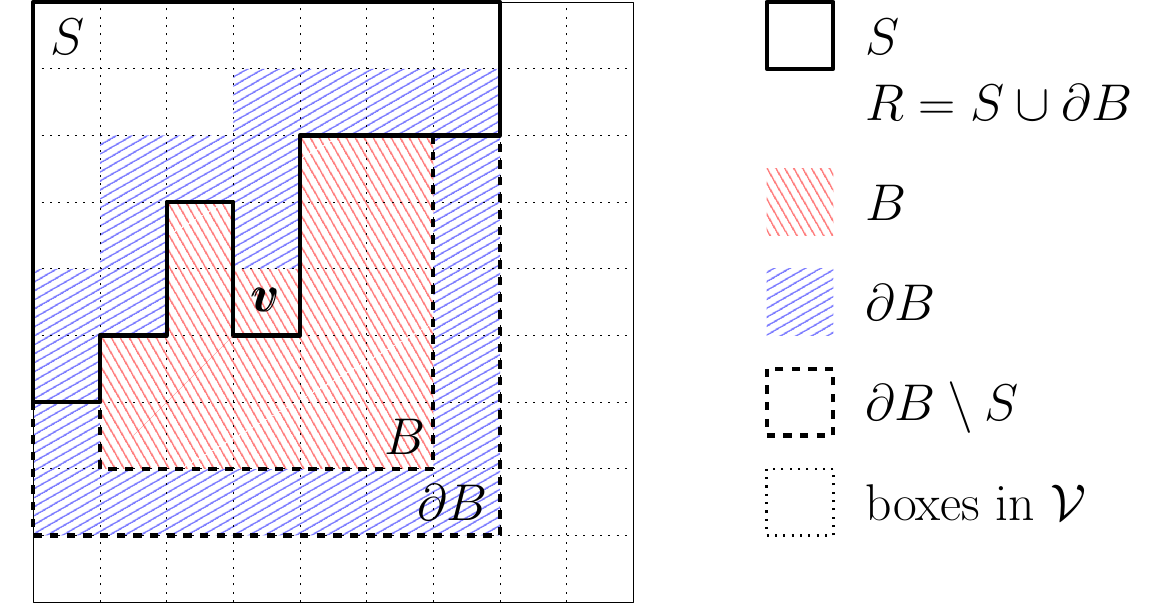}
	   \caption{Example for the regions considered throughout the proof of \Cref{lemma:invariant_case0} for $\updateRadius = 2$.}
	   \label{fig:regions_F=0}
    \end{figure}
    
	Let $E = \set{\boxesToFix{t} = S, \boxChosen{t} = \vectorize{v}}$ for $S$ and $\vectorize{v}$ as in the definition of $E_0$.	
    Since $\currentPointset{t+1} \cap R^{\comp} \in \pointsets_{R^{\comp}}$ and $\gibbs[R^{\comp}][\currentPointset{t+1} \cap R]$ is only supported on $\pointsets_{R^{\comp}}$, it again suffices to prove the statement for $A \in \pointsetEvents_{R^{\comp}}$.
	Our first step is to show that the following claim
	\begin{claim} \label{claim:invariant_case0}
	    It holds that
        \[
            \Pr{\currentPointset{t} \cap R^{\comp} \in A}[\currentPointset{t} \cap R; E] 
		      = \gibbs[R^{\comp}][\currentPointset{t} \cap R](A)
        \]
	\end{claim}
	
	Given \Cref{claim:invariant_case0}, we aim for changing the probability measure from $\PrSymbol[E]$ to $\PrSymbol[E_0]$.
	To achieve this, note that by definition $\bayesFilter{t}$ does only depend on $\currentPointset{t}$ via $\currentPointset{t} \cap R$.
	Thus, it holds that
	\begin{align*}
		\Pr{\bayesFilter{t} = 0}[\currentPointset{t} \cap R; E]
		&= \Pr{\bayesFilter{t} = 0}[\currentPointset{t} \cap R, \currentPointset{t} \cap R^{\comp}; E] .
	\end{align*}
	Therefore, \Cref{lemma:adding_conditions} implies
	\begin{align*}
		\Pr{\currentPointset{t} \cap R^{\comp} \in A}[\currentPointset{t} \cap R; E_0]
		&= \Pr{\currentPointset{t} \cap R^{\comp} \in A}[\currentPointset{t} \cap R; E] \\
		&= \gibbs[R^{\comp}][\currentPointset{t} \cap R](A) .
	\end{align*}
	Finally, observe that, given $E_0$, it holds that $\currentPointset{t+1} = \currentPointset{t}$.
	Thus, we have
	\[
		\Pr{\currentPointset{t+1} \cap R^{\comp} \in A}[\currentPointset{t+1} \cap R; E_0] 
		=\gibbs[R^{\comp}][\currentPointset{t+1} \cap R](A) 
	\]
	as desired.
\end{proof}

Before using \Cref{lemma:invariant_case1,lemma:invariant_case0} to show \Cref{lemma:invariant}, we first prove \Cref{claim:invariant_case0}.
\begin{proof}[Proof of \Cref{claim:invariant_case0}]
	Define $f: \statespace \to [0, 1], \omega \mapsto \gibbs[R^{\comp}][\currentPointset{t}(\omega) \cap R](A)$.
	Since $f$ is $\sigma(\currentPointset{t} \cap R)$-measurable, it suffices to prove that
	\begin{align} \label{lemma:invariant_case0:eq1}
		\E{\ind{\currentPointset{t} \cap R \in D} \cdot f}[][E] 
		= \E{\ind{\currentPointset{t} \cap R \in D} \cdot \ind{\currentPointset{t} \cap R^{\comp} \in A}}[][E]
	\end{align}
	for all events $D \in \pointsetEvents[R]$.
	Further, since $R$ is the union of the disjoint sets $S$ and $\partial \boxesToUpdate \setminus S$, \Cref{lemma:pi_system} allows us to focus on events of the form $\{\currentPointset{t} \cap (\partial \boxesToUpdate \setminus S) \in D_1, \currentPointset{t} \cap S \in D_2\}$ for $D_1 \in \pointsetEvents[\partial \boxesToUpdate \setminus S]$ and $D_2 \in \pointsetEvents[S]$.
	Writing 
	\begin{align*}
		\E{\ind{\currentPointset{t} \cap (\partial \boxesToUpdate \setminus S) \in D_{1}} \cdot \ind{\currentPointset{t} \cap S \in D_{2}} \cdot f}[][E] 
		&= \E{\ind{\currentPointset{t} \cap S \in D_{2}} \cdot \E{\ind{\currentPointset{t} \cap (\partial \boxesToUpdate \setminus S) \in D_{1}} \cdot f}[\currentPointset{t} \cap S][E]}[][E]
	\end{align*}
	shows that, for proving (\ref{lemma:invariant_case0:eq1}), it suffices to show that
	\begin{align*} 
		\E{\ind{\currentPointset{t} \cap (\partial \boxesToUpdate \setminus S) \in D_{1}} \cdot f}[\currentPointset{t} \cap S][E] 
		&= \E{\ind{\currentPointset{t} \cap (\partial \boxesToUpdate \setminus S) \in D_{1}} \cdot \ind{\currentPointset{t} \cap R^{\comp} \in A}}[\currentPointset{t} \cap S][E] .
	\end{align*}
	To this end, define $\hat{f}: \pointsets[\partial \boxesToUpdate \setminus S] \times \pointsets[S] \to [0, 1]$ by
	\[
		\hat{f}(\eta_1, \eta_2) = \gibbs[R^{\comp}][\eta_1 \cup \eta_2](A) 
	\]
	and note that $\hat{f}(\currentPointset{t}(\omega) \cap (\partial \boxesToUpdate \setminus S), \currentPointset{t}(\omega) \cap S) = f(\omega)$ for all $\omega \in \statespace$.
	Therefore, we have
	\begin{align*}
		&\E{\ind{\currentPointset{t} \cap (\partial \boxesToUpdate \setminus S) \in D_{1}} \cdot f}[\currentPointset{t} \cap S][E] \\
		&\quad\quad = \E{\ind{\currentPointset{t} \cap (\partial \boxesToUpdate \setminus S) \in D_{1}} \cdot \hat{f}(\currentPointset{t} \cap (\partial \boxesToUpdate \setminus S), \currentPointset{t} \cap S)}[\currentPointset{t} \cap S][E] .
	\end{align*}
	We proceed by deriving an explicit expression for the right-hand side.
	By \Cref{lemma:invariant_projection} and the assumption of the lemma, we know that
	\[
	   (\omega, \cdot) \mapsto \gibbs[S^{\comp}][\currentPointset{t}(\omega) \cap S] \circ \projection{\partial \boxesToUpdate \setminus S}^{-1}(\cdot) 
	\]
	is regular conditional distribution for $\currentPointset{t} \cap (\partial \boxesToUpdate \setminus S)$ given $\currentPointset{t} \cap S$ under $\PrSymbol[E]$. 
	Thus, using \Cref{thm:integration_RCP}, we obtain 
	\begin{align*}
		&\E{\ind{\currentPointset{t} \cap (\partial \boxesToUpdate \setminus S) \in D_{1}} \cdot \hat{f}(\currentPointset{t} \cap (\partial \boxesToUpdate \setminus S), \currentPointset{t} \cap S)}[\currentPointset{t} \cap S][E](\omega) \\
		&\quad= \int_{D_{1}} \hat{f}(\eta, \currentPointset{t}(\omega) \cap S) \,\gibbs[S^{\comp}][\currentPointset{t}(\omega) \cap S] \circ \projection{\partial \boxesToUpdate \setminus S}^{-1}(\intD \eta) \\
		&\quad= \frac{1}{\partitionFunction[S^{\comp}][\currentPointset{t}(\omega) \cap S]}\sum_{n \ge 0} \frac{1}{n!} \int_{(\boxRegion{\partial \boxesToUpdate \setminus S})^n} \ind{\setFromTuple{\vectorize{x}} \in D_{1}} \cdot \hat{f}(\setFromTuple{\vectorize{x}}, \currentPointset{t}(\omega) \cap S)\\
		&\hspace*{10em}\cdot (\activity_{\currentPointset{t}(\omega) \cap S})^{\vectorize{x}} \cdot \eulerE^{- \hamiltonian[\vectorize{x}]} \cdot \partitionFunction[(\partial\boxesToUpdate \setminus S)^{\comp} \cap S^{\comp}][\setFromTuple{\vectorize{x}} \cup (\currentPointset{t}(\omega) \cap S)] \intD \vectorize{x} \\
		&\quad= \frac{1}{\partitionFunction[S^{\comp}][\currentPointset{t}(\omega) \cap S]} \sum_{n \ge 0} \frac{1}{n!} \int_{(\boxRegion{\partial \boxesToUpdate \setminus S})^n} \ind{\setFromTuple{\vectorize{x}} \in D_{1}} \cdot \hat{f}(\setFromTuple{\vectorize{x}}, \currentPointset{t}(\omega) \cap S)\\
		&\hspace*{10em}\cdot (\activity_{\currentPointset{t}(\omega) \cap S})^{\vectorize{x}} \cdot \eulerE^{- \hamiltonian[\vectorize{x}]} \cdot \partitionFunction[R^{\comp}][\setFromTuple{\vectorize{x}} \cup (\currentPointset{t}(\omega) \cap S)] \intD \vectorize{x}
	\end{align*}
	for $\PrSymbol[E]$-almost all $\omega \in \statespace$.
	Further, note that
	\begin{align*}
		&\hat{f}(\setFromTuple{\vectorize{x}}, \currentPointset{t}(\omega) \cap S) \\
		&\quad\quad = \gibbs[R^{\comp}][\setFromTuple{\vectorize{x}} \cup (\currentPointset{t}(\omega) \cap S)] \\
		&\quad\quad = \frac{1}{\partitionFunction[R^{\comp}][\setFromTuple{\vectorize{x}} \cup (\currentPointset{t}(\omega) \cap S)]} \sum_{m \ge 0} \frac{1}{m!} \int_{\left(\boxRegion{R}^{\comp}\right)^m} \ind{\setFromTuple{\vectorize{y}} \in A} \cdot (\activity_{\setFromTuple{\vectorize{x}} \cup (\currentPointset{t}(\omega) \cap S)})^{\vectorize{y}}\cdot \eulerE^{-\hamiltonian[\vectorize{y}]} \intD \vectorize{y} .
	\end{align*}
	Thus, after canceling $\partitionFunction[R^{\comp}][\setFromTuple{\vectorize{x}} \cup (\currentPointset{t}(\omega) \cap S)]$, we obtain
	\begin{align*}
		&\E{\ind{\currentPointset{t} \cap (\partial \boxesToUpdate \setminus S) \in D_{1}} \cdot \hat{f}(\currentPointset{t} \cap (\partial \boxesToUpdate \setminus S), \currentPointset{t} \cap S)}[\currentPointset{t} \cap S][E](\omega) \\
		&\quad= \frac{1}{\partitionFunction[S^{\comp}][\currentPointset{t}(\omega) \cap S]} \sum_{n \ge 0} \frac{1}{n!} \int_{(\boxRegion{\partial \boxesToUpdate \setminus S})^n} \ind{\setFromTuple{\vectorize{x}} \in D_{1}} \cdot (\activity_{\currentPointset{t}(\omega) \cap S})^{\vectorize{x}} \cdot \eulerE^{- \hamiltonian[\vectorize{x}]} \\
		&\hspace*{10em}\cdot \left[\sum_{m \ge 0} \frac{1}{m!} \int_{\left(\boxRegion{R}^{\comp}\right)^m} \ind{\setFromTuple{\vectorize{y}} \in A} \cdot (\activity_{\setFromTuple{\vectorize{x}} \cup (\currentPointset{t}(\omega) \cap S)})^{\vectorize{y}} \cdot \eulerE^{-\hamiltonian[\vectorize{y}]} \intD \vectorize{y} \right] \intD \vectorize{x} \\
		&\quad= \gibbs[S^{\comp}][\currentPointset{t}(\omega) \cap S]\left(\projection{\partial \boxesToUpdate \setminus S}^{-1}(D_1) \cap \projection{R^{\comp}}^{-1}(A)\right) \\
		&\quad= \Pr{\currentPointset{t} \cap (\partial \boxesToUpdate \setminus S) \in D_{1}, \currentPointset{t} \cap R^{\comp} \in A}[\currentPointset{t} \cap S; E](\omega)
	\end{align*}
	for $\PrSymbol[E]$-almost all $\omega \in \statespace$, where the last equality is once again due to \Cref{lemma:invariant_projection}.
	This shows (\ref{lemma:invariant_case0:eq1}) and consequently proves the claim.
\end{proof}

We proceed to show \Cref{lemma:invariant} by combining \Cref{lemma:invariant_case1,lemma:invariant_case0}.
\begin{proof}[Proof of \Cref{lemma:invariant}]
	First, note that is suffices to show that, for all $t \in \N_{0}$ and $S \subseteq \boxIds$ with $\Pr{\boxesToFix{t} = S} > 0$, it holds that $\gibbs[S^{\comp}][\currentPointset{t} \cap S](A)$ is a version of the conditional expectation $\Pr{\currentPointset{t} \cap S^{\comp} \in A}[\currentPointset{t} \cap S, \boxesToFix{t} = S]$. 
	The second part of the statement then follows as $\gibbs[S^{\comp}][\eta]$ is a probability distribution for all $\eta \in \pointsets[S]$.
	
	We proceed by proving our claim by an induction over $t$.
	For $t = 0$, the statement is trivially true since $\boxesToFix{0} = \boxIds$ with probability $1$.
	
	Now, assume the lemma holds for some fixed $t \in N$ and let $R \subseteq \boxIds$ be such that $\Pr{\boxesToFix{t+1} = R} > 0$.
	We start with the case $R \neq \emptyset$.
	We consider two sets of event $\mathcal{C}_{0}$ and $\mathcal{C}_1$, where $\mathcal{C}_{i}$ for $i \in \set{0, 1}$ consists of all events of the form $C = \set{\boxesToFix{t} = S, \boxChosen{t} = \vectorize{v}, \bayesFilter{t} = i}$ with $S \in \powerset{\boxIds} \setminus\set{\emptyset}$ and $\vectorize{v} \in S$ such that $\Pr{C \cap \set{\boxesToFix{t+1} = R}} > 0$.
	Set $\mathcal{C} = \mathcal{C}_0 \cup \mathcal{C}_1$ and note that all events in $\mathcal{C}$ are pairwise disjoint.
	Moreover, it is easy to check that
	\begin{align*}
		\Pr{\bigcup\nolimits_{C \in \mathcal{C}} C}[\boxesToFix{t+1} = R] 
		= 1 .
	\end{align*}
	Thus, given we show that
	\begin{align} \label{lemma:invariant:eq1}
		\Pr{\currentPointset{t+1} \cap R^{\comp} \in A}[\currentPointset{t} \cap R; C] 
		= \gibbs[R^{\comp}][\currentPointset{t+1} \cap R](A)
	\end{align}
	for all $C \in \mathcal{C}$, then \Cref{lemma:law_of_total_probability} implies that
	$\gibbs[R^{\comp}][\currentPointset{t+1} \cap R](A)$ is also a version of the conditional expectation $\Pr{\currentPointset{t+1} \cap R^{\comp} \in A}[\currentPointset{t} \cap R; \boxesToFix{t+1} = R]$ as desired.
	Suppose that $C \in \mathcal{C}_1$.
	Then $C$ must have the form $\set{\boxesToFix{t} = S, \boxChosen{t} = \vectorize{v}, \bayesFilter{t} = 1}$ for some $S$ and $\vectorize{v}$ with $R = S \setminus \vectorize{v}$.
	Thus, using the induction hypothesis and applying \Cref{lemma:invariant_case1} proves \eqref{lemma:invariant:eq1}.
	Otherwise, if $C \in \mathcal{C}_0$, then $C$ is of the form $\set{\boxesToFix{t} = S, \boxChosen{t} = \vectorize{v}, \bayesFilter{t} = 0}$ with $R = S \cup \partial \boxesToUpdate[S][\vectorize{v}][\updateRadius]$. 
	Using the induction hypothesis and \Cref{lemma:invariant_case0} shows \eqref{lemma:invariant:eq1}.
	
	It remains to consider the case $R = \emptyset$.
	We construct $\mathcal{C}_0$ and $\mathcal{C}_1$ as before, but we set $\mathcal{C} = \mathcal{C}_0 \cup \mathcal{C}_1 \cup \set{\set{\boxesToFix{t} = \emptyset}}$.
	Again, by \Cref{lemma:law_of_total_probability} it suffices to argue \eqref{lemma:invariant:eq1} for all $C \in \mathcal{C}$.
	The cases $C \in \mathcal{C}_{1}$ and $C \in \mathcal{C}_{0}$ are handled as before and we focus on $C = \set{\boxesToFix{t} = \emptyset}$.
	By our definition of the process, we fixed $\currentPointset{t}$ and $\boxesToFix{t}$ to remain constant once $\boxesToFix{t} = \emptyset$.
	Thus, \eqref{lemma:invariant:eq1} follows directly from the induction hypothesis, concluding the proof. 
\end{proof}

\section{Strong spatial mixing and success probabilities of Bayes filters}
\label{sec:point_densities}
Recall that for \Cref{algo:sampling} to terminate rapidly, we need to ensure that the success probability of the Bayes filter is close to $1$.
In this section, we prove a general statement that allows us to control the success probabilities of the Bayes filters  we will construct in the upcoming sections.
Readers only interested in the actual construction of the Bayes filter may skip this section for now and return to it later for the running time analysis. 

The main technical lemma of this section states that, under strong spatial mixing, a certain fraction of partition functions that is central for the construction of our Bayes filters can be brought arbitrarily close to $1$ by increasing the update radius $\updateRadius$.
\begin{lemma} \label{lemma:ssm_partition_function}
	Let $S \subseteq \boxIds$ be non-empty, $\vectorize{v} \in S$, $\updateRadius \in \N$, and set $\boxesToUpdate = \boxesToUpdate[S][\vectorize{v}][\updateRadius]$ and $H = \complementOf{S \cup \boxesToUpdate}$.
	Suppose $\potential$ exhibits $(a, b)$-strong spatial mixing up to $\activity$.
	Then, for all feasible $\eta \in \pointsets[\region]$ and all $\xi_1, \xi_2 \in \pointsets[H]$ it holds that 
	\begin{align*}
		\exponential{- a 3^{\dimensions} \range^{\dimensions} \eulerE^{2 b \range} \left(\activity \range^{\dimensions} + \eulerE^{\activity 3^{\dimensions} \range^{\dimensions}}\right) \eulerE^{-b \range \updateRadius}} 
		&\le
		\frac{ 
			\partitionFunction[\boxesToUpdate \setminus \vectorize{v}][\xi_1 \cup (\eta \cap S)]
		}{
			\partitionFunction[\boxesToUpdate][\xi_1 \cup (\eta \cap (S \setminus \vectorize{v}))]
		}
		\cdot
		\frac
		{
			\partitionFunction[\boxesToUpdate][\xi_2 \cup (\eta \cap (S \setminus \vectorize{v}))]
		}{ 
			\partitionFunction[\boxesToUpdate \setminus \vectorize{v}][\xi_2 \cup (\eta \cap S)]
		} \\
		&\le 
		\exponential{a 3^{\dimensions} \range^{\dimensions} \eulerE^{2 b \range} \left(\activity \range^{\dimensions} + \eulerE^{\activity 3^{\dimensions} \range^{\dimensions}}\right) \eulerE^{-b \range \updateRadius}} .
	\end{align*}
\end{lemma}

Before we prove \Cref{lemma:ssm_partition_function}, we briefly sketch how  it helps  control the success probability of the Bayes filter. Recall \Cref{def:bayes_filter} and assume we would directly use 
\begin{align} \label{eq:example_correction}
	\bayesFilterCorrection[S][\vectorize{v}][\eta] 
	= \inf_{\substack{\xi \in \pointsets[H]\\ \xi \cup (\eta \cap S) \text{ is feasible}}}  \left\{\frac{\partitionFunction[\boxesToUpdate \setminus \vectorize{v}][\xi \cup (\eta \cap S)]}{\partitionFunction[\boxesToUpdate][\xi \cup (\eta \cap (S \setminus \vectorize{v}))]}\right\}
\end{align}
as Bayes filter correction.
Assuming $\currentPointset{t} = \eta$, $\boxesToFix{t} = S$ and $\boxChosen{t} = \vectorize{v} \in S$, \Cref{lemma:spatial_markov} yields that the probability that $\bayesFilter{t} = 1$ is 
\begin{align*}
	&\inf_{\substack{\xi \in \pointsets[H]\\ \xi \cup (\eta \cap S) \text{ is feasible}}}  \left\{\frac{
        \partitionFunction[\boxesToUpdate \setminus \vectorize{v}][\xi \cup (\eta \cap S)]
    }{
        \partitionFunction[\boxesToUpdate][\xi \cup (\eta \cap (S \setminus \vectorize{v}))]}
    \right\} \cdot \frac{
		\partitionFunction[\boxesToUpdate][\eta \cap \partial \boxesToUpdate] 
	}{ 
		\partitionFunction[\boxesToUpdate \setminus \vectorize{v}][\eta \cap (\partial \boxesToUpdate \cup \vectorize{v})]
	} \\
	&\quad\quad= \inf_{\substack{\xi \in \pointsets[H]\\ \xi \cup (\eta \cap S) \text{ is feasible}}}  \left\{\frac{
        \partitionFunction[\boxesToUpdate \setminus \vectorize{v}][\xi \cup (\eta \cap S)]
    }{
        \partitionFunction[\boxesToUpdate][\xi \cup (\eta \cap (S \setminus \vectorize{v}))]}
	\cdot
	\frac
	{
		\partitionFunction[\boxesToUpdate][(\eta \cap H) \cup (\eta \cap (S \setminus \vectorize{v}))]
	}{ 
		\partitionFunction[\boxesToUpdate \setminus \vectorize{v}][(\eta \cap H) \cup (\eta \cap S)]
	} \right\} .
\end{align*}
Applying \Cref{lemma:ssm_partition_function} with $\xi_1 = \xi$ and $\xi_2 = \eta \cap H$ allows us to lower bound the probability that $\bayesFilter{t} = 1$ by $\exponential{- a 3^{\dimensions} \range^{\dimensions} \eulerE^{2 b \range} \left(\activity \range^{\dimensions} + \eulerE^{\activity 3^{\dimensions} \range^{\dimensions}}\right) \eulerE^{-b \range \updateRadius}}$.
Thus, by increasing the update radius $\updateRadius$ we could bring the success probability of the Bayes filter arbitrary close to $1$.
While we will not use exactly \eqref{eq:example_correction} as Bayes filter correction, we can apply \Cref{lemma:ssm_partition_function} in a similar fashion when using a suitable approximation.
More on that in \Cref{secHardSphere} and \Cref{secRepulsive}.

To prove \Cref{lemma:ssm_partition_function}, we first show that strong spatial mixing implies correlation decay in terms of $k$-point density functions.
The converse of this statement was previously shown in \cite{michelen2022strong}.
We then use an identity from \cite{michelen2022strong,mp-CC} to derive our lemma.

\subsection{Strong spatial mixing and point density functions}
Let $\region \subset \R^{\dimensions}$ be a bounded measurable region and $\activityFunction$ be an activity function.
For every $k \in \N$ the $k$-point density of the Gibbs point process $\gibbs[\activityFunction, \region]$ at $\vectorize{x} \in \region^{k}$ is defined as
\[
	\density[\activityFunction][\vectorize{x}] = \activityFunction^{\vectorize{x}} \frac{\partitionFunction[\region](\activityFunction_{\vectorize{x}})}{\partitionFunction[\region](\activityFunction)} \eulerE^{-\hamiltonian[\vectorize{x}]} ,
\]
where $\activityFunction_{\vectorize{x}}$ denotes the activity function $y \mapsto \activityFunction(y) \eulerE^{-\sum_{i \in [k]} \potential(y, x_i)}$.

Recall \Cref{def:ssmHS} and note that, for a measurable space $(\Omega, \mathcal{A})$ and probability measures $P$ and $Q$, an equivalent definition of  total variation distance is
\[
	\dtv{P}{Q} = \sup_{f: \Omega \to [-1, 1]}\absolute{\E{f}[][][P] - \E{f}[][][Q]},
\]
where the supremum is taken over measurable functions.
Using this definition we obtain the following statement.

\begin{lemma} \label{lemma:point_density_ssm_additive}
	Let $\potential$ be a repulsive pair potential of range $\range$ and let $\activity, a, b \in \R_{>0}$ such that $\potential$ exhibits $(a, b)$-strong spatial mixing up to $\activity$.
	Let $\region \subset \R^{\dimensions}$ be bounded and measurable, and consider activity functions $\activityFunction, \activityFunction' < \activity$.
	For all measurable $\subregion \subseteq \region$ disjoint from $\supp{\activityFunction - \activityFunction'}$, all $k \in \N$ and all $\vectorize{x} \in \subregion^{k}$ it holds that
	\[
	\absolute{\density[\activityFunction][\vectorize{x}] - \density[\activityFunction'][\vectorize{x}]} \le \activityFunction^{\vectorize{x}} \eulerE^{-\hamiltonian[\vectorize{x}]} a \volume{\subregion^{(\range)}} \eulerE^{b \range} \eulerE^{- b \cdot \dist{\subregion}{\supp{\activityFunction - \activityFunction'}}} ,
	\]
	where $\subregion^{(\range)} = \set{y \in \region}[\dist{y}{\subregion} < \range]$, and $\density[\activityFunction]$ and $\density[\activityFunction']$ are the $k$-point densities of $\gibbs[\activityFunction, \region]$ and $\gibbs[\activityFunction', \region]$.
\end{lemma}

\begin{proof}
	For every $k \in \N$ and $\vectorize{x} = (x_1, \dots, x_k) \in \subregion^{k}$ define $f_{\vectorize{x}}: \pointsets[\region] \to [0, 1]$ by
	\[
	f_{\vectorize{x}}(\eta) = \eulerE^{-\sum_{i \in [k]} \sum_{y \in \eta} \potential[x_i][y]}.
	\]
	By definition, it holds that
	\begin{align*}
		\density[\activityFunction][\vectorize{x}] 
		&= \activityFunction^{\vectorize{x}} \frac{\partitionFunction[\region](\activityFunction_{\vectorize{x}})}{\partitionFunction[\region](\activityFunction)} \eulerE^{-\hamiltonian[\vectorize{x}]} \\
		&= \activityFunction^{\vectorize{x}} \eulerE^{-\hamiltonian[\vectorize{x}]}
		\E{f_{\vectorize{x}}}[][][\gibbs[\activityFunction, \region]] .
	\end{align*} 
	Since the range of $\potential$ is bounded by $\range$, it holds that $f_{\vectorize{x}}$ is local on $\subregion^{(\range)} = \set{y \in \region}[\dist{y}{\subregion} < \range]$.
    Applying the same reasoning to $\density[\activityFunction'][\vectorize{x}]$ and observing that $\activityFunction = \activityFunction'$ on $\subregion$ yields
    \[
        \density[\activityFunction'][\vectorize{x}] 
		= \activityFunction^{\vectorize{x}} \eulerE^{-\hamiltonian[\vectorize{x}]}
		\E{f_{\vectorize{x}}}[][][\gibbs[\activityFunction', \region]].
    \]
	Next, note that for every $\eta \in \pointsets[\region]$ it holds that $f_{\vectorize{x}}(\eta) = f_{\vectorize{x}} \circ \projection{\subregion^{(\range)}} (\eta)$, where $\projection{\subregion^{(\range)}}$ is the projection $\eta \mapsto \eta \cap \subregion^{(\range)}$.
    Using the change-of-variables formula for Lebesgue integration, we get
	\begin{align*}
		\absolute{\density[\activityFunction][\vectorize{x}] - \density[\activityFunction'][\vectorize{x}]}
		&= \activityFunction^{\vectorize{x}} \eulerE^{-\hamiltonian[\vectorize{x}]} \absolute{\E{f_{\vectorize{x}}}[][][\gibbs[\activityFunction, \region] \circ \projection{\subregion^{(\range)}}^{-1}] - \E{f_{\vectorize{x}}}[][][\gibbs[\activityFunction', \region] \circ \projection{\subregion^{(\range)}}^{-1}]} \\
		&\le \activityFunction^{\vectorize{x}} \eulerE^{-\hamiltonian[\vectorize{x}]} \dtvProjected{\gibbs[\activityFunction, \region]}{\gibbs[\activityFunction', \region]}{\subregion^{(\range)}} ,
	\end{align*} 
	where the inequality follows from the definition of the total variation distance given above and the fact that $f_{\vectorize{x}}$ has domain $[0, 1]$.
	Finally, applying $(a, b)$-strong spatial mixing and noting that $\dist{\subregion^{(\range)}}{\supp{\activityFunction - \activityFunction'}} \ge \dist{\subregion}{\supp{\activityFunction - \activityFunction'}} - \range$ concludes the proof.
\end{proof}

\begin{remark}
	Note that, without fixing a particular region $\subregion \subseteq \region$ that contains $x_1, \dots, x_k$ in \Cref{lemma:point_density_ssm_additive}, we can always set $\subregion = \{x_1, \dots, x_k\}$, which yields $\volume{\subregion^{(\range)}} \le k v_{\dimensions} \range^{\dimensions}$ with $v_{\dimensions}$ being the volume of a unit ball in $\dimensions$ dimensions.
\end{remark}

The following multiplicative bound for $k$-point densities with different activity functions follows immediately.
\begin{corollary}\label{cor:point_density_ssm}
	Consider the setting of \Cref{lemma:point_density_ssm_additive}. 
	It holds that
	\[
	\density[\activityFunction][\vectorize{x}] \le \left(1 + a \volume{\subregion^{(\range)}} \eulerE^{b \range + \activity \volume{\subregion^{(\range)}}} \eulerE^{- b \cdot \dist{\subregion}{\supp{\activityFunction - \activityFunction'}}}\right)\density[\activityFunction'][\vectorize{x}] .
	\]
\end{corollary}

\begin{proof}
	Since $\activityFunction = \activityFunction'$ on $\subregion$, it holds that $\activityFunction^{\vectorize{x}} \eulerE^{-\hamiltonian[\vectorize{x}]} = \activityFunction'^{\vectorize{x}} \eulerE^{-\hamiltonian[\vectorize{x}]}$.
	If $\activityFunction^{\vectorize{x}} \eulerE^{-\hamiltonian[\vectorize{x}]} = 0$, then the desired inequality holds trivially since both sides are $0$.
	
	Assume $\activityFunction^{\vectorize{x}} \eulerE^{-\hamiltonian[\vectorize{x}]} > 0$. 
	Defining $f_{\vectorize{x}}$ as in the proof of \Cref{lemma:point_density_ssm_additive} and following the same arguments we have
	\[
	\density[\activityFunction'][\vectorize{x}] 
	= \activityFunction'^{\vectorize{x}} \eulerE^{-\hamiltonian[\vectorize{x}]} \E{f_{\vectorize{x}}}[][][\gibbs[\activityFunction', \region] \circ \projection{\subregion^{(\range)}}^{-1}] 
	= \activityFunction^{\vectorize{x}} \eulerE^{-\hamiltonian[\vectorize{x}]} \E{f_{\vectorize{x}}}[][][\gibbs[\activityFunction', \region] \circ \projection{\subregion^{(\range)}}^{-1}] .
	\]
	Observe that $f_{\vectorize{x}}$ is non-negative and $f_{\vectorize{x}}(\emptyset) = 1$.
	Combined with Poisson domination, we obtain
	\[
	   \E{f_{\vectorize{x}}}[][][\gibbs[\activityFunction', \region] \circ \projection{\subregion^{(\range)}}^{-1}] 
	   \ge \gibbs[\activityFunction', \region] \circ \projection{\subregion^{(\range)}}^{-1}(\set{\emptyset})
	   \ge \eulerE^{-\activity \volume{\subregion^{(\range)}}},
	\]
	which implies $\density[\activityFunction'][\vectorize{x}] \ge \activityFunction^{\vectorize{x}} \eulerE^{-\hamiltonian[\vectorize{x}]} \eulerE^{-\activity \volume{\subregion^{(\range)}}}$.
	In particular, we have $\density[\activityFunction'][\vectorize{x}] > 0$, and applying \Cref{lemma:point_density_ssm_additive} yields
	\begin{align*}
		\density[\activityFunction][\vectorize{x}] 
		&\le \left(1 + \frac{\absolute{\density[\activityFunction][\vectorize{x}] - \density[\activityFunction'][\vectorize{x}]}}{\density[\activityFunction'][\vectorize{x}]}\right)\density[\activityFunction'][\vectorize{x}] \\
		&\le \left(1 + a \volume{\subregion^{(\range)}} \eulerE^{b \range + \activity \volume{\subregion^{(\range)}}} \eulerE^{- b \cdot \dist{\subregion}{\supp{\activityFunction - \activityFunction'}}}\right)\density[\activityFunction'][\vectorize{x}] . \qedhere
	\end{align*}
\end{proof}

\subsection{Proof of \Cref{lemma:ssm_partition_function}}
To prove the main lemma of the section, we use the following identity by Michelen and Perkins \cite{michelen2022strong,mp-CC}. 
\begin{lemma}[{\cite[Lemma 12]{michelen2022strong}}] \label{lemma:density_to_correlation}
	Let $\subregion \subseteq \region$ be measurable.
	Fix a point in $z \in \subregion$ and, for any given activity function $\activityFunction$ and any point $y \in \region$, let 
	\[
	\widehat{\activityFunction}_{y}(w) = \begin{cases}
		0 \text{ if } \widetilde{\distSymbol}(z, w) < \widetilde{\distSymbol}(z, y) \\
		\activityFunction[w] \text{ otherwise} 
	\end{cases} ,
	\]
	where $\widetilde{\distSymbol}(u, v) = \dist{u}{v} + diam(\subregion) \cdot \ind{\set{u, v} \nsubseteq \subregion}$ for $u, v \in \R^{\dimensions}$. 
	For all $k \in \N$ and $\vectorize{x} \in \subregion^k$, it holds that
	\[
	\activityFunction^{\vectorize{x}}\frac{\partitionFunction[\region \setminus \subregion](\activityFunction_{\vectorize{x}})}{\partitionFunction[\region](\activityFunction)} \eulerE^{- \hamiltonian[\vectorize{x}]} 
	= \density[\activityFunction][\vectorize{x}] \exponential{-\int_{\subregion} \density[\widehat{(\activityFunction_{\vectorize{x}})}_{y}][y] \intD y}.
	\]
	Moreover, we have
	\[
	\frac{\partitionFunction[\region \setminus \subregion](\activityFunction)}{\partitionFunction[\region](\activityFunction)}  
	=  \exponential{-\int_{\subregion} \density[\widehat{\activityFunction}_{y}][y] \intD y}.
	\]
\end{lemma} 

We use  \Cref{cor:point_density_ssm} to show the following intermediate statement.
\begin{lemma}\label{lemma:density_ssm}
	Consider the setting of \Cref{lemma:point_density_ssm_additive}. 
	For all $k \in \N_{0}$ and all $\vectorize{x} \in \subregion^{k}$ it holds that
	\[
	\activityFunction^{\vectorize{x}}\frac{\partitionFunction[\region \setminus \subregion](\activityFunction_{\vectorize{x}})}{\partitionFunction[\region](\activityFunction)} \eulerE^{- \hamiltonian[\vectorize{x}]} 
	\le \exponential{a \volume{\subregion^{(\range)}} e^{b \range} \left(\activity \volume{\subregion} + \eulerE^{\activity \volume{\subregion^{(\range)}}}\right) \eulerE^{-b \cdot \dist{\subregion}{\supp{\activityFunction - \activityFunction'}}}} \activityFunction'^{\vectorize{x}}\frac{\partitionFunction[\region \setminus \subregion](\activityFunction'_{\vectorize{x}})}{\partitionFunction[\region](\activityFunction')} \eulerE^{- \hamiltonian[\vectorize{x}]},
	\]
	where $\subregion^{(\range)} = \set{y \in \region}[\dist{y}{\subregion} < \range]$, and $\partitionFunction[\region](\activityFunction)$ and $\partitionFunction[\region](\activityFunction')$ are the partition functions on $\region$ for the potential $\potential$.
\end{lemma}

\begin{proof}
	Using \Cref{lemma:density_to_correlation} we have
	\[
	   \activityFunction^{\vectorize{x}}\frac{\partitionFunction[\region \setminus \subregion](\activityFunction_{\vectorize{x}})}{\partitionFunction[\region](\activityFunction)} \eulerE^{- \hamiltonian[\vectorize{x}]} 
	   = \density[\activityFunction][\vectorize{x}] \exponential{-\int_{\subregion} \density[\widehat{(\activityFunction_{\vectorize{x}})}_{y}][y] \intD y} ,
	\]
	where we treat case that $\vectorize{x}$ is the empty tuple by setting $\density[\activityFunction][\vectorize{x}] = 1$.
	By \Cref{cor:point_density_ssm}, we have
	\begin{align*}
		\density[\activityFunction][\vectorize{x}] 
		&\le \left(1 + a \volume{\subregion^{(\range)}} \eulerE^{b \range + \activity \volume{\subregion^{(\range)}}} \eulerE^{- b \cdot \dist{\subregion}{\supp{\activityFunction - \activityFunction'}}}\right)\density[\activityFunction'][\vectorize{x}] \\
		&\le \exponential{a \volume{\subregion^{(\range)}} \eulerE^{b \range + \activity \volume{\subregion^{(\range)}}} \eulerE^{- b \cdot \dist{\subregion}{\supp{\activityFunction - \activityFunction'}}}} \density[\activityFunction'][\vectorize{x}].
	\end{align*}
	Moreover, observe that $\supp{\widehat{(\activityFunction_{\vectorize{x}})}_{y} - \widehat{(\activityFunction'_{\vectorize{x}})}_{y}} \subseteq \supp{\activityFunction - \activityFunction'}$, $\widehat{(\activityFunction_{\vectorize{x}})}_{y} \le \activityFunction < \activity$ and $\widehat{(\activityFunction'_{\vectorize{x}})}_{y} \le \activityFunction' < \activity $ for all $y \in \subregion$.
	Thus, \Cref{lemma:point_density_ssm_additive} implies 
	\begin{align*}
		-\int_{\subregion} \density[\widehat{(\activityFunction_{\vectorize{x}})}_{y}][y] \intD y
		&\le -\int_{\subregion} \density[\widehat{(\activityFunction'_{\vectorize{x}})}_{y}][y] \intD y + \int_{\subregion} \absolute{\density[\widehat{(\activityFunction_{\vectorize{x}})}_{y}][y] -  \density[\widehat{(\activityFunction'_{\vectorize{x}})}_{y}][y]} \intD y \\
		&\le -\int_{\subregion} \density[\widehat{(\activityFunction'_{\vectorize{x}})}_{y}][y] \intD y + a \volume{\subregion^{(\range)}} \eulerE^{b \range} \eulerE^{- b \cdot \dist{\subregion}{\supp{\activityFunction - \activityFunction'}}} \cdot \int_{\subregion} \widehat{(\activityFunction'_{\vectorize{x}})}_{y}(y)  \intD y \\
		&\le -\int_{\subregion} \density[\widehat{(\activityFunction'_{\vectorize{x}})}_{y}][y] \intD y +  a \volume{\subregion^{(\range)}} \eulerE^{b \range} \eulerE^{- b \cdot \dist{\subregion}{\supp{\activityFunction - \activityFunction'}}} \cdot \activity \volume{\subregion} .
	\end{align*}
	We conclude that
	\begin{align*}
&\activityFunction^{\vectorize{x}}\frac{\partitionFunction[\region \setminus \subregion](\activityFunction_{\vectorize{x}})}{\partitionFunction[\region](\activityFunction)} \eulerE^{- \hamiltonian[\vectorize{x}]} \\
		&\quad\quad\le \exponential{a \volume{\subregion^{(\range)}} e^{b \range} \left(\activity \volume{\subregion} + \eulerE^{\activity \volume{\subregion^{(\range)}}}\right) \eulerE^{-b \cdot \dist{\subregion}{\supp{\activityFunction - \activityFunction'}}}} \density[\activityFunction'][\vectorize{x}] \exponential{-\int_{\subregion} \density[\widehat{(\activityFunction'_{\vectorize{x}})}_{y}][y] \intD y} .
	\end{align*} 
	Applying \Cref{lemma:density_to_correlation} again concludes the proof.
\end{proof}

With \Cref{lemma:density_ssm}, we can  prove \Cref{lemma:ssm_partition_function}.
\begin{proof}[Proof of \Cref{lemma:ssm_partition_function}]
	We aim to apply \Cref{lemma:density_ssm}.
	To this end, we start by writing the involved partition functions in terms of two new activity functions $\activityFunction, \activityFunction'$.
	
	Set $\activityFunction = \activity_{\xi_1 \cup (\eta \cap (S \setminus \vectorize{v}))} \ind{\boxRegion{\boxesToUpdate}}$ and $\activityFunction' = \activity_{\xi_2 \cup (\eta \cap (S \setminus \vectorize{v}))} \ind{\boxRegion{\boxesToUpdate}}$, and observe that  $\activityFunction, \activityFunction' \le \activity$.
	Set $k = \size{\eta \cap \vectorize{v}}$ and let  $\vectorize{x} \in \boxRegion{\vectorize{v}}^{k}$ be such that $\setFromTuple{\vectorize{x}} = \eta \cap \vectorize{v}$ (i.e., any tuple containing exactly the points in $\eta \cap \vectorize{v}$).
	Note that 
	\[
	   \activity_{\xi_1 \cup (\eta \cap S)} \ind{\boxRegion{\boxesToUpdate \setminus \vectorize{v}}}
	   = \activityFunction_{\eta \cap \boxRegion{\vectorize{v}}} \ind{\region \setminus \boxRegion{\vectorize{v}}}
	   = \activityFunction_{\vectorize{x}} \ind{\region \setminus \boxRegion{\vectorize{v}}}
	\]  
	and analogously $\activity_{\xi_2 \cup (\eta \cap S)} \ind{\boxRegion{\boxesToUpdate \setminus \vectorize{v}}} = \activityFunction'_{\vectorize{x}} \ind{\region \setminus \boxRegion{\vectorize{v}}}$.
	Thus, we obtain
	\begin{align} \label{lemma:ssm_partition_function:eq1}
		\frac{ 
			\partitionFunction[\boxesToUpdate \setminus \vectorize{v}][\xi_1 \cup (\eta \cap S)]
		}{
			\partitionFunction[\boxesToUpdate][\xi_1 \cup (\eta \cap (S \setminus \vectorize{v}))]
		}
		\cdot
		\frac
		{
			\partitionFunction[\boxesToUpdate][\xi_2 \cup (\eta \cap (S \setminus \vectorize{v}))]
		}{ 
			\partitionFunction[\boxesToUpdate \setminus \vectorize{v}][\xi_2 \cup (\eta \cap S)]
		}
		= \frac{
			\partitionFunction[\region](\activityFunction') 
			\partitionFunction[\region \setminus \boxRegion{\vectorize{v}}](\activityFunction_{\vectorize{x}})
		}{ 
			\partitionFunction[\region](\activityFunction) 
			\partitionFunction[\region \setminus \boxRegion{\vectorize{v}}](\activityFunction'_{\vectorize{x}})
		}
	\end{align}
	We proceed by lower bounding the distance between $\supp{\activityFunction - \activityFunction'}$ and $\boxRegion{\vectorize{v}}$.
	Note that $\xi_1 \cup (\eta \cap (S \setminus \vectorize{v}))$ and $\xi_2 \cup (\eta \cap (S \setminus \vectorize{v}))$ agree on $\boxRegion{S \setminus \vectorize{v}}$ and can only disagree on $\boxRegion{H}$. 
	By construction, it holds that $\dist{\boxRegion{H}}{\boxRegion{\vectorize{v}}}\ge \updateRadius \range$.
	As the potential range is bounded by $\range$, it follows that $\dist{\supp{\activityFunction - \activityFunction'}}{\boxRegion{\vectorize{v}}} \ge (\updateRadius - 1) \range$.
	
	Now, note that in particular $\activityFunction$ and $\activityFunction'$ agree on $\boxRegion{\vectorize{v}}$ and $\activityFunction^{\vectorize{x}} = \activityFunction'^{\vectorize{x}}$.
	We may assume $\activity>0$, since otherwise all involved partition functions are $1$ and the statement holds trivially.
	Since further $\eta$ is feasible we have $\activityFunction^{\vectorize{x}} \eulerE^{- \hamiltonian[\vectorize{x}]} = \activityFunction'^{\vectorize{x}} \eulerE^{- \hamiltonian[\vectorize{x}]} > 0$.
	Thus, multiplying  (\ref{lemma:ssm_partition_function:eq1}) with $\frac{\activityFunction^{\vectorize{x}} \eulerE^{- \hamiltonian[\vectorize{x}]}}{\activityFunction'^{\vectorize{x}} \eulerE^{- \hamiltonian[\vectorize{x}]}} = 1$ yields
	\begin{align*}
		\frac{ 
			\partitionFunction[\boxesToUpdate \setminus \vectorize{v}][\xi_1 \cup (\eta \cap S)]
		}{
			\partitionFunction[\boxesToUpdate][\xi_1 \cup (\eta \cap (S \setminus \vectorize{v}))]
		}
		\cdot
		\frac
		{
			\partitionFunction[\boxesToUpdate][\xi_2 \cup (\eta \cap (S \setminus \vectorize{v}))]
		}{ 
			\partitionFunction[\boxesToUpdate \setminus \vectorize{v}][\xi_2 \cup (\eta \cap S)]
		}
		& = \frac{
			\activityFunction^{\vectorize{x}}
			\partitionFunction[\region](\activityFunction') 
			\partitionFunction[\region \setminus \boxRegion{\vectorize{v}}](\activityFunction_{\vectorize{x}})
			\eulerE^{- \hamiltonian[\vectorize{x}]}
		}{ 
			\activityFunction'^{\vectorize{x}}
			\partitionFunction[\region](\activityFunction) 
			\partitionFunction[\region \setminus \boxRegion{\vectorize{v}}](\activityFunction'_{\vectorize{x}})
			\eulerE^{- \hamiltonian[\vectorize{x}]}
		} \\
		&= \frac{ 
			\activityFunction^{\vectorize{x}}
			\frac{\partitionFunction[\region \setminus \boxRegion{\vectorize{v}}](\activityFunction'_{\vectorize{x}})}{\partitionFunction[\region](\activityFunction)}
			\eulerE^{- \hamiltonian[\vectorize{x}]}
		}{
			\activityFunction'^{\vectorize{x}}
			\frac{\partitionFunction[\region \setminus \boxRegion{\vectorize{v}}](\activityFunction'_{\vectorize{x}})}{\partitionFunction[\region](\activityFunction')}
			\eulerE^{- \hamiltonian[\vectorize{x}]}
		} .
	\end{align*}
	Finally, by \Cref{lemma:density_ssm} we have the upper bound
	\[
	\frac{ 
		\partitionFunction[\boxesToUpdate \setminus \vectorize{v}][\xi_1 \cup (\eta \cap S)]
	}{
		\partitionFunction[\boxesToUpdate][\xi_1 \cup (\eta \cap (S \setminus \vectorize{v}))]
	}
	\cdot
	\frac
	{
		\partitionFunction[\boxesToUpdate][\xi_2 \cup (\eta \cap (S \setminus \vectorize{v}))]
	}{ 
		\partitionFunction[\boxesToUpdate \setminus \vectorize{v}][\xi_2 \cup (\eta \cap S)]
	}
	\le \exponential{a 3^{\dimensions} \range^{\dimensions} \eulerE^{2 b \range} \left(\activity \range^{\dimensions} + \eulerE^{\activity 3^{\dimensions} \range^{\dimensions}}\right) \eulerE^{-b \range \updateRadius}}, 
	\]
	and applying the same reasoning after swapping the roles of $\activityFunction$ and $\activityFunction'$ results in the corresponding lower bound, which proves the claim.
\end{proof}

\section{Hard-sphere model}
\label{secHardSphere}
In this section we focus on the  hard-sphere model.
Recall that for an interaction range $\range>0$, the hard-sphere model is defined by the potential 
\[
	\potentialHS[x][y] = 
	\begin{cases}
		\infty \text{ if } \dist{x}{y} < \range \\
		0 \text{ otherwise }
	\end{cases} .
\]

To simplify notation, define for every $\vectorize{x} \in \R^{k}$ and $\vectorize{y} \in \R^{m}$ 
\begin{align*}
	\validHS[\vectorize{x}] &= \eulerE^{-\sum_{\set{i, j} \in \binom{[k]}{2}} \potentialHS[x_i][x_j]} = \prod_{\set{i, j} \in \binom{[k]}{2}} \ind{\dist{x_i}{x_j} \ge \range} \text{ and } \\
	\validHS[\vectorize{x}][\vectorize{y}] &= \eulerE^{-\sum_{i \in [k]} \sum_{j \in [m]} \potentialHS[x_i][y_j]} = \prod_{i \in [k]} \prod_{j \in [m]} \ind{\dist{x_i}{y_j} \ge \range} .
\end{align*}
We extend this definition from tuples of points to finite point sets in the obvious way.
This allows us to write 
\[
	\partitionFunction[\subregion][\eta](\activity) = \sum_{k \ge 0} \frac{\activity^k}{k!} \int_{\subregion^k} \validHS[\vectorize{x}] \validHS[\vectorize{x}][\eta] \intD \vectorize{x},
\]
for all measurable $\subregion \subseteq \region$ and all $\eta \in \pointsets[\region]$.

\subsection{Constructing the Bayes filter}
We start by constructing a suitable Bayes filter correction for the hard-sphere model.
The key ingredient will be  computing such a correction by enumerating a finite set of boundary configurations that closely approximates all possible boundary conditions.
This is made precise by the following lemma.

\begin{lemma} \label{lemma:approximate_boundary}
	Let $S \subseteq \boxIds$ be non-empty, $\vectorize{v} \in S$, $\eta \in \pointsets[\region]$, $\boxesToUpdate = \boxesToUpdate[S][\vectorize{v}][\updateRadius]$ and $H = (S \cup \boxesToUpdate)^{\comp}$.
	For all $\varepsilon > 0$, $\xi \in \pointsets[H]$ and 
	\[
		\delta \le \varepsilon \cdot \left(\size{\xi \cap \partial \boxesToUpdate} 2^{\dimensions} \range^{\dimensions - 1} \dimensions^{3/2} \activity \eulerE^{\activity (2\updateRadius + 1)^{\dimensions} \range^{\dimensions}}\right)^{-1}
	\]
	there is some $\gamma \subseteq (\delta \Z)^{\dimensions} \cap \boxRegion{H \cap \partial \boxesToUpdate}$ such that
	\begin{align*}
		\eulerE^{-\varepsilon} \cdot \frac{\partitionFunction[\boxesToUpdate \setminus \vectorize{v}][\xi \cup (\eta \cap S)]}{\partitionFunction[\boxesToUpdate][\xi \cup (\eta \cap (S \setminus \vectorize{v}))]}
		\le \frac{\partitionFunction[\boxesToUpdate \setminus \vectorize{v}][\gamma \cup (\eta \cap S)]}{\partitionFunction[\boxesToUpdate][\gamma \cup (\eta \cap (S \setminus \vectorize{v}))]}
		\le \eulerE^{\varepsilon} \cdot \frac{\partitionFunction[\boxesToUpdate \setminus \vectorize{v}][\xi \cup (\eta \cap S)]}{\partitionFunction[\boxesToUpdate][\xi \cup (\eta \cap (S \setminus \vectorize{v}))]} .
	\end{align*} 
\end{lemma}   

\begin{proof}
	Fix some $\xi \in \pointsets[H]$, let $\Phi$ map every point in $\boxRegion{H \cap \partial \boxesToUpdate}$ to its closest point in $(\delta \Z)^{\dimensions} \cap \boxRegion{H \cap \partial \boxesToUpdate}$ in $\ell_{\infty}$-distance and set $\gamma = \set{\Phi(x) \mid x \in \xi \cap \partial \boxesToUpdate}$.
	We first prove that
	\[
		\eulerE^{-\frac{\varepsilon}{2}} \cdot \partitionFunction[\boxesToUpdate \setminus \vectorize{v}][\xi \cup (\eta \cap S)]
		\le \partitionFunction[\boxesToUpdate \setminus \vectorize{v}][\gamma \cup (\eta \cap S)]
		\le \eulerE^{\frac{\varepsilon}{2}} \cdot \partitionFunction[\boxesToUpdate \setminus \vectorize{v}][\xi \cup (\eta \cap S)] .
	\]
	 To this end, note that by \Cref{lemma:spatial_markov} 
	 \[
	 	\partitionFunction[\boxesToUpdate \setminus \vectorize{v}][\xi \cup (\eta \cap S)]
	 	= 
	 	\partitionFunction[\boxesToUpdate \setminus \vectorize{v}][(\xi \cap \partial \boxesToUpdate) \cup (\eta \cap S)] .
	 \] 
	 Thus, we have
	 \begin{align*}
	 	&\absolute{\partitionFunction[\boxesToUpdate \setminus \vectorize{v}][\xi \cup (\eta \cap S)] - \partitionFunction[\boxesToUpdate \setminus \vectorize{v}][\gamma \cup (\eta \cap S)]} \\
	 	& \hspace{4em} \le \sum_{k \ge 1} \frac{\activity^{k}}{k!} \int_{(\boxRegion{\boxesToUpdate \setminus \vectorize{v}})^k} \validHS[\vectorize{x}] \cdot \absolute{\validHS[\vectorize{x}][(\xi \cap \partial \boxesToUpdate) \cup (\eta \cap S)] - \validHS[\vectorize{x}][\gamma \cup (\eta \cap S)]} \intD \vectorize{x} \\
	 	& \hspace{4em} \le \sum_{k \ge 1} \frac{\activity^{k}}{k!} \int_{(\boxRegion{\boxesToUpdate \setminus \vectorize{v}})^k} \absolute{\validHS[\vectorize{x}][(\xi \cap \partial \boxesToUpdate) \cup (\eta \cap S)] - \validHS[\vectorize{x}][\gamma \cup (\eta \cap S)]} \intD \vectorize{x} .
	 \end{align*}
 	Next, observe that for every $\vectorize{x} \in (\boxRegion{\boxesToUpdate \setminus \vectorize{v}})^k$ it holds that 
    \[
        \validHS[\vectorize{x}][(\xi \cap \partial \boxesToUpdate) \cup (\eta \cap S)] \neq \validHS[\vectorize{x}][\gamma \cup (\eta \cap S)]
    \] 
    implies that there is some $i \in [k]$ and some $y \in \xi \cap \partial \boxesToUpdate$ such that either $\dist{x_i}{y} \ge \range > \dist{x_i}{\Phi(y)}$ or $\dist{x_i}{y} < \range \le \dist{x_i}{\Phi(y)}$.
 	Further, note that for every $y \in \boxRegion{H \cap \partial \boxesToUpdate}$ (and particular $y \in \xi \cap \partial \boxesToUpdate$) it holds that $\dist{y}{\Phi(y)} \le \sqrt{\dimensions} \delta$.
 	Using union bound and observing that $\validHS[\cdot][(\xi \cap \partial \boxesToUpdate) \cup (\eta \cap S)]$ and $\validHS[\cdot][\gamma \cup (\eta \cap S)]$ are symmetric functions, we obtain
 	\begin{align*}
 		&\absolute{\partitionFunction[\boxesToUpdate \setminus \vectorize{v}][\xi \cup (\eta \cap S)] - \partitionFunction[\boxesToUpdate \setminus \vectorize{v}][\gamma \cup (\eta \cap S)]} \\
 		&\hspace{4em}\le \size{\xi \cap \partial \boxesToUpdate} \cdot \left[(\range + \sqrt{\dimensions} \delta)^{\dimensions} - (\range + \sqrt{\dimensions} \delta)^{\dimensions}\right] \cdot \sum_{k \ge 1} \frac{\activity^k}{k!} k \volume{\boxRegion{\boxesToUpdate \setminus \vectorize{v}}}^{k-1} \\
 		&\hspace{4em}\le \size{\xi \cap \partial \boxesToUpdate} \cdot \left[(\range + \sqrt{\dimensions} \delta)^{\dimensions} - (\range + \sqrt{\dimensions} \delta)^{\dimensions}\right] \cdot \activity \eulerE^{\activity \volume{\boxRegion{\boxesToUpdate \setminus \vectorize{v}}}}.
 	\end{align*}
 	Elementary calculations yield
 	\[
 		\left[(\range + \sqrt{\dimensions} \delta)^{\dimensions} - (\range + \sqrt{\dimensions} \delta)^{\dimensions}\right] 
 		\le \dimensions^{3/2} \delta (\range + \sqrt{\dimensions} \delta)^{\dimensions - 1}
 		\le \dimensions^{3/2} \delta 2^{\dimensions - 1} \range^{\dimensions - 1}.
 	\]
 	Further, it holds that
 	\[
 		\volume{\boxRegion{\boxesToUpdate \setminus \vectorize{v}}} 
 		\le \volume{\boxRegion{\boxesToUpdate}} 
 		\le (2\updateRadius + 1)^{\dimensions} \range^{\dimensions} .
 	\]
 	Thus, for our choice of $\delta$ we obtain
 	\[
 		\absolute{\partitionFunction[\boxesToUpdate \setminus \vectorize{v}][\xi \cup (\eta \cap S)] - \partitionFunction[\boxesToUpdate \setminus \vectorize{v}][\gamma \cup (\eta \cap S)]} \le \frac{\varepsilon}{2}, 
 	\]
 	and, since 
 	\[
 		\min\set{\partitionFunction[\boxesToUpdate \setminus \vectorize{v}][\xi \cup (\eta \cap S)], \partitionFunction[\boxesToUpdate \setminus \vectorize{v}][\gamma \cup (\eta \cap S)]} \ge 1,
 	\]
 	this proves the desired multiplicative bound.
 	
 	It remains to show
 	\[
 		\eulerE^{-\frac{\varepsilon}{2}} \cdot \partitionFunction[\boxesToUpdate][\xi \cup (\eta \cap (S \setminus \vectorize{v}))]
 		\le \partitionFunction[\boxesToUpdate][\gamma \cup (\eta \cap (S \setminus \vectorize{v}))]
 		\le \eulerE^{\frac{\varepsilon}{2}} \cdot \partitionFunction[\boxesToUpdate][\xi \cup (\eta \cap (S \setminus \vectorize{v}))],
 	\]
 	which is done analogously, concluding the proof.
\end{proof}

In particular, we obtain the following corollary.
\begin{corollary}\label{cor:approximate_boundary}
	Let $S \subseteq \boxIds$ be non-empty, $\vectorize{v} \in S$, $\eta \in \pointsets[\region]$, $\boxesToUpdate = \boxesToUpdate[S][\vectorize{v}][\updateRadius]$ and $H = (S \cup \boxesToUpdate)^{\comp}$.
	For all $\varepsilon > 0$ and
	\[
		\delta \le \varepsilon \cdot \left(4^{\dimensions} \range^{-1} \dimensions^{(\dimensions + 3)/2} (2\updateRadius + 3)^{\dimensions} \range^{\dimensions} \activity \eulerE^{\activity (2\updateRadius + 1)^{\dimensions} \range^{\dimensions}}\right)^{-1}
	\]
	it holds that
	\begin{align*}
		\eulerE^{-\varepsilon} \cdot \min_{\gamma \subseteq (\delta \Z)^{\dimensions} \cap \boxRegion{H \cap \partial \boxesToUpdate}} \frac{\partitionFunction[\boxesToUpdate \setminus \vectorize{v}][\gamma \cup (\eta \cap S)]}{\partitionFunction[\boxesToUpdate][\gamma \cup (\eta \cap (S \setminus \vectorize{v}))]}
		\le \inf_{\substack{\xi \in \pointsets[H]\\ \xi \cup (\eta \cap S) \text{ is feasible}}}  \frac{\partitionFunction[\boxesToUpdate \setminus \vectorize{v}][\xi \cup (\eta \cap S)]}{\partitionFunction[\boxesToUpdate][\xi \cup (\eta \cap (S \setminus \vectorize{v}))]} .
	\end{align*}
\end{corollary} 

\begin{proof}
	Note that, if $\xi \in \pointsets[H]$ is such that $\xi \cup (\eta \cap S)$ is feasible, this implies in particular that $\xi' = \xi \cap \partial \boxesToUpdate$ is feasible.
	Thus, the claim follows from \Cref{lemma:approximate_boundary} by arguing that every feasible configuration $\xi' \in \pointsets[H \cap \partial \boxesToUpdate]$ satisfies $\size{\xi'} \le \left(\frac{2 \sqrt{\dimensions}}{\range}\right)^{\dimensions} (2 \updateRadius + 3)^{\dimensions} \range^{\dimensions}$.
	To see this, note that every point in $\xi'$ blocks at least a volume of $\left(\frac{\range}{2 \sqrt{\dimensions}}\right)^{\dimensions}$, where no other point can be placed. 
	Moreover, it holds that
	\[
		\volume{\boxRegion{H \cap \partial \boxesToUpdate}} \le \volume{\boxRegion{\boxesToUpdate \cup \partial \boxesToUpdate}} 
		\le (2 \updateRadius + 3)^{\dimensions} \range^{\dimensions},
	\]
	which concludes the proof.
\end{proof}

\Cref{cor:approximate_boundary} allows us to replace the minimization over the uncountable set of boundary conditions $\pointsets[H]$ by a minimization over the finite set $(\delta \Z)^{\dimensions} \cap \boxRegion{H \cap \partial \boxesToUpdate}$.
The second ingredient that we need for computing a suitable Bayes filter correction is a way to approximate the involved partition functions.

To this end, for every non-empty $S \subseteq \boxIds$, $\eta \in \pointsets[\region]$ and $\delta > 0$, define
\begin{align} \label{eq:approx_HS}
	\partitionFunctionApprox[S][\eta][\delta] = \sum_{\gamma \subseteq (\delta \Z)^{\dimensions} \cap \boxRegion{S}} \activity^{\size{\gamma}} \delta^{\size{\gamma}} \cdot \validHS[\gamma] \cdot \validHS[\gamma][\eta \cap \partial S] .
\end{align}
The follow lemma justifies using $\partitionFunctionApprox[S][\eta][\delta]$ as an approximation for the hard-sphere partition function $\partitionFunction[S][\eta \cap S^{\comp}]$, given that $\delta$ is chosen sufficiently small.

\begin{lemma}\label{lemma:approx_hardspheres}
	Let $S \subseteq \boxIds$ be non-empty and $\eta \in \pointsets[\region]$.
	For all $\varepsilon > 0$ and 
	\[
		\delta \le \varepsilon \left[\dimensions^{3/2} 2^{\dimensions} \max\set{\range^{-1}, \range^{\dimensions-1}} \max\set{\activity, \activity^2} (\size{\eta \cap \partial S} + \volume{\boxRegion{S \cup \partial S}}) \eulerE^{\activity \volume{\boxRegion{S \cup \partial S}}} \right]^{-1}
	\]
	it holds that
	\[
		\eulerE^{-\varepsilon} \cdot \partitionFunction[S][\eta \cap S^{\comp}] 	
		\le \partitionFunctionApprox[S][\eta][\delta]
		\le \eulerE^{\varepsilon} \cdot \partitionFunction[S][\eta \cap S^{\comp}] ,
	\]
	where $\partitionFunctionApprox[S][\eta][\delta]$ is defined as in (\ref{eq:approx_HS}).
\end{lemma}

\begin{proof}
	Define $\boxRegion{S, \delta} = \bigcup_{x \in  \subseteq (\delta \Z)^{\dimensions} \cap \boxRegion{S}} \ball{x}{\delta/2}[\infty]$, where $\ball{x}{\delta/2}[\infty]$ is the closed $\delta/2$-ball around $x$ in infinity norm.
	Moreover, let $\Phi$ map every point in $\boxRegion{S, \delta}$ to its closest point in $(\delta \Z)^{\dimensions} \cap \boxRegion{S}$ in $\ell_{\infty}$-distance, breaking ties arbitrarily.
	With some abuse of notation, we extend $\Phi$ to tuples $\vectorize{x} \in \boxRegion{S, \delta}$ by setting $\Phi(\vectorize{x}) = (\Phi(x_1), \dots, \Phi(x_k))$.
	Now, note that
	\begin{align*}
		\partitionFunctionApprox[S][\eta][\delta] 
		&= \sum_{k \ge 0} \frac{\activity^{k}}{k!} \delta^{k} \sum_{\vectorize{x} \in (\delta \Z)^{\dimensions k} \cap \boxRegion{S}^k} \validHS[\vectorize{x}] \cdot \validHS[\vectorize{x}][\eta \cap \partial S] \\
		&= \sum_{k \ge 0} \frac{\activity^{k}}{k!} \int_{\boxRegion{S, \delta}^k} \validHS[\Phi(\vectorize{x})] \cdot \validHS[\Phi(\vectorize{x})][\eta \cap \partial S] \intD \vectorize{x} ,\\
	\end{align*}
	where the first equality uses the fact $\validHS[\vectorize{x}] = 0$ whenever $\vectorize{x}$ contains the same point more than once.
	
	We proceed by relating $\partitionFunctionApprox[S][\eta][\delta]$ to $\partitionFunction[S][\eta \cap S^{\comp}]$ in two steps. 
	First, we compare $\partitionFunction[S][\eta \cap S^{\comp}]$ with $\partitionFunction[\boxRegion{S, \delta}][\eta \cap \partial S]$, and then we compare $\partitionFunction[\boxRegion{S, \delta}][\eta \cap \partial S]$ with $\partitionFunctionApprox[S][\eta][\delta]$ using the expression above.
	
	For the first part, note that by \Cref{lemma:spatial_markov} it holds that $\partitionFunction[S][\eta \cap S^{\comp}] = \partitionFunction[S][\eta \cap \partial S]$.
	Moreover, we have
	\begin{align*}
		\partitionFunction[S][\eta \cap \partial S]
		&\le \partitionFunction[\boxRegion{S} \setminus \boxRegion{S, \delta}][\eta \cap \partial S] \cdot \partitionFunction[\boxRegion{S} \cap \boxRegion{S, \delta}][\eta \cap \partial S] \\
		&\le \eulerE^{\activity \volume{\boxRegion{S} \symmDiff \boxRegion{S, \delta}}} \cdot \partitionFunction[\boxRegion{S, \delta}][\eta \cap \partial S] ,
	\end{align*}
	where $\symmDiff$ denotes the symmetric difference.
	Analogously, it holds that
	\[
		\partitionFunction[\boxRegion{S, \delta}][\eta \cap \partial S]
		\le \eulerE^{\activity \volume{\boxRegion{S} \symmDiff \boxRegion{S, \delta}}} \cdot  \partitionFunction[S][\eta \cap \partial S] .
	\]
	Thus, if we show that $\activity \volume{\boxRegion{S} \symmDiff \boxRegion{S, \delta}} \le \varepsilon/2$ for our choice of $\delta$, then 
	\begin{align}
		\eulerE^{-\varepsilon/2}\partitionFunction[S][\eta \cap S^{\comp}]
		\le \partitionFunction[\boxRegion{S, \delta}][\eta \cap \partial S]
		\le \eulerE^{\varepsilon/2}\partitionFunction[S][\eta \cap S^{\comp}] \label{eq:approx_HS:bound1}.
	\end{align}
	To this end, note that, if $x \in \boxRegion{S, \delta} \setminus \boxRegion{S}$, then $x \notin \boxRegion{S}$ but $\dist{x}{\boxRegion{S}} \le \frac{\sqrt{\dimensions} \delta}{2}$.
	Similarly, if $x \in \boxRegion{S} \setminus \boxRegion{S, \delta}$, then $x \in \boxRegion{S}$ but $\dist{x}{\boxRegion{S}^{\comp}} \le \frac{\sqrt{\dimensions} \delta}{2}$.
	Taking the union bound over boxes $\vectorize{v} \in S$ yields
	\begin{align*}
		\volume{\boxRegion{S} \symmDiff \boxRegion{S, \delta}} 
		&\le \size{S} \cdot \left[(\range + \sqrt{\dimensions} \delta /2)^{\dimensions} - (\range - \sqrt{\dimensions} \delta /2)^{\dimensions}\right] \\
		&\le \size{S}  \dimensions^{3/2} \delta \cdot (\range + \sqrt{\dimensions} \delta /2)^{\dimensions - 1} \\
		&\le \volume{\boxRegion{S}} \dimensions^{3/2} 2^{\dimensions-1} \range^{-1} \delta .
	\end{align*}
	Thus, for $\delta \le \varepsilon \left(\activity \volume{\boxRegion{S}} \dimensions^{3/2} 2^{\dimensions} \range^{- 1}\right)^{-1}$ the desired inequality is satisfied.
	
	We proceed by relating $\partitionFunction[\boxRegion{S, \delta}][\eta \cap \partial S]$ to $\partitionFunctionApprox[S][\eta][\delta]$.
	First, note that
	\begin{align*}
		&\absolute{\partitionFunction[\boxRegion{S, \delta}][\eta \cap \partial S] - \partitionFunctionApprox[S][\eta][\delta]} \\
		&\le \sum_{k \ge 0} \frac{\activity^{k}}{k!} \int_{\boxRegion{S, \delta}^k} \absolute{\validHS[\vectorize{x}] \cdot \validHS[\vectorize{x}][\eta \cap \partial S] - \validHS[\Phi(\vectorize{x})] \cdot \validHS[\Phi(\vectorize{x})][\eta \cap \partial S]} \intD \vectorize{x} \\
		&\le \sum_{k \ge 2} \frac{\activity^{k}}{k!} \int_{\boxRegion{S, \delta}^k} \absolute{\validHS[\vectorize{x}] - \validHS[\Phi(\vectorize{x})]} \intD \vectorize{x} 
		+ \sum_{k \ge 1} \frac{\activity^{k}}{k!} \int_{\boxRegion{S, \delta}^k} \absolute{ \validHS[\vectorize{x}][\eta \cap \partial S] -  \validHS[\Phi(\vectorize{x})][\eta \cap \partial S]} \intD \vectorize{x}.
	\end{align*}
	We bound each of the terms in this sum separately.
	To this end, note that for any $\vectorize{x} \in \boxRegion{S, \delta}^k$ it holds that $\validHS[\vectorize{x}] \neq \validHS[\Phi(\vectorize{x})]$ implies that there are $i<j$ such that either $\dist{x_i}{x_j} < \range \le \dist{\Phi(x_i)}{\Phi(x_j)}$ or $\dist{x_i}{x_j} \ge \range > \dist{\Phi(x_i)}{\Phi(x_j)}$.
	Since $\dist{x}{\Phi(x)} \le \frac{\sqrt{\dimensions} \delta}{2}$, applying union bound over $1 \le i < j \le k$ yields
	\begin{align*}
		\sum_{k \ge 2} \frac{\activity^{k}}{k!} \int_{\boxRegion{S, \delta}^k} \absolute{\validHS[\vectorize{x}] - \validHS[\Phi(\vectorize{x})]} \intD \vectorize{x}
		&\le \left[(\range + \sqrt{\dimensions} \delta)^{\dimensions} - (\range - \sqrt{\dimensions} \delta)^{\dimensions}\right] \sum_{k \ge 2} \frac{\activity^{k}}{k!} k (k-1) \volume{\boxRegion{S, \delta}}^{k-1} \\
		&\le \dimensions^{3/2} \delta \cdot (\range + \sqrt{\dimensions} \delta)^{\dimensions-1} \activity^2 \volume{\boxRegion{S, \delta}} \eulerE^{\activity \volume{\boxRegion{S, \delta}}} \\
		&\le \dimensions^{3/2} \delta 2^{\dimensions-1} \range^{\dimensions-1} \activity^2 \volume{\boxRegion{S, \delta}} \eulerE^{\activity \volume{\boxRegion{S, \delta}}} .
	\end{align*}
	Similarly, we have 
	\begin{align*}
		\sum_{k \ge 1} \frac{\activity^{k}}{k!} \int_{\boxRegion{S, \delta}^k} \absolute{ \validHS[\vectorize{x}][\eta \cap \partial S] -  \validHS[\Phi(\vectorize{x})][\eta \cap \partial S]} \intD \vectorize{x}
		\le \dimensions^{3/2} \delta 2^{\dimensions-1} \range^{\dimensions-1} \activity \size{\eta \cap \partial S} \eulerE^{\activity \volume{\boxRegion{S, \delta}}}.
	\end{align*}
	Combining both and noting that $\boxRegion{S, \delta} \subseteq \bigcup_{x \in \boxRegion{S}} \ball{x}{\delta/2}[\infty] \subseteq \boxRegion{S \cup \partial S}$ yields
	\[
		\absolute{\partitionFunction[\boxRegion{S, \delta}][\eta \cap \partial S] - \partitionFunctionApprox[S][\eta][\delta]}
		\le  \dimensions^{3/2} \delta 2^{\dimensions-1} \range^{\dimensions - 1} \max\set{\activity, \activity^2} (\size{\eta \cap \partial S}+\volume{\boxRegion{S \cup \partial S}}) \eulerE^{\activity \volume{\boxRegion{S \cup \partial S}}} .
	\]
	For $\delta \le \varepsilon \left(\dimensions^{3/2} 2^{\dimensions} \range^{\dimensions-1} \max\set{\activity, \activity^2} (\size{\eta \cap \partial S + \volume{\boxRegion{S \cup \partial S}}}) \eulerE^{\activity \volume{\boxRegion{S \cup \partial S}}}\right)^{-1}$ this gives
	\[
		\absolute{\partitionFunction[\boxRegion{S, \delta}][\eta \cap \partial S] - \partitionFunctionApprox[S][\eta][\delta]} \le \varepsilon/2 ,
	\]
	and, since $\partitionFunction[\boxRegion{S, \delta}][\eta \cap \partial S] \ge 1$ and $\partitionFunctionApprox[S][\eta][\delta] \ge 1$, 
	\[
		\eulerE^{-\varepsilon/2} \cdot \partitionFunction[\boxRegion{S, \delta}][\eta \cap \partial S]
		\le \partitionFunctionApprox[S][\eta][\delta]
		\le \eulerE^{\varepsilon/2} \cdot \partitionFunction[\boxRegion{S, \delta}][\eta \cap \partial S] .
	\]
	Combining this with (\ref{eq:approx_HS:bound1}) concludes the proof.
\end{proof}

We now combine \Cref{cor:approximate_boundary} and \Cref{lemma:approx_hardspheres} to obtain our Bayes filter correction for the hard-sphere model.

\begin{lemma} \label{lemma:correction_HS}
	For $\varepsilon > 0$, non-empty $S \subseteq \boxIds$, $\vectorize{v} \in S$ and feasible $\eta \in \pointsets[\region]$ set
	\begin{align*}
		\delta_1 &= \delta_1(\varepsilon) \coloneqq \frac{\varepsilon}{2} \cdot \left(4^{\dimensions} \range^{-1} \dimensions^{(\dimensions + 3)/2} (2\updateRadius + 3)^{\dimensions} \range^{\dimensions} \activity \eulerE^{\activity (2\updateRadius + 1)^{\dimensions} \range^{\dimensions}}\right)^{-1} \text{ and } \\
		\delta_2 &= \delta_2(\varepsilon) \coloneqq \frac{\varepsilon}{4} \cdot \left(\dimensions^{3/2} 2^{\dimensions} \max\set{\range^{-1}, \range^{\dimensions-1}} \max\set{\activity, \activity^2} m \eulerE^{\activity (2\updateRadius + 3)^{\dimensions} \range^{\dimensions}} \right)^{-1},
	\end{align*}
	where $m = 2^{\dimensions}(2\updateRadius + 3)^{\dimensions} \range^{\dimensions} \left(\delta_1^{-\dimensions} + \dimensions^{\dimensions/2} \range^{-\dimensions}\right) + (2\updateRadius + 3)^{\dimensions} \range^{\dimensions}$, and define
	\[
		\bayesFilterCorrectionHS{\varepsilon}[S][\vectorize{v}][\eta] \coloneqq \eulerE^{-\varepsilon} \cdot \min_{\gamma \subseteq (\delta_1 \Z)^{\dimensions} \cap \boxRegion{H \cap \partial \boxesToUpdate}} 
			\frac{\partitionFunctionApprox[\boxesToUpdate \setminus \vectorize{v}][\gamma \cup (\eta \cap S)][\delta_2]}
			{\partitionFunctionApprox[\boxesToUpdate][\gamma \cup (\eta \cap \boxRegion{S \setminus \vectorize{v}})][\delta_2]} ,
	\]
	where $\boxesToUpdate = \boxesToUpdate[S][\vectorize{v}][\updateRadius]$ and $H = (S \cup \boxesToUpdate)^{\comp}$.
	Then $\bayesFilterCorrectionHS{\varepsilon}[S][\vectorize{v}][\eta]$ is a Bayes filter correction as in \Cref{def:bayes_filter}.
\end{lemma}

\begin{proof}
	We start by arguing that $\bayesFilterCorrectionHS{\varepsilon}$ is a Bayes filter correction.
	For the measurability, note that for every fixed non-empty $S \subseteq \boxIds$ and $\vectorize{v} \in S$ it holds that $\bayesFilterCorrectionHS{\varepsilon}[S][\vectorize{v}][\cdot]$ is a minimum of a finite set of $\pointsetEvents$-measurable functions. 
	Moreover, it can be easily seen that $\bayesFilterCorrectionHS{\varepsilon}[S][\vectorize{v}][\eta] = \bayesFilterCorrectionHS{\varepsilon}[S][\vectorize{v}][\eta']$ for every two configurations $\eta, \eta' \in \pointsets$ that agree on $\boxRegion{S}$.
	
	Next, we argue that for all feasible $\eta \in \pointsets[\region]$ it holds that $\bayesFilterCorrectionHS{\varepsilon}[S][\vectorize{v}][\eta]$ is bounded away from $0$ and 
	\[
		\bayesFilterCorrectionHS{\varepsilon}[S][\vectorize{v}][\eta] 
		\le \inf_{\substack{\xi \in \pointsets[H]\\ \xi \cup (\eta \cap S) \text{ is feasible}}}  \frac{\partitionFunction[\boxesToUpdate \setminus \vectorize{v}][\xi \cup (\eta \cap S)]}{\partitionFunction[\boxesToUpdate][\xi \cup (\eta \cap (S \setminus \vectorize{v}))]} .
	\]
	For the lower bound, note that for all $\gamma \subseteq (\delta_1 \Z)^{\dimensions} \cap \boxRegion{H \cap \partial \boxesToUpdate}$ it holds that
	\[
		\frac{\partitionFunctionApprox[\boxesToUpdate \setminus \vectorize{v}][\gamma \cup (\eta \cap S)][\delta_2]}
		{\partitionFunctionApprox[\boxesToUpdate][\gamma \cup (\eta \cap (S \setminus \vectorize{v}))][\delta_2]}
		\ge \left(1 + \activity \delta_2\right)^{- \size{(\delta_2 \Z)^{\dimensions} \cap \boxRegion{\boxesToUpdate}}} > 0
	\]  
	independent of $\eta$.
	For the upper bound, we start by observing that, for our choice of $\delta_1$, \Cref{cor:approximate_boundary} yields
	\begin{align*}
		\eulerE^{-\varepsilon/2} \cdot \min_{\gamma \subseteq (\delta_1 \Z)^{\dimensions} \cap \boxRegion{H \cap \partial \boxesToUpdate}} \frac{\partitionFunction[\boxesToUpdate \setminus \vectorize{v}][\gamma \cup (\eta \cap S)]}{\partitionFunction[\boxesToUpdate][\gamma \cup (\eta \cap (S \setminus \vectorize{v}))]}
		\le \inf_{\substack{\xi \in \pointsets[H]\\ \xi \cup (\eta \cap S) \text{ is feasible}}} \frac{\partitionFunction[\boxesToUpdate \setminus \vectorize{v}][\xi \cup (\eta \cap S)]}{\partitionFunction[\boxesToUpdate][\xi \cup (\eta \cap (S \setminus \vectorize{v}))]} .
	\end{align*}
	Note that $\volume{\boxRegion{\boxesToUpdate \setminus \vectorize{v}}} \le \volume{\boxRegion{\boxesToUpdate}} \le (2\updateRadius + 1)^{\dimensions} \range^{\dimensions}$ and $\volume{\boxRegion{\boxesToUpdate \setminus \vectorize{v} \cup \partial(\boxesToUpdate \setminus \vectorize{v})}} \le \volume{\boxRegion{\boxesToUpdate \cup \partial \boxesToUpdate}} \le (2 \updateRadius + 3)^{\dimensions} \range^{\dimensions}$.
	Moreover, we have the crude bound $\size{(\delta_1 \Z)^{\dimensions} \cap \boxRegion{H \cap \partial \boxesToUpdate}} \le 2^{\dimensions} \delta_1^{-\dimensions} (2 \updateRadius + 3)^{\dimensions} \range^{\dimensions}$, and, for every feasible $\eta \in \pointsets[\region]$, it holds that $\size{\eta \cap \boxRegion{S \cap \partial (\boxesToUpdate \setminus \vectorize{v})}} \le 2^{\dimensions} \dimensions^{\dimensions/2} \range^{-\dimensions} (2 \updateRadius + 3)^{\dimensions} \range^{\dimensions}$.
	Therefore, for all $\gamma \subseteq (\delta_1 \Z)^{\dimensions} \cap \boxRegion{H \cap \partial \boxesToUpdate}$ we have $\size{(\gamma \cup (\eta \cap S)) \cap \partial (\boxesToUpdate \setminus \vectorize{v})} \le 2^{\dimensions}(2\updateRadius + 3)^{\dimensions} \range^{\dimensions} \left(\delta_1^{-\dimensions} + \dimensions^{\dimensions/2} \range^{-\dimensions}\right)$.
	Analogously, it holds that $\size{(\gamma \cup (\eta \cap (S \setminus \vectorize{v}))) \cap \partial \boxesToUpdate} \le 2^{\dimensions}(2\updateRadius + 3)^{\dimensions} \range^{\dimensions} \left(\delta_1^{-\dimensions} + \dimensions^{\dimensions/2} \range^{-\dimensions}\right)$.
	Thus, \Cref{lemma:approx_hardspheres} yields for our choice of $\delta_2$
	\[
		\eulerE^{-\varepsilon/2} \cdot \frac{\partitionFunction[\boxesToUpdate \setminus \vectorize{v}][\gamma \cup (\eta \cap S)]}{\partitionFunction[\boxesToUpdate][\gamma \cup (\eta \cap (S \setminus \vectorize{v}))]}
		\le \frac{\partitionFunctionApprox[\boxesToUpdate \setminus \vectorize{v}][\gamma \cup (\eta \cap S)][\delta_2]}
		{\partitionFunctionApprox[\boxesToUpdate][\gamma \cup (\eta \cap (S \setminus \vectorize{v}))][\delta_2]}
		\le \eulerE^{\varepsilon/2} \cdot \frac{\partitionFunction[\boxesToUpdate \setminus \vectorize{v}][\gamma \cup (\eta \cap S)]}{\partitionFunction[\boxesToUpdate][\gamma \cup (\eta \cap (S \setminus \vectorize{v}))]}.
	\]
	In particular, this proves
	\[
		\eulerE^{-\varepsilon} \cdot \min_{\gamma \subseteq (\delta_1 \Z)^{\dimensions} \cap \boxRegion{H \cap \partial \boxesToUpdate}} \frac{\partitionFunctionApprox[\boxesToUpdate \setminus \vectorize{v}][\gamma \cup (\eta \cap S)][\delta_2]}{\partitionFunctionApprox[\boxesToUpdate][\gamma \cup (\eta \cap \boxRegion{S \setminus \vectorize{v}})][\delta_2]}
		\le \inf_{\substack{\xi \in \pointsets[H]\\ \xi \cup (\eta \cap S) \text{ is feasible}}}  \frac{\partitionFunction[\boxesToUpdate \setminus \vectorize{v}][\xi \cup (\eta \cap S)]}{\partitionFunction[\boxesToUpdate][\xi \cup (\eta \cap (S \setminus \vectorize{v}))]},
	\]
	implying that $\bayesFilterCorrectionHS{\varepsilon}$ is a Bayes filter correction.
\end{proof}

\subsection{Efficiency of the algorithm}

We now argue that under the assumption of strong spatial mixing we can use $\bayesFilterCorrectionHS{\varepsilon}$ to obtain an efficient implementation of \Cref{algo:sampling}.
Our argument will consist of two steps. 
First, we discuss how to implement each step of the algorithm efficiently.
In particular, we argue that we can efficiently update the configuration (line~\ref{algo:sampling:update}), and that we can efficiently sample a Bayes filter based on $\bayesFilterCorrectionHS{\varepsilon}(\cdot)$ (line~\ref{algo:sampling:filter}).
For the latter, we make use of a Bernoulli factory to circumvent the lack of an algorithm for exact computation of partition functions.
In the second part, we focus on the overall number of iterations of the algorithm.
This is where the assumption of strong spatial mixing comes into play to ensure that the success probability of our Bayes filter is sufficiently large, which implies rapid termination of the algorithm.

We start with discussing the running time of each iteration of \Cref{algo:sampling}.
For updating the configuration, we will use a rejection sampling method which, as long as the updated region $\boxRegion{\boxesToUpdate}$ has constant volume, will be efficient enough for our setting.
Since we apply the same argument for more general repulsive potentials, the following lemma is stated in this general setting.  
\begin{lemma}\label{lemma:exp_perfect_sampling}
	Let $\boxesToUpdate \subseteq \boxIds$ and $\eta \in \pointsets[\region]$. 
	For any repulsive finite-range potential $\potential$ we can sample from $\gibbs[\boxesToUpdate][\eta \cap (\boxesToUpdate)^{\comp}]$ in expected time $(\activity \volume{\boxRegion{\boxesToUpdate}} + \size{\eta \cap \partial \boxesToUpdate}) \cdot \activity \volume{\boxRegion{\boxesToUpdate}} \eulerE^{\activity \volume{\boxRegion{\boxesToUpdate}}}$. 
\end{lemma}

\begin{proof}
	Let $\PPP{\boxesToUpdate}[\activity]$ denote a Poisson point process on $\boxRegion{\boxesToUpdate}$ with intensity $\activity$.
	We consider the rejection sampling algorithm given in \Cref{algo:rejection_sampler}.
	\begin{algorithm}
		\caption{Sample from $\gibbs[\boxesToUpdate][\eta \cap (\boxesToUpdate)^{\comp}]$ }\label{algo:rejection_sampler}
		\Repeat{$W=1$}{
			Draw $Y \sim \PPP{\boxesToUpdate}[\activity]$ \\
			Compute $w = \eulerE^{-\sum_{\set{x, y} \in \binom{Y}{2}} \potential[x][y]} \cdot \eulerE^{-\sum_{x \in \eta \cap \partial \boxesToUpdate, y \in Y} \potential[x][y]}$ \\
			Draw $W \sim \Ber{w}$
		}
		\Return $Y$
	\end{algorithm}
	
	To prove that this rejection sampling method produces the correct out put distribution, it suffices to argue that $w$ as computed in the algorithm is proportional to the density of $\gibbs[\boxesToUpdate][\eta \cap (\boxesToUpdate)^{\comp}]$ with respect to $\PPP{\boxesToUpdate}[\activity]$ for $\PPP{\boxesToUpdate}[\activity]$-almost all $\xi \in \pointsets[\boxesToUpdate]$.
	This is true since for all $\xi \in \pointsets[\boxesToUpdate]$
	\begin{align*}
		\sum_{x \in \eta \cap \boxesToUpdate^{\comp}, y \in \xi} \potential[x][y]
		&= \sum_{x \in \eta \cap \partial \boxesToUpdate, y \in \xi} \potential[x][y] + \sum_{x \in \eta \cap (\boxesToUpdate \cup \partial \boxesToUpdate)^{\comp}, y \in \xi} \potential[x][y] \\
		&= \sum_{x \in \eta \cap \partial \boxesToUpdate, y \in \xi} \potential[x][y],
	\end{align*}
	where the last equality follows from the fact that $\dist{\boxRegion{\boxesToUpdate}}{(\boxRegion{\boxesToUpdate \cup \partial \boxesToUpdate})^{\comp}} \ge \range$ and therefore $\potential[x][y] = 0$ for all $x \in (\boxRegion{\boxesToUpdate \cup \partial \boxesToUpdate})^{\comp}, y \in \boxRegion{\boxesToUpdate}$.

	We proceed by using Wald's identity as given in \Cref{lemma:walds_equation} to bound the expected running time of the procedure above.
	To this end, let $(Y_n)_{n \in \N}$ be a sequence of independent samples from $\PPP{\boxesToUpdate}[\activity]$.
	Assume the algorithm draws $Y = Y_n$ at iteration $n \in \N$.
	Let $S_n$ denote the running time of the rejection sampler in iteration $n$ and let $N$ be the (random) number of iterations until the algorithm terminates.
	We aim for bounding $\E{\sum_{n = 1}^{N} S_n}$.
	
	First, observe that $S_n$ is dominated by the time for computing $w$, implying $S_n \le \size{Y_n}^2 + \size{Y_n} \cdot \size{\eta \cap \boxesToUpdate}$.
	Since further $\size{Y_n}$ follows a Poisson distribution with parameter $\activity \volume{\boxRegion{\boxesToUpdate}}$, we have $\E{S_n} \le \activity^2 \volume{\boxRegion{\boxesToUpdate}}^2 + \activity \volume{\boxRegion{\boxesToUpdate}} \cdot \size{\eta \cap \partial \boxesToUpdate} = (\activity \volume{\boxRegion{\boxesToUpdate}} + \size{\eta \cap \partial \boxesToUpdate}) \cdot \activity \volume{\boxRegion{\boxesToUpdate}}$. 
	
	Moreover, observe that the random variable $\ind{N \ge n}$ only depends on $(Y_i)_{i \le n-1}$, whereas $S_n$ only depends on $Y_n$.
	Therefore, $S_n$ and $\ind{N \ge n}$ are independent and $\E{S_n \ind{N \ge n}} = \E{S_n} \E{\ind{N \ge n}}$.
	
	Applying \Cref{lemma:walds_equation} yields $\E{\sum_{n=1}^{N} S_n} \le (\activity \volume{\boxRegion{\boxesToUpdate}} + \size{\eta \cap \partial \boxesToUpdate}) \cdot \activity \volume{\boxRegion{\boxesToUpdate}} \cdot \E{N}$.
	To obtain a bound $\E{N}$ on, observe that the algorithm always terminates if $Y = \emptyset$, which happens in every iteration independently with a probability of $\eulerE^{- \activity \volume{\boxRegion{\boxesToUpdate}}}$.
	Thus, $N$ is dominated by a geometric random variable with success probability $\eulerE^{- \activity \volume{\boxRegion{\boxesToUpdate}}}$ and $\E{N} \le \eulerE^{\activity \volume{\boxRegion{\boxesToUpdate}}}$, which concludes the proof.
\end{proof}
Note that, in the case of the hard-sphere model with $\range > 0$, if $\eta \in \pointsets[\region]$ is feasible, then $\size{\eta \cap \partial \boxesToUpdate}$ is a linear function of the volume $\volume{\boxRegion{\partial \boxesToUpdate}}$.

We proceed with bounding the running time for sampling the Bayes filter for the hard-sphere model in each step.
To this end, we start with the following observation. 
\begin{observation} \label{obs:compute_correction_HS}
	Consider the setting of \Cref{lemma:correction_HS}.
	The required running time for computing $\bayesFilterCorrectionHS{\varepsilon}[S][\vectorize{v}][\eta]$ does only depend on $\varepsilon$, $\updateRadius$, $\range$, $\activity$ and $\dimensions$.
\end{observation}

This follows directly from enumerating all subsets $\gamma \subseteq (\delta_1 \Z)^{\dimensions} \cap \boxRegion{H \cap \partial \boxesToUpdate}$ and brute-force computation of $\partitionFunctionApprox[\boxesToUpdate \setminus \vectorize{v}][\gamma \cup (\eta \cap S)][\delta_2]$ and $\partitionFunctionApprox[\boxesToUpdate][\gamma \cup (\eta \cap \boxRegion{S \setminus \vectorize{v}})][\delta_2]$, where $\delta_1, \delta_2$ are as in \Cref{lemma:correction_HS}.

In fact, we will not use $\bayesFilterCorrectionHS{\varepsilon}$ directly for our Bayes filter, but rather a slightly scaled version $\eulerE^{-\varepsilon} \bayesFilterCorrectionHS{\varepsilon}$, which is again a Bayes filter correction.
The slack due to the additional scaling allows us to efficiently sample the Bayes filter by using a Bernoulli factory, as we argue in the next lemma.

\begin{lemma} \label{lemma:sampling_bayes_filter_HS}
	Let $S \subseteq \boxIds$ be non-empty, $\vectorize{v} \in S$ and $\eta \in \pointsets[\region]$ be feasible, and set $\boxesToUpdate = \boxesToUpdate[S][\vectorize{v}][\updateRadius]$.
	For all $\varepsilon > 0$ we can sample a Bernoulli random variable with success probability
	\[
	\eulerE^{-\varepsilon}\bayesFilterCorrectionHS{\varepsilon}[S][\vectorize{v}][\eta] \cdot
	\frac{
		\partitionFunction[\boxesToUpdate][\eta \cap \partial \boxesToUpdate] 
	}{ 
		\partitionFunction[\boxesToUpdate \setminus \vectorize{v}][\eta \cap (\partial \boxesToUpdate \cup \vectorize{v})]
	}
	\] 
	with expected running time only depending on $\varepsilon$, $\updateRadius$, $\range$, $\activity$ and $\dimensions$.
\end{lemma}
\begin{proof}
	Our goal is to use a Bernoulli factory of the form $\frac{p}{q}$ to perform this task.
	To bring the desired success probability into such a form, note that
	\begin{align*}
		\partitionFunction[\boxesToUpdate][\eta \cap \partial \boxesToUpdate] &= \left(\gibbs[\boxesToUpdate][\eta \cap \partial \boxesToUpdate](\set{\emptyset})\right)^{-1} \\
		\partitionFunction[\boxesToUpdate \setminus \vectorize{v}][\eta \cap (\partial \boxesToUpdate \cup \vectorize{v})] &= \left(\gibbs[\boxesToUpdate \setminus \vectorize{v}][\eta \cap (\partial \boxesToUpdate \cup \vectorize{v})](\set{\emptyset})\right)^{-1} .
	\end{align*}
	Moreover, note that $0 \le \eulerE^{-\varepsilon}\bayesFilterCorrectionHS{\varepsilon}[S][\vectorize{v}][\eta] \le 1$.
	Thus, by setting $p = \eulerE^{-\varepsilon}\bayesFilterCorrectionHS{\varepsilon}[S][\vectorize{v}][\eta] \cdot \gibbs[\boxesToUpdate \setminus \vectorize{v}][\eta \cap (\partial \boxesToUpdate \cup \vectorize{v})](\set{\emptyset})$ and $q = \gibbs[\boxesToUpdate][\eta \cap \partial \boxesToUpdate](\set{\emptyset})$ we have $p \in [0, 1], q \in [0, 1]$ and 
	\[
		\eulerE^{-\varepsilon}\bayesFilterCorrectionHS{\varepsilon}[S][\vectorize{v}][\eta] \cdot
		\frac{
			\partitionFunction[\boxesToUpdate][\eta \cap \partial \boxesToUpdate] 
		}{ 
			\partitionFunction[\boxesToUpdate \setminus \vectorize{v}][\eta \cap (\partial \boxesToUpdate \cup \vectorize{v})]
		}
		= \frac{p}{q} .
	\]
	
	We are now going to use \Cref{lemma:bernoulli_frac} to prove that we can obtain a sample from $\Ber{\frac{p}{q}}$ within the desired expected running time.
	To this end, we need to provide a positive lower bound on $q - p$ and we need an efficient way for generating independent samples from $\Ber{q}$ and $\Ber{p}$.
	
	For the lower bound, note that by \Cref{lemma:spatial_markov} $\partitionFunction[\boxesToUpdate][\eta \cap \partial \boxesToUpdate] = \partitionFunction[\boxesToUpdate][\eta \cap \boxesToUpdate^{\comp}]$ and $\partitionFunction[\boxesToUpdate \setminus \vectorize{v}][\eta \cap (\partial \boxesToUpdate \cup \vectorize{v})] = \partitionFunction[\boxesToUpdate \setminus \vectorize{v}][\eta \cap (\boxesToUpdate \setminus \vectorize{v})^{\comp}]$. 
	Moreover, since $\bayesFilterCorrectionHS{\varepsilon}[S][\vectorize{v}][\eta]$ is a Bayes filter correction by \Cref{lemma:correction_HS} and $\eta$ is feasible, we have
	\[
		\bayesFilterCorrectionHS{\varepsilon}[S][\vectorize{v}][\eta] \cdot
		\frac{
			\partitionFunction[\boxesToUpdate][\eta \cap \partial \boxesToUpdate] 
		}{ 
			\partitionFunction[\boxesToUpdate \setminus \vectorize{v}][\eta \cap (\partial \boxesToUpdate \cup \vectorize{v})]
		} \le 1 .
	\]
	Consequently, $\frac{p}{q} \le \eulerE^{-\varepsilon}$ and 
	\begin{align*}
		q - p 
		\ge \left(1 - \eulerE^{-\varepsilon}\right) q
		= \left(1 - \eulerE^{-\varepsilon}\right) \cdot(\partitionFunction[\boxesToUpdate][\eta \cap \partial \boxesToUpdate])^{-1}
		\ge \left(1 - \eulerE^{-\varepsilon}\right) \eulerE^{-\activity \volume{\boxRegion{\boxesToUpdate}}} .
	\end{align*}
	Using the upper bound $\volume{\boxRegion{\boxesToUpdate}} \le (2\updateRadius + 1)^{\dimensions} \range^{\dimensions}$ yields $q - p \ge \left(1 - \eulerE^{-\varepsilon}\right) \eulerE^{-\activity (2\updateRadius + 1)^{\dimensions} \range^{\dimensions}}$.
	
	We proceed by arguing that we can obtain an oracle for $\Ber{p}$ and $\Ber{q}$ as required by \Cref{lemma:bernoulli_frac}.
	Firstly, note that by \Cref{obs:compute_correction_HS} we can compute $\bayesFilterCorrectionHS{\varepsilon}[S][\vectorize{v}][\eta]$ with running time only depending on $\varepsilon$, $\updateRadius$, $\range$, $\activity$ and $\dimensions$.
	After computing $\bayesFilterCorrectionHS{\varepsilon}[S][\vectorize{v}][\eta]$, each independent sample from  $\Ber{\eulerE^{-\varepsilon}\bayesFilterCorrectionHS{\varepsilon}[S][\vectorize{v}][\eta]}$ can be obtained in constant time. 
	Thus, it remains to argue that we can efficiently sample independent Bernoulli random variables with success probabilities $\gibbs[\boxesToUpdate \setminus \vectorize{v}][\eta \cap (\partial \boxesToUpdate \cup \vectorize{v})](\set{\emptyset})$ and $\gibbs[\boxesToUpdate][\eta \cap \partial \boxesToUpdate](\set{\emptyset})$.
    By \Cref{lemma:exp_perfect_sampling}, we can obtain independent samples from $\gibbs[\boxesToUpdate][\eta \cap \partial \boxesToUpdate]$, each in expected time at most $(\activity \volume{\boxRegion{\boxesToUpdate}} + \size{\eta \cap \partial \boxesToUpdate}) \cdot \activity \volume{\boxRegion{\boxesToUpdate}} \eulerE^{\activity \volume{\boxRegion{\boxesToUpdate}}}$.
	Note that $\volume{\boxRegion{\boxesToUpdate}} \le (2 \updateRadius + 1)^{\dimensions} \range^{\dimensions}$ and that, for feasible $\eta \in \pointsets[\region]$, $\size{\eta \cap \partial \boxesToUpdate} \le (2 \sqrt{\dimensions}/\range)^{\dimensions} \volume{\boxRegion{\partial \boxesToUpdate}} \le (2 \sqrt{\dimensions}/\range)^{\dimensions} \cdot (2\updateRadius + 3)^{\dimensions} \range^{\dimensions}$.
	Therefore, the expected running time for obtaining independent Bernoulli samples with success probability $\gibbs[\boxesToUpdate][\eta \cap \partial \boxesToUpdate](\set{\emptyset})$ is bounded by some function of $\varepsilon$, $\updateRadius$, $\range$, $\activity$ and $\dimensions$.
	Treating $\gibbs[\boxesToUpdate \setminus \vectorize{v}][\eta \cap (\partial \boxesToUpdate \cup \vectorize{v})](\set{\emptyset})$ analogously and applying \Cref{lemma:bernoulli_frac} now proves our claim.
\end{proof}

We conclude the following bound on the running time of each iteration.
\begin{corollary} \label{cor:sampling_bayes_filter_HS}
	Suppose we run \Cref{algo:sampling} on a hard-sphere model with $\bayesFilterCorrection(\cdot) = \eulerE^{-\varepsilon}\bayesFilterCorrectionHS{\varepsilon}(\cdot)$ as Bayes filter correction in line~\ref{algo:sampling:filter} for some $\varepsilon > 0$, and let $\iterationTime{t}$ denote the running time of iteration $t \in \N$.
	Then, for all $t \in \N$, $\E{\iterationTime{t}}[\currentPointset{t-1}, \boxesToFix{t-1}, \boxChosen{t-1}]$ is almost surely bounded by some function of $\varepsilon$, $\updateRadius$, $\range$, $\activity$ and $\dimensions$.
\end{corollary}
\begin{proof}
	Set $\boxesToUpdate = \boxesToUpdate[\boxesToFix{t-1}][\boxChosen{t-1}][\range][\updateRadius]$, and note that the bulk of the running time in each iteration of \Cref{algo:sampling} is due to sampling the Bayes filter in \cref{algo:sampling:filter} and updating the point configuration on $\boxRegion{\boxesToUpdate}$ in \cref{algo:sampling:update}.
	
	For \cref{algo:sampling:filter}, note that $\currentPointset{t-1}$ is almost surely feasible by \Cref{lemma:feasible_bounded}.
	Thus, \Cref{lemma:sampling_bayes_filter_HS} yields that the expected time for sampling the Bayes filter, conditioned on $\currentPointset{t-1}, \boxesToFix{t-1}$ and $\boxChosen{t-1}$, almost surely bounded by some function of $\varepsilon$, $\updateRadius$, $\range$, $\activity$ and $\dimensions$.
	
	For \cref{algo:sampling:update}, we can use \Cref{lemma:exp_perfect_sampling} to bound the expected time for sampling from $\gibbs[\boxesToUpdate][\currentPointset{t-1} \cap \boxesToUpdate^{\comp}]$ is bounded by $(\activity \volume{\boxRegion{\boxesToUpdate}} + \size{\currentPointset{t-1} \cap \partial \boxesToUpdate}) \cdot \activity \volume{\boxRegion{\boxesToUpdate}} \eulerE^{\activity \volume{\boxRegion{\boxesToUpdate}}}$.
	Note that $\volume{\boxRegion{\boxesToUpdate}} \le (2 \updateRadius + 1)^{\dimensions} \range^{\dimensions}$ and, if $\currentPointset{t-1}$ is feasible, $\size{\currentPointset{t-1} \cap \partial \boxesToUpdate} \le (2 \sqrt{\dimensions}/\range)^{\dimensions} \volume{\boxRegion{\partial \boxesToUpdate}} \le (2 \sqrt{\dimensions}/\range)^{\dimensions} \cdot (2\updateRadius + 3)^{\dimensions} \range^{\dimensions}$.
	Since $\currentPointset{t-1}$ is indeed almost surely feasible, the expected running time for \cref{algo:sampling:update}, conditioned on $\currentPointset{t-1}, \boxesToFix{t-1}$ and $\boxChosen{t-1}$, is almost surely bounded by some function of $\varepsilon$, $\updateRadius$, $\range$, $\activity$ and $\dimensions$ as well, which concludes the proof.
\end{proof}

We proceed by bounding the expected number of iterations of \Cref{algo:sampling}, running on a hard-sphere model.
To this end, we start with the following lower bound on the success probability of the Bayes filter with correction $\eulerE^{-\varepsilon} \bayesFilterCorrectionHS{\varepsilon}(\cdot)$ for a particular choice of $\varepsilon$.
\begin{lemma} \label{lemma:lower_bound_filter_HS}
	Consider a hard-sphere model that exhibits $(a, b)$-strong spatial mixing up to $\activity$.
	Then there are constants $a', b'$, only depending on $a$, $b$, $\range$, $\activity$ and $\dimensions$, such that for all non-empty $S \subseteq \boxIds$, $\vectorize{v} \in S$ and feasible $\eta \in \pointsets[\region]$ it holds that
	\[
		\eulerE^{-\eulerE^{-\updateRadius}}\bayesFilterCorrectionHS{\eulerE^{-\updateRadius}}[S][\vectorize{v}][\eta] \cdot
		\frac{
			\partitionFunction[\boxesToUpdate][\eta \cap \partial \boxesToUpdate] 
		}{ 
			\partitionFunction[\boxesToUpdate \setminus \vectorize{v}][\eta \cap (\partial \boxesToUpdate \cup \vectorize{v})]
		}
		\ge 1 - a' \eulerE^{-b' \updateRadius} .
	\]
\end{lemma}
\begin{proof}
	Set $\delta_{1} = \delta_{1}\left(\eulerE^{-\updateRadius}\right)$ and $\delta_{2} = \delta_{2}\left(\eulerE^{-\updateRadius}\right)$ as defined in \Cref{lemma:correction_HS}.
	Note that by \Cref{lemma:spatial_markov} we have
	\[
		\partitionFunction[\boxesToUpdate][\eta \cap \partial \boxesToUpdate] 
		= \partitionFunction[\boxesToUpdate][\eta \cap \boxesToUpdate^{\comp}] 
		= \partitionFunction[\boxesToUpdate][(\eta \cap H) \cup (\eta \cap (S \setminus \vectorize{v}))],
	\] 
	where the last equality comes from the fact that $H$ and $S \setminus \vectorize{v}$ form a partitioning of $\boxesToUpdate^{\comp}$.
	Similarly, we obtain
	\[
		\partitionFunction[\boxesToUpdate \setminus \vectorize{v}][\eta \cap (\partial \boxesToUpdate \cup \vectorize{v})] 
		= \partitionFunction[\boxesToUpdate \setminus \vectorize{v}][\eta \cap (\boxesToUpdate \setminus \vectorize{v})^{\comp}]
		= \partitionFunction[\boxesToUpdate \setminus \vectorize{v}][(\eta \cap H) \cup (\eta \cap S)].
	\] 
	Thus, applying \Cref{lemma:ssm_partition_function} with $\xi_1 = \gamma$ and $\xi_2 = \eta \cap H$ yields
	\begin{align*}
		&\min_{\gamma \subseteq (\delta_1 \Z)^{\dimensions} \cap \boxRegion{H \cap \partial \boxesToUpdate}} \left\{\frac{\partitionFunction[\boxesToUpdate \setminus \vectorize{v}][\gamma \cup (\eta \cap S)]}{\partitionFunction[\boxesToUpdate][\gamma \cup (\eta \cap (S \setminus \vectorize{v}))]}\right\} \cdot \frac{
			\partitionFunction[\boxesToUpdate][\eta \cap \partial \boxesToUpdate] 
		}{ 
			\partitionFunction[\boxesToUpdate \setminus \vectorize{v}][\eta \cap (\partial \boxesToUpdate \cup \vectorize{v})]
		}\\
		&\hspace*{4em}\ge \exponential{- a 3^{\dimensions} \range^{\dimensions} \eulerE^{2 b \range} \left(\activity \range^{\dimensions} + \eulerE^{\activity 3^{\dimensions} \range^{\dimensions}}\right) \eulerE^{-b \range \updateRadius}} .
	\end{align*}
	Moreover, following the same arguments as in \Cref{lemma:correction_HS}, we have
	\[
		\eulerE^{-\eulerE^{-\updateRadius}/2} \cdot \min_{\gamma \subseteq (\delta_1 \Z)^{\dimensions} \cap \boxRegion{H \cap \partial \boxesToUpdate}} \frac{\partitionFunction[\boxesToUpdate \setminus \vectorize{v}][\gamma \cup (\eta \cap S)]}{\partitionFunction[\boxesToUpdate][\gamma \cup (\eta \cap (S \setminus \vectorize{v}))]}
		\le \min_{\gamma \subseteq (\delta_1 \Z)^{\dimensions} \cap \boxRegion{H \cap \partial \boxesToUpdate}} \frac{\partitionFunctionApprox[\boxesToUpdate \setminus \vectorize{v}][\gamma \cup (\eta \cap S)][\delta_2]}
		{\partitionFunctionApprox[\boxesToUpdate][\gamma \cup (\eta \cap (S \setminus \vectorize{v}))][\delta_2]}.
	\]
	Recalling the definition of $\bayesFilterCorrectionHS{\eulerE^{-\updateRadius}}[S][\vectorize{v}][\eta]$ in \Cref{lemma:correction_HS} and noting that for $\updateRadius \ge 2$ it holds that $b \cdot (\updateRadius-1)\range \ge \frac{b \range}{2} \updateRadius$, this implies
	\begin{align*}
		\eulerE^{-\eulerE^{-\updateRadius}}\bayesFilterCorrectionHS{\eulerE^{-\updateRadius}}[S][\vectorize{v}][\eta] \cdot
		\frac{
			\partitionFunction[\boxesToUpdate][\eta \cap \partial \boxesToUpdate] 
		}{ 
			\partitionFunction[\boxesToUpdate \setminus \vectorize{v}][\eta \cap (\partial \boxesToUpdate \cup \vectorize{v})]
		}
		&\ge \exponential{- a 3^{\dimensions} \range^{\dimensions} \eulerE^{2 b \range} \left(\activity \range^{\dimensions} + \eulerE^{\activity 3^{\dimensions} \range^{\dimensions}}\right) \eulerE^{-b \range \updateRadius}} \cdot \exponential{2\eulerE^{- \updateRadius}} \\
		&\ge 1 - a' \eulerE^{- b' \updateRadius} 
	\end{align*}
	for $b' = \min\set{1, b \range}$ and $a' = a 3^{\dimensions} \range^{\dimensions} \eulerE^{2b \range} \left(\activity \range^{\dimensions} + \eulerE^{\activity 3^{\dimensions} \range^{\dimensions}}\right) + 2$, which concludes the proof.
\end{proof}

\Cref{lemma:lower_bound_filter_HS} allows us to control the success probability of the Bayes filter in terms of $\updateRadius$.
This leads to the following statement.

\begin{lemma} \label{lemma:iterations_HS}
	Consider a hard-sphere model that exhibits $(a, b)$-strong spatial mixing up to $\activity$.
	Suppose we run \Cref{algo:sampling} with $\bayesFilterCorrection(\cdot) = \eulerE^{-\eulerE^{-\updateRadius}}\bayesFilterCorrectionHS{\eulerE^{-\updateRadius}}(\cdot)$ as Bayes filter correction in line~\ref{algo:sampling:filter}, and let $\iterationNumber = \inf\set{t \in \N_{0} \mid \boxesToFix{t} = \emptyset}$ denote the number of iterations until the algorithm terminates.
	Then, for $\updateRadius$ sufficiently large depending on $a$, $b$, $\range$, $\activity$ and $\dimensions$, it holds that $\E{T}[] \le 2 \size{\boxIds}$.
\end{lemma}
\begin{proof}
	We aim for applying \Cref{thm:additive_drift} to prove our claim.
	To this end, consider the process $(\size{\boxesToFix{t}})_{t \in \N_{0}}$ and the filtration $(\mathcal{F}_t)_{t \in \N_{0}}$ defined by $\mathcal{F}_t = \sigma\left((\currentPointset{j}, \boxesToFix{j}, \boxChosen{j})_{j \le t}\right)$. 
	Further, observe that our desired hitting time can equivalently be written as $\iterationNumber = \inf\set{t \in \N_{0}}[\size{\boxesToFix{t}} \le 0]$.
	Since we are interested in the expectation of $\iterationNumber$, we only need to check assumptions \ref{thm:additive_drift:non_negative} and \ref{thm:additive_drift:drift} of \Cref{thm:additive_drift}.
		
	For \ref{thm:additive_drift:non_negative} of \Cref{thm:additive_drift}, observe that $\size{\boxesToFix{t}} \ge 0$ for all $t \in \N_{0}$.
	For \ref{thm:additive_drift:drift}, we prove that $\E{(\size{\boxesToFix{t}} - \size{\boxesToFix{t+1}}) \ind{T > t}}[\mathcal{F}_t] \ge \frac{1}{2} \ind{\iterationNumber > t}$ if $\updateRadius$ is sufficiently large.
	Since
	\[
		\E{(\size{\boxesToFix{t}} - \size{\boxesToFix{t+1}}) \ind{\iterationNumber > t}}[\mathcal{F}_t] 
		= \left(\size{\boxesToFix{t}} - \E{\size{\boxesToFix{t+1}}}[\mathcal{F}_t]\right) \ind{\iterationNumber > t}
	\]
	it suffices to show that $\E{\size{\boxesToFix{t+1}}}[\mathcal{F}_t] \ind{\iterationNumber > t} \le \left(\size{\boxesToFix{t}} - \frac{1}{2} \right)\ind{\iterationNumber > t}$.
		
	To simplify notation, we will omit the indicator of $\iterationNumber >t$ while still restricting ourselves to the setting where $\boxesToFix{t} \neq \emptyset$.
	Next, observe that, if $\bayesFilter{t} = 1$, then $\size{\boxesToFix{t+1}} = \size{\boxesToFix{t}} - 1$.
	On the other hand, if $\bayesFilter{t} = 0$, then
	\[
		\size{\boxesToFix{t+1}} 
		= \size{\boxesToFix{t}} + \size{\partial \boxesToUpdate}
		\le \size{\boxesToFix{t}} + \size{\boxesToUpdate \cup \partial \boxesToUpdate}  
		\le \size{\boxesToFix{t}} + \left(2 \updateRadius + 3\right)^{\dimensions},
	\]
	where $\boxesToUpdate = \boxesToUpdate[\boxesToFix{t}][\boxChosen{t}][\updateRadius]$.
	Thus, it suffices if
	\[
		\left(2 \updateRadius + 3\right)^{\dimensions} \cdot (1-\E{\bayesFilter{t}}[\mathcal{F}_t]) - \E{\bayesFilter{t}}[\mathcal{F}_t] \le - \frac{1}{2} .
	\]
	Since, by \Cref{lemma:lower_bound_filter_HS}, 
	\begin{align*}
		\E{\bayesFilter{t}}[\mathcal{F}_t] \ge 1 - a' \eulerE^{-b' \updateRadius},
	\end{align*}
		for $a'$ and $b'$ only depending on $a$, $b$, $\activity$, $\range$ and $\dimensions$, this is satisfied for $\updateRadius$ sufficiently large, depending on $a$, $b$, $\activity$, $\range$ and $\dimensions$.
	Applying \Cref{thm:additive_drift} then yields $\E{\iterationNumber}[] \le 2 \size{\boxIds}$, which concludes the proof.
\end{proof}

We conclude the following theorem.
\begin{theorem} \label{thm:sampling_HS}
	Consider \Cref{algo:sampling} on a hard-sphere model with $\bayesFilterCorrection(\cdot) = \eulerE^{-\eulerE^{-\updateRadius}}\bayesFilterCorrectionHS{\eulerE^{-\updateRadius}}(\cdot)$ as Bayes filter correction in line~\ref{algo:sampling:filter}.
	We can run the algorithm in almost-surely finite running time and, on termination, it outputs a sample from the hard-sphere Gibbs measure on $\region$.
	Moreover, if the hard-sphere model satisfies $(a, b)$-strong spatial mixing and if $\updateRadius$ is chosen as a sufficiently large constant, depending on $a$, $b$, $\range$, $\activity$ and $\dimensions$, then we can run the algorithm in expected time $\bigO{\volume{\region}}$. 
\end{theorem}

\begin{proof}	
	For the first part of the statement, note that the correct output distribution follows directly from \Cref{thm:correctness} and the fact that $\eulerE^{-\eulerE^{-\updateRadius}}\bayesFilterCorrectionHS{\eulerE^{-\updateRadius}}(\cdot)$ is a Bayes filter correction by \Cref{lemma:correction_HS}. 
	Let $\iterationNumber$ denote that number of iterations of \Cref{algo:sampling}, and let $R_t$ denote the running time in iterations $t \in \N$.
	By \Cref{cor:finite_iterations} we know that $\iterationNumber$ is almost surely finite.
	Moreover, it holds that $\E{R_t}[] = \E{\E{R_t}[\currentPointset{t-1}, \boxesToFix{t-1}, \boxChosen{t-1}]}[]$.
	Since by \Cref{cor:sampling_bayes_filter_HS} $\E{R_t}[\currentPointset{t-1}, \boxesToFix{t-1}, \boxChosen{t-1}] \le t(\updateRadius, \range, \activity, \dimensions)$ for some function $t: \R_{\ge 0}^{3} \times \N_{0} \to \R_{\ge 0}$, it also holds that $\E{R_t}[] \le t(\updateRadius, \range, \activity, \dimensions)$. 
	Consequently, $R_t$ must be finite almost surely, and \Cref{algo:sampling} has almost surely finite running time.
	
	For the second part of the statement, suppose the hard-sphere model satisfies $(a, b)$-strong spatial mixing up to $\activity$.
	Observe that the expected running time of \Cref{algo:sampling} can be expressed as
	\begin{align*}
		\E{\sum_{t = 1}^{\iterationNumber} R_t}[] 
		&= \E{\sum_{t \ge 1} \ind{t \ge \iterationNumber} \E{R_t}[\currentPointset{t-1}, \boxesToFix{t-1}, \boxChosen{t-1}]}[] \\
		&\le t(\updateRadius, \range, \activity, \dimensions)  \E{\sum_{t \ge 1} \ind{t \ge \iterationNumber}}[] \\
		&= t(\updateRadius, \range, \activity, \dimensions) \E{\iterationNumber}[],
	\end{align*}
	where the first equality uses the fact that $\ind{t \ge \iterationNumber} = \ind{\boxesToFix{t-1} \neq \emptyset}$ is $\sigma(\boxesToFix{t-1})$-measurable.
	By \Cref{lemma:iterations_HS}, we can choose $\updateRadius$ sufficiently large, depending on $a$, $b$, $\range$, $\activity$ and $\dimensions$ only, such that $\E{\iterationNumber}[] \le 2 \size{\boxIds} \in \bigO{\volume{\region}}$, proving linear expected running time of the algorithm for any such choice of $\updateRadius$.
\end{proof}

\section{General repulsive potentials}
\label{secRepulsive}
We now extend our perfect sampling algorithm to the setting of more general bounded-range repulsive potentials $\potential$.
In contrast to the hard-sphere model, it is not clear how to perform the minimization task in involved in constructing the Bayes filter in this setting.
We will instead assume knowledge of the rate of strong spatial mixing for constructing the Bayes filter.

\begin{lemma} \label{lemma:filter_repsulsive}
	Let $S \subseteq \boxIds$ be non-empty and $\vectorize{v} \in S$. 
	Let $\eta \in \pointsets[\region]$ be feasible and set $\boxesToUpdate = \boxesToUpdate[S][\vectorize{v}][\updateRadius]$.
	Suppose $a, b > 0$ are such that $\potential$ satisfies $(a, b)$-strong spatial mixing up to $\activity$, and set
	\[
		\correction = \correction[a][b]  \coloneqq \exponential{- a 3^{\dimensions} \range^{\dimensions} \eulerE^{2 b \range} \left(\activity \range^{\dimensions} + \eulerE^{\activity 3^{\dimensions} \range^{\dimensions}}\right) \eulerE^{-b \range \updateRadius}} 
	\]
	and 
	\[
		\bayesFilterCorrectionR{a}{b}[S][\vectorize{v}][\eta] \coloneqq \correction \cdot \frac{\partitionFunction[\boxesToUpdate \setminus \vectorize{v}][\eta \cap S]}{\partitionFunction[\boxesToUpdate][\eta \cap (S \setminus \vectorize{v})]} .
	\]
	Then $\bayesFilterCorrectionR{a}{b}[S][\vectorize{v}][\eta]$ is a Bayes filter correction as in \Cref{def:bayes_filter}.
	Moreover, it holds that 
	\[
		\exponential{- 2 a 3^{\dimensions} \range^{\dimensions} \eulerE^{2 b \range} \left(\activity \range^{\dimensions} + \eulerE^{\activity 3^{\dimensions} \range^{\dimensions}}\right) \eulerE^{-b \range \updateRadius}}
		\le 
		\bayesFilterCorrectionR{a}{b}[S][\vectorize{v}][\eta] \cdot \frac{
			\partitionFunction[\boxesToUpdate][\eta \cap \partial \boxesToUpdate] 
		}{ 
			\partitionFunction[\boxesToUpdate \setminus \vectorize{v}][\eta \cap (\partial \boxesToUpdate \cup \vectorize{v})]
		}
		\le 1 .
	\]
\end{lemma}

\begin{proof}
	We start by checking that $\bayesFilterCorrectionR{a}{b}(\cdot)$ is a Bayes filter correction.
	For the measurability, note that for all non-empty $S \subseteq \boxIds$, $\vectorize{v} \in S$ and $\eta \in \pointsets[\region]$ it holds that $\bayesFilterCorrectionR{a}{b}[S][\vectorize{v}][\eta]$ does only depend on $\eta \cap S$.
	Moreover, observe that
	\begin{align*}
		\bayesFilterCorrectionR{a}{b}[S][\vectorize{v}][\eta] 
		&\ge \exponential{- a 3^{\dimensions} \range^{\dimensions} \eulerE^{2 b \range} \left(\activity \range^{\dimensions} + \eulerE^{\activity 3^{\dimensions} \range^{\dimensions}}\right) \eulerE^{-b \range \updateRadius}} \cdot \eulerE^{- \activity \volume{\boxRegion{\boxesToUpdate \cup \partial \boxesToUpdate}}} \\
		&\ge \exponential{- a 3^{\dimensions} \range^{\dimensions} \eulerE^{2 b \range} \left(\activity \range^{\dimensions} + \eulerE^{\activity 3^{\dimensions} \range^{\dimensions}}\right) \eulerE^{-b \range \updateRadius}} \cdot \exponential{- \activity (2 \updateRadius + 3)^{\dimensions} \range^{\dimensions}}
	\end{align*}
	uniformly in $\eta$.
	For the upper bound, we apply \Cref{lemma:ssm_partition_function} with $\xi_1 = \emptyset$ and $\xi_2 = \xi$ to obtain
	\begin{align*}
		\frac{\partitionFunction[\boxesToUpdate \setminus \vectorize{v}][\eta \cap S]}{\partitionFunction[\boxesToUpdate][\eta \cap (S \setminus \vectorize{v})]} \cdot \inf_{\substack{\xi \in \pointsets[H]\\ \xi \cup (\eta \cap S) \text{ is feasible}}}  \frac{\partitionFunction[\boxesToUpdate][\xi \cup (\eta \cap (S \setminus \vectorize{v}))]}{\partitionFunction[\boxesToUpdate \setminus \vectorize{v}][\xi \cup (\eta \cap S)]} \le \exponential{a 3^{\dimensions} \range^{\dimensions} \eulerE^{2 b \range} \left(\activity \range^{\dimensions} + \eulerE^{\activity 3^{\dimensions} \range^{\dimensions}}\right) \eulerE^{-b \range \updateRadius}} .
	\end{align*}
	Therefore, multiplying both sides with $\correction$ yields
	\[
			\bayesFilterCorrectionR{a}{b}[S][\vectorize{v}][\eta] \le \inf_{\substack{\xi \in \pointsets[H]\\ \xi \cup (\eta \cap S) \text{ is feasible}}}  \frac{\partitionFunction[\boxesToUpdate \setminus \vectorize{v}][\xi \cup (\eta \cap S)]}{\partitionFunction[\boxesToUpdate][\xi \cup (\eta \cap (S \setminus \vectorize{v}))]}
	\] 
	as desired.
	
	For the second part of the statement, note that by \Cref{lemma:spatial_markov} we have
	\[
		\partitionFunction[\boxesToUpdate][\eta \cap \partial \boxesToUpdate] 
		= \partitionFunction[\boxesToUpdate][\eta \cap \boxesToUpdate^{\comp}] 
		= \partitionFunction[\boxesToUpdate][(\eta \cap H) \cup (\eta \cap (S \setminus \vectorize{v}))],
	\] 
	where the last equality comes from the fact that $H$ and $S \setminus \vectorize{v}$ form a partitioning of $B^{\comp}$.
	Similarly, we obtain
	\[
		\partitionFunction[\boxesToUpdate \setminus \vectorize{v}][\eta \cap (\partial \boxesToUpdate \cup \vectorize{v})] 
		= \partitionFunction[\boxesToUpdate \setminus \vectorize{v}][\eta \cap (\boxesToUpdate \setminus \vectorize{v})^{\comp}]
		= \partitionFunction[\boxesToUpdate \setminus \vectorize{v}][(\eta \cap H) \cup (\eta \cap S)].
	\] 
	Since $\bayesFilterCorrectionR{a}{b}[S][\vectorize{v}][\eta]$ is a Bayes filter correction, it follows that
	\[
		\bayesFilterCorrectionR{a}{b}[S][\vectorize{v}][\eta] \cdot \frac{
			\partitionFunction[\boxesToUpdate][\eta \cap \partial \boxesToUpdate] 
		}{ 
			\partitionFunction[\boxesToUpdate \setminus \vectorize{v}][\eta \cap (\partial \boxesToUpdate \cup \vectorize{v})]
		}
		\le 1 .
	\]
	Moreover, applying \Cref{lemma:ssm_partition_function} with $\xi_1 = \emptyset$ and $\xi_2 = \eta \cap \boxRegion{H}$ yields
	\begin{align*}
		\frac{\partitionFunction[\boxesToUpdate \setminus \vectorize{v}][\eta \cap S]}{\partitionFunction[\boxesToUpdate][\eta \cap (S \setminus \vectorize{v})]} \cdot \frac{
			\partitionFunction[\boxesToUpdate][\eta \cap \partial \boxesToUpdate] 
		}{ 
			\partitionFunction[\boxesToUpdate \setminus \vectorize{v}][\eta \cap (\partial \boxesToUpdate \cup \vectorize{v})]
		}
		\ge \exponential{-a 3^{\dimensions} \range^{\dimensions} \eulerE^{2 b \range} \left(\activity \range^{\dimensions} + \eulerE^{\activity 3^{\dimensions} \range^{\dimensions}}\right) \eulerE^{-b \range \updateRadius}}
	\end{align*} 
	and multiplying both sides with $\correction$ proves the claim.
\end{proof}

While the first part of \Cref{lemma:filter_repsulsive} is sufficient to guarantee correctness of \Cref{algo:sampling}, the second part of the statement will be useful for bounding the running by allowing us to control the probability of the event $\bayesFilter{t} = 0$ for each iteration $t \in \N_{0}$.
Similarly as in the setting of the hard-sphere model, we will not work directly with $\bayesFilterCorrectionR{a}{b}[S][\vectorize{v}][\eta]$, but rather use a slightly scaled version, which is a Bayes filter correction in its own right.

\begin{corollary} \label{cor:scaled_filter_repsulsive}
	In the setting of \Cref{lemma:filter_repsulsive}, it holds that $\eulerE^{-\eulerE^{-\updateRadius}} \bayesFilterCorrectionR{a}{b}[S][\vectorize{v}][\range]$ is a Bayes filter correction, and there are constants $a', b'$, only depending on $a$, $b$, $\range$, $\activity$ and $\dimensions$, such that
	\[
		1 - a' \eulerE^{-b' \updateRadius}
		\le 
		\eulerE^{-\eulerE^{-\updateRadius}} \bayesFilterCorrectionR{a}{b}[S][\vectorize{v}][\eta] \cdot \frac{
			\partitionFunction[\boxesToUpdate][\eta \cap \partial \boxesToUpdate] 
		}{ 
			\partitionFunction[\boxesToUpdate \setminus \vectorize{v}][\eta \cap (\partial \boxesToUpdate \cup \vectorize{v})]
		}
		\le \eulerE^{-\eulerE^{-\updateRadius}}.
	\]
\end{corollary}

\begin{proof}
	Since $0 < \eulerE^{-\eulerE^{-\updateRadius}} \le 1$, it is obvious that $\eulerE^{-\eulerE^{-\updateRadius}} \bayesFilterCorrectionR{a}{b}[S][\vectorize{v}][\range]$ satisfies \Cref{def:bayes_filter}.
	Now, note that for $\updateRadius \ge 2$ it holds that $b (\updateRadius - 1) \range \ge \frac{b \range}{2} \updateRadius$.
	The statement directly follows from \Cref{lemma:filter_repsulsive} by setting $b' = \min\set{1, b \range}$ and $a' = 2 a 3^{\dimensions} \range^{\dimensions} \eulerE^{2 b \range} \left(\activity \range^{\dimensions} + \eulerE^{\activity 3^{\dimensions} \range^{\dimensions}}\right) + 1$, and observing that $\exponential{-a' \eulerE^{-b' \updateRadius}} \ge 1 - a' \eulerE^{-b' \updateRadius}$.
\end{proof}

Next, we focus on how to sample the Bayes filter, using $\eulerE^{-\eulerE^{-\updateRadius}} \bayesFilterCorrectionR{a}{b}(\cdot)$ as in \Cref{cor:scaled_filter_repsulsive} as Bayes filter correction.
In contrast to our approach for the hard-sphere model, we do not know how to compute $\bayesFilterCorrectionR{a}{b}(\cdot)$ directly.
Again, we solve this problem using a Bernoulli factory.

\begin{lemma} \label{lemma:sampling_bayes_filter}
	Let $S \subseteq \boxIds$ be non-empty, $\vectorize{v} \in S$ and $\eta \in \pointsets[\region]$ be feasible, and set $\boxesToUpdate = \boxesToUpdate[S][\vectorize{v}][\updateRadius]$.
	Suppose $a, b > 0$ are such that $\potential$ satisfies $(a, b)$-strong spatial mixing up to $\activity$ and let $\bayesFilterCorrectionR{a}{b}(\cdot)$ be as in \Cref{lemma:filter_repsulsive}. 
	We can sample a Bernoulli random variable with success probability
	\[
		\eulerE^{-\eulerE^{-\updateRadius}} \bayesFilterCorrectionR{a}{b}[S][\vectorize{v}][\eta] \cdot \frac{
			\partitionFunction[\boxesToUpdate][\eta \cap \partial \boxesToUpdate] 
		}{ 
			\partitionFunction[\boxesToUpdate \setminus \vectorize{v}][\eta \cap (\partial \boxesToUpdate \cup \vectorize{v})]
		}
	\] 
	with expected running time in $\bigO{\size{\eta \cap (\partial \boxesToUpdate \cup \vectorize{v}}}$, where the constants only depend on $\updateRadius$, $\range$, $\activity$ and $\dimensions$.
\end{lemma}
\begin{proof}
	Our goal is to use a Bernoulli factory of the form $\frac{p}{q}$ to perform this task.
	To bring the desired success probability into such a form, note that
	\begin{align*}
		\partitionFunction[\boxesToUpdate][\eta \cap \partial \boxesToUpdate] &= \left(\gibbs[\boxesToUpdate][\eta \cap \partial \boxesToUpdate](\set{\emptyset})\right)^{-1} \\
		\partitionFunction[\boxesToUpdate \setminus \vectorize{v}][\eta \cap (\partial \boxesToUpdate \cup \vectorize{v})] &= \left(\gibbs[\boxesToUpdate \setminus \vectorize{v}][\eta \cap (\partial \boxesToUpdate \cup \vectorize{v})](\set{\emptyset})\right)^{-1} .
	\end{align*}
	Moreover, using \Cref{lemma:spatial_markov} we have
	\begin{align*}
		\partitionFunction[\boxesToUpdate \setminus \vectorize{v}][\eta \cap S]
		&= \partitionFunction[\boxesToUpdate \setminus \vectorize{v}][\eta \cap ((\partial \boxesToUpdate \cap S) \cup \vectorize{v})] 
		= \left(\gibbs[\boxesToUpdate \setminus \vectorize{v}][\eta \cap ((\partial \boxesToUpdate \cap S) \cup \vectorize{v})](\set{\emptyset})\right)^{-1}\\
		\partitionFunction[\boxesToUpdate][\eta \cap (S \setminus \vectorize{v})]
		&= \partitionFunction[\boxesToUpdate][\eta \cap (\partial \boxesToUpdate \cap S)]
		= \left(\gibbs[\boxesToUpdate][\eta \cap (\partial \boxesToUpdate \cap S)](\set{\emptyset})\right)^{-1}.
	\end{align*}
	Finally, note that $\eulerE^{-\eulerE^{-\updateRadius}} \in [0, 1]$ and, for $\correction = \correction[\updateRadius]$ as in \Cref{lemma:filter_repsulsive}, $\correction \in [0, 1]$. 
	Thus, by setting 
	\begin{align*}
		p &\coloneqq \eulerE^{-\eulerE^{-\updateRadius}} \correction \cdot \gibbs[\boxesToUpdate][\eta \cap (\partial \boxesToUpdate \cap S)](\set{\emptyset}) \cdot \gibbs[\boxesToUpdate \setminus \vectorize{v}][\eta \cap (\partial \boxesToUpdate \cup \vectorize{v})](\set{\emptyset}) \\
		q &\coloneqq \gibbs[\boxesToUpdate \setminus \vectorize{v}][\eta \cap ((\partial \boxesToUpdate \cap S) \cup \vectorize{v})](\set{\emptyset}) \cdot \gibbs[\boxesToUpdate][\eta \cap \partial \boxesToUpdate](\set{\emptyset})
	\end{align*}
	we have $p \in [0, 1], q \in [0, 1]$ and 
	\[
		\eulerE^{-\eulerE^{-\updateRadius}} \bayesFilterCorrectionR{a}{b}[S][\vectorize{v}][\eta] \cdot \frac{
		\partitionFunction[\boxesToUpdate][\eta \cap \partial \boxesToUpdate] 
		}{ 
		\partitionFunction[\boxesToUpdate \setminus \vectorize{v}][\eta \cap (\partial \boxesToUpdate \cup \vectorize{v})]
		}
		= \frac{p}{q} .
	\]
	
	We are now going to use \Cref{lemma:bernoulli_frac} to prove that we can obtain a sample from $\Ber{\frac{p}{q}}$ within the desired expected running time.
	To this end, we need to provide a positive lower bound on $q - p$ and we need an efficient way for generating independent samples from $\Ber{q}$ and $\Ber{p}$.
	
	For the lower bound, note that by \Cref{cor:scaled_filter_repsulsive} it holds that 
	\[
		\eulerE^{-\eulerE^{-\updateRadius}} \bayesFilterCorrectionR{a}{b}[S][\vectorize{v}][\eta] \cdot \frac{
			\partitionFunction[\boxesToUpdate][\eta \cap \partial \boxesToUpdate] 
		}{ 
			\partitionFunction[\boxesToUpdate \setminus \vectorize{v}][\eta \cap (\partial \boxesToUpdate \cup \vectorize{v})]
		}
		\le \eulerE^{-\eulerE^{-\updateRadius}}.
	\]
	Consequently, we have 
	\begin{align*}
		q - p 
		&\ge \left(1 - \eulerE^{-\eulerE^{-\updateRadius}}\right) q \\
		&= \left(1 - \eulerE^{-\eulerE^{-\updateRadius}}\right) \cdot \left(\partitionFunction[\boxesToUpdate \setminus \vectorize{v}][\eta \cap ((\partial \boxesToUpdate \cap S) \cup \vectorize{v})] \cdot \partitionFunction[\boxesToUpdate][\eta \cap \partial \boxesToUpdate]\right)^{-1} \\
		&\ge \left(1 - \eulerE^{-\eulerE^{-\updateRadius}}\right) \eulerE^{-2\activity \volume{\boxRegion{\boxesToUpdate}}} .
	\end{align*}
	Using the upper bound $\volume{\boxRegion{\boxesToUpdate}} \le (2\updateRadius + 1)^{\dimensions} \range^{\dimensions}$ yields $q - p \ge \left(1 - \eulerE^{-\eulerE^{-\updateRadius}}\right) \eulerE^{-\activity (2\updateRadius + 1)^{\dimensions} \range^{\dimensions}}$.
	
	We proceed by arguing that we can obtain an oracle for $\Ber{p}$ and $\Ber{q}$ as required by \Cref{lemma:bernoulli_frac}.
	In particular, we focus on $\Ber{p}$ since $\Ber{q}$ can be treated analogously. 
	Firstly, note that we can sample a Bernoulli random variable with success probability $\eulerE^{-\eulerE^{-\updateRadius}} \correction$ in constant time, since we can compute it explicitly.
	It remains to argue that we can sample Bernoulli random variables with success probability $\gibbs[\boxesToUpdate][\eta \cap (\partial \boxesToUpdate \cap S)](\set{\emptyset})$ and $\gibbs[\boxesToUpdate \setminus \vectorize{v}][\eta \cap (\partial \boxesToUpdate \cup \vectorize{v})](\set{\emptyset})$ in the desired running time.
	Again, we focus on $\gibbs[\boxesToUpdate][\eta \cap (\partial \boxesToUpdate \cap S)](\set{\emptyset})$ and treat $\gibbs[\boxesToUpdate \setminus \vectorize{v}][\eta \cap (\partial \boxesToUpdate \cup \vectorize{v})](\set{\emptyset})$ analogously. 
	By \Cref{lemma:exp_perfect_sampling}, we can obtain independent samples from $\gibbs[\boxesToUpdate][\eta \cap \partial \boxesToUpdate]$, each in expected time at most $(\activity \volume{\boxRegion{\boxesToUpdate}} + \size{\eta \cap \partial \boxesToUpdate}) \cdot \activity \volume{\boxRegion{\boxesToUpdate}} \eulerE^{\activity \volume{\boxRegion{\boxesToUpdate}}}$.
	Noting that $\volume{\boxRegion{\boxesToUpdate}} \le (2 \updateRadius + 1)^{\dimensions} \range^{\dimensions}$ and that $\size{\eta \cap \partial \boxesToUpdate} \le \size{\eta \cap (\partial \boxesToUpdate \cup \vectorize{v})}$ yields an expected running time of $\bigO{\size{\eta \cap (\partial \boxesToUpdate \cup \vectorize{v})}}$.
	Applying the same argument to sample a Bernoulli random variable with success probability $\gibbs[\boxesToUpdate \setminus \vectorize{v}][\eta \cap (\partial \boxesToUpdate \cup \vectorize{v})](\set{\emptyset})$ yields an oracle for $\Ber{p}$ with expected running time in $\bigO{\size{\eta \cap (\partial \boxesToUpdate \cup \vectorize{v})}}$.
	Finally, applying the same procedure for $\Ber{q}$ and using \Cref{lemma:bernoulli_frac} concludes the proof.
\end{proof}

We obtain the following bound for the running time of each iteration.

\begin{corollary}\label{cor:sampling_bayes_filter}
	Let $a, b > 0$ be such that $\potential$ satisfies $(a, b)$-strong spatial mixing up to $\activity$.
	Suppose we run \Cref{algo:sampling} on $\potential$ with $\bayesFilterCorrection(\cdot) = \eulerE^{-\eulerE^{-\updateRadius}} \bayesFilterCorrectionR{a}{b}(\cdot)$ as Bayes filter correction in line~\ref{algo:sampling:filter}, and let $\iterationTime{t}$ denote the running time for iteration $t \in \N$.
	Then, for all $t \in \N$, it holds that
	\[
		\E{\iterationTime{t}}[\currentPointset{t-1}, \boxesToFix{t-1}, \boxChosen{t-1}] \le \bigO{\size{\currentPointset{t-1} \cap (\partial \boxesToUpdate \cup \boxChosen{t-1})}} ,
	\]
	where $\boxesToUpdate = \boxesToUpdate[\boxesToFix{t-1}][\boxChosen{t-1}][\updateRadius]$, and the constants in the asymptotic notation only depend on $\updateRadius$, $\range$, $\activity$ and $\dimensions$.
\end{corollary}

\begin{proof}
	Set $\boxesToUpdate = \boxesToUpdate[\boxesToFix{t-1}][\boxChosen{t-1}][\updateRadius]$, and note that the bulk of the running time in each iteration of \Cref{algo:sampling} is due to sampling the Bayes filter in \cref{algo:sampling:filter} and updating the point configuration on $\boxRegion{\boxesToUpdate}$ in \cref{algo:sampling:update}.
	
	For \cref{algo:sampling:filter}, note that $\currentPointset{t-1}$ is almost surely feasible by \Cref{lemma:feasible_bounded}.
	Thus, \Cref{lemma:sampling_bayes_filter} yields that the expected time for sampling the Bayes filter, conditioned on $\currentPointset{t-1}, \boxesToFix{t-1}$ and $\boxChosen{t-1}$, almost surely bounded by some function in $\bigO{\size{\currentPointset{t-1} \cap (\partial \boxesToUpdate \cup \boxChosen{t-1})}}$.
	Further, for \cref{algo:sampling:update}, we can use \Cref{lemma:exp_perfect_sampling} to bound the expected time for sampling from $\gibbs[\boxesToUpdate][\currentPointset{t-1} \cap \boxesToUpdate^{\comp}]$ is bounded by $(\activity \volume{\boxRegion{\boxesToUpdate}} + \size{\currentPointset{t-1} \cap \partial \boxesToUpdate}) \cdot \activity \volume{\boxRegion{\boxesToUpdate}} \eulerE^{\activity \volume{\boxRegion{\boxesToUpdate}}} \in \bigO{\size{\currentPointset{t-1} \cap (\partial \boxesToUpdate \cup \boxChosen{t-1})}}$ as desired.
\end{proof}

Next, we derive a bound on the expected number of iterations of \Cref{algo:sampling} given strong spatial mixing.

\begin{lemma} \label{lemma:iterations}
	Let $a, b > 0$ be such that $\potential$ satisfies $(a, b)$-strong spatial mixing up to $\activity$.
	Suppose we run \Cref{algo:sampling} with $\bayesFilterCorrection(\cdot) = \eulerE^{-\eulerE^{-\updateRadius}}\bayesFilterCorrectionR{a}{b}(\cdot)$ as Bayes filter correction in line~\ref{algo:sampling:filter}, and let $\iterationNumber = \inf\set{t \in \N_{0} \mid \boxesToFix{t} = \emptyset}$ denote the number of iterations until the algorithm terminates.
	Then, for $\updateRadius$ sufficiently large depending on $a$, $b$, $\range$, $\activity$ and $\dimensions$, it holds that $\E{\iterationNumber}[] \le  2 \size{\boxIds}$, and for all $k \ge 4 \size{\boxIds}$ it holds that $\Pr{\iterationNumber \ge k} \le \exponential{-\frac{k}{\alpha}}$ for some constant $\alpha \in \R_{>0}$ that only depends on $\updateRadius$, $\dimensions$ and $\range$.
\end{lemma}

\begin{proof}
	We aim for applying \Cref{thm:additive_drift} to prove our claim.
	For bounding $\E{\iterationNumber}[]$, we proceed analogously as in the proof of \Cref{lemma:iterations_HS}.
	In particular, we consider the process $(\size{\boxesToFix{t}})_{t \in \N_{0}}$ with the filtration $(\mathcal{F}_t)_{t \in \N_{0}}$ defined by $\mathcal{F}_t = \sigma\left((\currentPointset{j}, \boxesToFix{j}, \boxChosen{j})_{j \le t}\right)$ and rewrite our desired hitting time as $\iterationNumber = \inf\set{t \in \N_{0}}[\size{\boxesToFix{t}} \le 0]$.
	Since \ref{thm:additive_drift:non_negative} of \Cref{thm:additive_drift} is trivially satisfied, we only need to check \ref{thm:additive_drift:drift}.
	Using the lower bound from \Cref{cor:scaled_filter_repsulsive} and the same arguments as in the proof of \Cref{lemma:iterations_HS}, we can show that $\E{(\size{\boxesToFix{t}} - \size{\boxesToFix{t+1}}) \ind{T > t}}[\mathcal{F}_t] \ge \frac{1}{2} \ind{\iterationNumber > t}$ if $\updateRadius$ is sufficiently large, depending on $a$, $b$, $\range$, $\activity$ and $\dimensions$.
	Thus, applying the first part of \Cref{thm:additive_drift} proves our bound on $\E{\iterationNumber}[]$.
	
	To obtain the tail bound on $\iterationNumber$, we apply the second part of \Cref{thm:additive_drift}.
	For \ref{thm:additive_drift:start}, note that $\size{\boxesToFix{0}} = \size{\boxIds}$ and, for \ref{thm:additive_drift:step_size}, observe that 
	\[
		\absolute{\size{\boxesToFix{t}} - \size{\boxesToFix{t-1}}} \le \max\set{1, (2 \updateRadius + 3)^{\dimensions} \range^{\dimensions}} .
	\]
	Thus, setting $\alpha \ge 64 \cdot \max\set{1, (2 \updateRadius + 3)^{2\dimensions} \range^{2 \dimensions}}$ concludes the proof.
\end{proof}

Note that, in contrast to \Cref{cor:sampling_bayes_filter_HS}, \Cref{cor:sampling_bayes_filter} does not give a deterministic bound on the running time of each iteration.
Due to potential dependencies between the running time of each iteration and termination of the algorithm, it is unclear if we can simply apply Wald's equation to derive the total running time of the algorithm.
Instead, we will use a more subtle argument for this.
As a first ingredient, we need upper-bound the probability of ever observing a large number of points in any box $\boxRegion{\vectorize{v}}$ for $\vectorize{v} \in \boxIds$ up to a given iteration $k \in \N_{0}$.

\begin{lemma}\label{lemma:poisson_domination}
	There is a constant $\alpha \in \R_{\ge 1}$, only depending on $\updateRadius$, $\dimensions$ and $\range$, such that, for all $k \in \N_{0}$, all $\gamma > 1$ and all $x \ge \eulerE^{\gamma} \activity \range^{\dimensions}$, it holds that
	\[
		\Pr{\bigcup_{t = 0}^{k} \bigcup_{\vectorize{v} \in \boxIds} \set{\size{\currentPointset{t} \cap \boxRegion{\vectorize{v}}} \ge x}} 
		\le (k + 1) \alpha \eulerE^{-(\gamma - 1)x} .
	\] 
\end{lemma}
\begin{proof}
	Let $\alpha \ge \left(2 \updateRadius + 1\right)^{\dimensions}$.
	We show the statement via induction over $k$.
	First, for $k=0$, note that $\currentPointset{0} = \emptyset$.
	Thus, we have that the left-hand side is $0$ for all $x > 0$ and the right-hand side is at least $1$ for $x=0$, proving the base case.
	
	Next, assume the statement holds for some fixed $k \in \N_{\ge 0}$.
	Let $\hat{\boxesToUpdate} \subseteq \boxIds$ be the set of boxes updated in iteration $k+1$ (possibly the empty set).
	Formally, that is $\hat{\boxesToUpdate} = \set{\vectorize{v} \in \boxIds}[\currentPointset{k+1} \cap \boxRegion{\vectorize{v}} \neq \currentPointset{k} \cap \boxRegion{\vectorize{v}}]$.
	By the induction hypothesis and union bound, it holds that
	\begin{align*}
		\Pr{\bigcup_{t = 0}^{k+1} \bigcup_{\vectorize{v} \in \boxIds} \set{\size{\currentPointset{t} \cap \boxRegion{\vectorize{v}}} \ge x}} 
		&= \Pr{\bigcup_{t = 0}^{k} \bigcup_{\vectorize{v} \in \boxIds} \set{\size{\currentPointset{t} \cap \boxRegion{\vectorize{v}}} \ge x} \cup \set{ \exists \vectorize{v} \in \hat{\boxesToUpdate} \text{ s.t. } \size{\currentPointset{k+1} \cap \boxRegion{\vectorize{v}}} \ge x}} \\ 
		&\le \Pr{\bigcup_{t = 0}^{k} \bigcup_{\vectorize{v} \in \boxIds} \set{\size{\currentPointset{t} \cap \boxRegion{\vectorize{v}}} \ge x}} + \Pr{\set{ \exists \vectorize{v} \in \hat{\boxesToUpdate} \text{ s.t. } \size{\currentPointset{k+1} \cap \boxRegion{\vectorize{v}}} \ge x}} \\
		&\le (k + 1) \alpha \cdot \eulerE^{-(\gamma - 1)x} + \Pr{\set{ \exists \vectorize{v} \in \hat{\boxesToUpdate} \text{ s.t. } \size{\currentPointset{k+1} \cap \boxRegion{\vectorize{v}}} \ge x}}	.
	\end{align*}
	It now suffices to show that
	\[
	\Pr{\set{ \exists \vectorize{v} \in \hat{\boxesToUpdate} \text{ s.t. } \size{\currentPointset{k+1} \cap \boxRegion{\vectorize{v}}} \ge x}} \le \alpha \cdot \Pr{Y \ge x} .
	\]
	First, note that $\hat{\boxesToUpdate} = \emptyset$ if $\bayesFilter{k} = 0$.
	Thus, we have
	\[
	\Pr{\set{ \exists \vectorize{v} \in \hat{\boxesToUpdate} \text{ s.t. } \size{\currentPointset{k+1} \cap \boxRegion{\vectorize{v}}} \ge x}}[\bayesFilter{k} = 0] \le \alpha \cdot \Pr{Y \ge x} 
	\]
	for all $x \in \R_{\ge 0}$.
	Now, fix $S \subseteq \boxIds$ and $\vectorize{w} \in S$ such that $\Pr{\boxesToFix{k} = S, \boxChosen{k} = \vectorize{w}, \bayesFilter{k} = 1} > 0$.
	Set $\boxesToUpdate = \boxesToUpdate[S][\vectorize{w}][\updateRadius]$ and observe that, given $\boxesToFix{k} = S$, $\boxChosen{k} = \vectorize{u}$ and $\bayesFilter{k} = 1$, it holds that $\hat{\boxesToUpdate} = \boxesToUpdate$.
	Using union bound, we have
	\begin{align*}
		&\Pr{\set{ \exists \vectorize{v} \in \hat{\boxesToUpdate} \text{ s.t. } \size{\currentPointset{k+1} \cap \boxRegion{\vectorize{v}}} \ge x}}[\boxesToFix{k} = S, \boxChosen{k} = \vectorize{w}, \bayesFilter{k} = 1] \\
		&\quad\quad= \Pr{\bigcup_{\vectorize{v} \in \boxesToUpdate}\set{\size{\currentPointset{k+1} \cap \boxRegion{\vectorize{v}}} \ge x}}[\boxesToFix{k} = S, \boxChosen{k} = \vectorize{w}, \bayesFilter{k} = 1] \\
		&\quad\quad\le \sum_{\vectorize{v} \in \boxesToUpdate} \Pr{\size{\currentPointset{k+1} \cap \boxRegion{\vectorize{v}}} \ge x}[\boxesToFix{k} = S, \boxChosen{k} = \vectorize{w}, \bayesFilter{k} = 1] .
	\end{align*}
	Given $\boxesToFix{k} = S$, $\boxChosen{k} = \vectorize{w}$ and $\bayesFilter{k} = 1$, it holds that $\currentPointset{k+1} \cap \boxRegion{\boxesToUpdate}$ is sampled from $\gibbs[\boxesToUpdate][\eta \cap \boxesToUpdate^{\comp}]$ for some feasible $\eta \in \pointsets[\region]$.
	Therefore, for all $\vectorize{v} \in \boxesToUpdate$,  $\size{\currentPointset{k+1} \cap \boxRegion{\vectorize{v}}}$ is dominated by a Poisson random variable $Y$ with parameter $\activity \volume{\boxRegion{\vectorize{v}}} \le \activity \range^{\dimensions}$.
	Further, observing that $\size{\boxesToUpdate} \le \alpha$ yields
	\begin{align*}
		\Pr{\set{ \exists \vectorize{v} \in \hat{\boxesToUpdate} \text{ s.t. } \size{\currentPointset{k+1} \cap \boxRegion{\vectorize{v}}} \ge x}}[\boxesToFix{k} = S, \boxChosen{k} = \vectorize{w}, \bayesFilter{k} = 1] \le \alpha \cdot \Pr{Y \ge x} .
	\end{align*}
	Using the law of total expectation and \Cref{cor:poisson_tail}, we obtain
	\[
		\Pr{\set{ \exists \vectorize{v} \in \hat{\boxesToUpdate} \text{ s.t. } \size{\currentPointset{k+1} \cap \boxRegion{\vectorize{v}}} \ge x}} 
		\le \alpha \cdot \eulerE^{-(\gamma - 1)x} ,
	\]	
	which proves the claim.
\end{proof}

Using \Cref{lemma:poisson_domination}, we derive our main result on perfect sampling for repulsive bounded-range potentials based on \Cref{algo:sampling}.

\begin{theorem} \label{thm:sampling}
	Let $a, b > 0$ be such that $\potential$ satisfies $(a, b)$-strong spatial mixing up to $\activity$.
	Consider \Cref{algo:sampling} with $\bayesFilterCorrection(\cdot) = \eulerE^{-\eulerE^{-\updateRadius}}\bayesFilterCorrectionR{a}{b}(\cdot)$ as Bayes filter correction in line~\ref{algo:sampling:filter}.
	On termination, the algorithm outputs a sample from the Gibbs measure $\gibbs[\activity, \region]$.
	Moreover, if $\updateRadius$ is chosen as a sufficiently large constant, depending on $a$, $b$, $\range$, $\activity$ and $\dimensions$, then we can run the algorithm in expected time $\bigOTilde{\volume{\region}}$. 
\end{theorem}
\begin{proof}
	If $\potential$ satisfies $(a, b)$-strong spatial mixing up to $\activity$, we know by \Cref{cor:scaled_filter_repsulsive} that $\eulerE^{-\eulerE^{-\updateRadius}}\bayesFilterCorrectionR{a}{b}(\cdot)$ is a Bayes filter correction.
	Thus, the first part of the statement follows from \Cref{thm:correctness}.
	
	Next, let $\updateRadius$ be chosen as a sufficiently large constant, depending on $a, b, \range, \activity$ and $\dimensions$ to satisfy \Cref{lemma:iterations}.
	We rewrite the running time of \Cref{algo:sampling} as 
	\begin{align*}
		\E{\sum_{t=1}^{\iterationNumber} \iterationTime{t}}[]
		&= \E{\E{\sum_{t=1}^{\infty}\ind{\iterationNumber \ge t} \iterationTime{t}}[\boxesToFix{t-1}, \currentPointset{t-1}, \boxChosen{t-1}]}[] \\
		&= \E{\sum_{t=1}^{\infty} \ind{\iterationNumber \ge t} \E{\iterationTime{t}}[\boxesToFix{t-1}, \currentPointset{t-1}, \boxChosen{t-1}]}[] \\
	\end{align*}
	Moreover, by \Cref{cor:sampling_bayes_filter}, there are constants $a_1, a_2$, only depending on $\updateRadius, \dimensions, \activity$ and $\range$, such that
	\begin{align*}
		\E{\iterationTime{t}}[\boxesToFix{t-1}, \currentPointset{t-1}, \boxChosen{t-1}] \le \alpha_1 \size{\currentPointset{t-1} \cap (\partial \boxesToUpdate_{t} \cup \boxChosen{t-1})} + \alpha_2 
	\end{align*} 
	where $\boxesToUpdate_{t} = \boxesToUpdate[\boxesToFix{t-1}][\boxChosen{t-1}][\updateRadius]$.
	Thus, we have
	\[
		\E{\sum_{t=1}^{\iterationNumber} \iterationTime{t}}[]
		\le \alpha_1 \E{\sum_{t=1}^{\infty} \ind{\iterationNumber \ge t} \size{\currentPointset{t-1} \cap (\partial \boxesToUpdate_{t} \cup \boxChosen{t-1})}}[] + \alpha_2 \E{\sum_{t=1}^{\infty} \ind{\iterationNumber \ge t}}[] .
	\]
	Note that by \Cref{lemma:iterations}
	\[
		\sum_{t=1}^{\infty} \E{\ind{\iterationNumber \ge t}}[] 
		= \sum_{t'=0}^{\infty} \Pr{\ind{\iterationNumber > t'}}
		= \E{\iterationNumber}[] 
		\in \bigO{\volume{\region}}.
	\]
	It remains to bound
	\[
		\E{\sum_{t=1}^{\infty} \ind{\iterationNumber \ge t} \cdot \size{\currentPointset{t-1} \cap (\partial \boxesToUpdate_{t} \cup \boxChosen{t-1})}}[] .
	\]
	To this end, write $W = \sum_{t=1}^{\infty} \ind{\iterationNumber \ge t} \cdot \size{\currentPointset{t-1} \cap (\partial \boxesToUpdate_{t} \cup \boxChosen{t-1})}$.
	Since $W$ is non-negative, we have
	\[
		\E{W}[] 
		= \int_{\R_{\ge 0}} \Pr{W > w} \intD w 
		\le \hat{w} + \int_{\hat{w}}^{\infty} \Pr{W > w} \intD w
	\]
	for every $\hat{w} \in \R_{\ge 0}$.
	Next, observe that for every $t \in \N$
	\[
		\size{\partial \boxesToUpdate_{t} \cup \boxChosen{t-1}}
		\le \size{\partial \boxesToUpdate_{t} \cup \boxesToUpdate_{t}}
		\le \left(2 \updateRadius + 1\right)^{\dimensions}
		\eqqcolon \alpha_3.
	\] 
	Thus, for every $w \in \R_{\ge 0}$ and every $k: \R_{\ge 0} \to \N$ the following holds: if $W > w$, then  $\iterationNumber > k(w)$ or there is a time point $t \le k(w)$ and a box $\vectorize{v} \in \boxIds$ such that $\size{\currentPointset{t-1} \cap \boxRegion{\vectorize{v}}} > \frac{w}{k(w) \alpha_3}$.
	Consequently, we have  
	\begin{align*}
		\Pr{W > w} 
		&\le \Pr{\set{\iterationNumber > k(w)} \cup \bigcup_{t = 1}^{k(w)} \bigcup_{\vectorize{v} \in \boxIds} \set{\size{\currentPointset{t-1} \cap \boxRegion{\vectorize{v}}} > \frac{w}{k(w) \alpha_3}}} \\
		&\le \Pr{\iterationNumber > k(w)} + \Pr{\bigcup_{t = 1}^{k(w)} \bigcup_{\vectorize{v} \in \boxIds} \set{\size{\currentPointset{t-1} \cap \boxRegion{\vectorize{v}}} > \frac{w}{k(w) \alpha_3}}}.
	\end{align*}
	Moreover, applying \Cref{lemma:poisson_domination} shows that there is a constant $\alpha_4$, only depending on $\updateRadius$, $\range$ and $\dimensions$, such that for all $w \ge \eulerE^2 \alpha_3 \activity \range^{\dimensions} k(w)$ it holds that
	\[
		\Pr{W > w} \le \Pr{\iterationNumber \ge k(w)} + (k(w) - 1) \alpha_4 \exponential{-\frac{w}{k(w) \alpha_3}},
	\]
	Now, suppose that $k$ is measurable, we get for $\hat{w} \in \R_{\ge 0}$
	\[
		\E{W}[] \le \hat{w} + \int_{\hat{w}}^{\infty} \Pr{\iterationNumber \ge k(w)} \intD w + \alpha_4 \int_{\hat{w}}^{\infty} k(w) \exponential{-\frac{w}{k(w) \alpha_3}} \intD w .
	\]
	We now claim that for a suitable choice of $\hat{w}$ and $k$ the right-hand side is in $\bigOTilde{\volume{\region}}$.
	
	To this end, let $k(w) = \max\left\{\left\lceil 4 \eulerE \size{\boxIds}\right\rceil, \left\lceil 2 \alpha_5 \ln(w)\right\rceil\right\}$, where $\alpha_5$ is the constant from the tail bound in \Cref{lemma:iterations}, which only depends on $\updateRadius, \range$ and $\dimensions$.
	Moreover, choose $\hat{w} \ge 1$ sufficiently large such that for all $w \ge \hat{w}$ it holds that $w \ge k(w)$, $w \ge \eulerE^2 \alpha_3 \activity \range^{\dimensions} k(w)$ and $\frac{w}{\ln(w)} \ge 3 \alpha_3 k(w)$.
	Note that this can be achieved for some $\hat{w} \in \bigOTilde{\size{\boxIds}}$.
	
	For our choice of $k(w)$, \Cref{lemma:iterations} yields
	\[
		\int_{\hat{w}}^{\infty} \Pr{\iterationNumber \ge k(w)} \intD w \le \int_{\hat{w}}^{\infty} w^{-2} \intD w ,
	\]
	which is bounded by $1$ for $\hat{w} \ge 1$.
	Moreover, for $w \ge k(w)$ and $\frac{w}{\ln(w)} \ge 3 \alpha_3 k(w)$ we have 
	\[
		\exponential{-\frac{w}{k(w) \alpha_3}} \le \frac{1}{k(w) w^2} 
	\]
	and thus
	\[
		\int_{\hat{w}}^{\infty} k(w) \exponential{-\frac{w}{k(w) \alpha_3}} \intD w  \le \int_{\hat{w}}^{\infty} w^{-2} \intD w \le 1  .
	\]
	Consequently, we have
	\[
		\E{W}[] \le \hat{w} + \bigO{1} \in \bigOTilde{\size{\boxIds}}
	\]
	and, since $\size{\boxIds} \in \bigO{\volume{\region}}$, this concludes the proof.
\end{proof}

\begin{remark}
	Since \Cref{thm:sampling} requires knowing constants $a, b$ such that the Point process satisfies $(a, b)$-strong spatial mixing (in contrast to \Cref{thm:sampling_HS}), we can use these constant to compute a sufficiently large value for $\updateRadius$.
	Elementary calculations suggest to choose $\updateRadius > \max\set{\frac{1 + 7 \range}{2}, 16 a' {b'}^{2}}$, where $b' = \min\set{1, b \range}$ and $a' = 2 a 3^{\dimensions} \range^{\dimensions} \eulerE^{2 b \range} \left(\activity \range^{\dimensions} + \eulerE^{\activity 3^{\dimensions} \range^{\dimensions}}\right) + 1$.
\end{remark}

\section{Bernoulli Factories}
\label{sec:BernoulliFactories}

In this section we prove \Cref{lemma:bernoulli_frac}, showing how to sample a random variable from $\Ber{\frac{p}{q}}$ given access to a $\Ber{p}$ and $\Ber{q}$ sampler, when $q-p>\epsilon$. This happens in the following three steps.

The first step is to sample a random variable according to $\Ber{\frac{1 - (q - p)}2 }$ given access to $\Ber{p}$ and $\Ber{q}$. This is achieved by \Cref{BerAux} below.
\begin{algorithm}
	\caption{$\Ber{\frac{1 - (q - p)}2}$ from $\Ber{p}$ and $\Ber{q}$}\label{BerAux}
	Draw $u \sim \Ber{1/2}$ \\
	\If{$u = 1$}{
		Draw $y \sim \Ber{q}$ \\
		\Return $1-y$
	}
	\Else{Draw $x \sim \Ber{p}$ \\
	\Return $x$}
\end{algorithm}
It is easy to verify that this algorithm returns~1 with the correct probability.

For the second step, let $\varrho=\frac{1 - (q - p)}2$. From \Cref{BerAux} we now assume to have access to a $\Ber{\varrho}$ random variable. The next step is to use Huber's algorithm \cite{MR3506427} and obtain a $\Ber{2\varrho}$. For the algorithm to work within the required run-time guarantees, we need $\varrho<\frac{1-\epsilon}{2}$, which holds since we assumed $q-p>\epsilon$. For convenience, we provide the pseudocode of Huber's algorithm in \Cref{Ber2p}. The correctness of the algorithm can be found in \cite[Section~2.3]{MR3506427}.
\begin{algorithm}
	\caption{$\Ber{2\varrho}$ from $\Ber{\varrho}$}\label{Ber2p}
	$\epsilon \leftarrow \min\{\epsilon, 0.644\}$,  
	$k \leftarrow \frac{23}{5\epsilon}$,
	$i \leftarrow 1$,
	$R \leftarrow 1$,
	$C \leftarrow 2$ \\
	\While{$i \neq 0$ and $R\neq0$}{
		\While{$0 < i < k$}{
			Draw $u \sim \Ber{\varrho}$ \\
			Draw $g \sim Geo\left(\frac{C-1}C\right)$ \\
			$i \leftarrow i - 1 + (1-u)g$
		}
		\If{$i \geq k$}{
			Draw $R \sim Ber\left(\left(1 + \frac\epsilon2\right)^{-i}\right)$ \\
			\If{$R = 0$}{
				\Return $0$
			}
			$C \leftarrow C(1 + \frac\epsilon2)$,
			$\epsilon \leftarrow \frac\epsilon2$,
			$k \leftarrow 2k$ 
		}
	}
	\Return $1$
\end{algorithm}

\Cref{Ber2p} now gives us access to a $\Ber{2\varrho}$ sampler and consequently to a $\Ber{1-2\varrho}=\Ber{q-p}$ sampler, simply by flipping the returned bit. The final step is to sample from $\Ber{\frac{p}{q}}$ when given access to a sampler for $\Ber{q - p}$. This is done via \Cref{BernFrac}.



\begin{algorithm}
	\caption{$\Ber{\frac pq }$ from $\Ber{p}$ and $\Ber{q-p}$} \label{BernFrac}
	Set $b = -1$ \\
	\While{$b = -1$}{
		Draw $u_1 \sim \Ber{1/2}$ \\
		\If{$u_1 = 1$}{
			Draw $x \sim \Ber{p}$ \\
			\If{$x = 1$}{
				Set $b = 1$
			}
		}
		\Else{Draw $y \sim \Ber{q-p}$ \\
		\If{$y = 1$}{
			Set $b = 0$
		}
  }
	}
	\Return $b$
\end{algorithm}

Regarding the correctness of \Cref{BernFrac} note that, within a single while-loop, the probability the algorithm returns 1 is $p/2$, while the probability that the algorithm enters the while-loop again is $1-q/2$. Conditioned on the fact that the algorithm will terminate, we observe that \Cref{BernFrac} returns 1 with probability $p/q$.



We are now ready to prove \Cref{lemma:bernoulli_frac}, whose statement we repeat here for convinience.

\bernoulliFrac*
\begin{proof}
	We use \Cref{BernFrac} which calls \Cref{Ber2p}, which in turn calls \Cref{BerAux}, as we explained above.
	For simplicity, we may assume that, every time we sample $\Ber{p}$, we also sample from $\Ber{q}$ and vice versa.
	Thus, let $(X_i)_{i \in \N} \in (\set{0, 1}^2)^{\N}$ be a sequence of samples from the product distribution $(\Ber{p} \otimes \Ber{q})^{\otimes \N}$ and assume that $X_i = (X_i^{(p)}, X_i^{(q)})$ is the outcome of the $i$th time the algorithm samples from $\Ber{p}$ and $\Ber{q}$.
	Moreover, let $(T_i)_{i \in \N}$ be the running time for obtaining the $X_i$ using the assumed oracle, and let $N$ denote the total number of samples from $\Ber{p}$ and $\Ber{q}$ that \Cref{BernFrac} requires.
	
	Our goal is to show that $\E{\sum_{i = 1}^{N} T_i} \in \bigO{t\epsilon^{-2}}$.
	To this end, note that, for every $i \in \N$, the event $N \ge i$ does only depend on the sequence $(X_i)_{i \in \N}$ via the subsequence $(X_j)_{j < i}$.
	Moreover, by our assumptions on the oracle, it holds that $\E{T_i}[(X_j)_{j < i}] \le 2t$, where the factor of $2$ comes from the fact that we sample both $\Ber{p}$ and $\Ber{q}$.
	Thus, by Wald's equation, we have
	\[
		\E{\sum_{i = 1}^{N} T_i} 
		= \E{\sum_{i \ge 1} \ind{N \ge i} \E{T_i}[(X_j)_{j < i}]}
		\le 2t \E{N}.
	\]
	It remains to show that $\E{N} \in \bigO{\epsilon^{-2}}$. 
 
 \Cref{BernFrac} will do $2/q\in\bigO{\epsilon^{-1}}$ while-loops in expectation, as each while-loop terminates with probability $q/2$. Furthermore, each while-loop calls either $\Ber{p}$ or \Cref{Ber2p}, both with probability $1/2$. As the total number of loops is determined by the outcome of the final loop, we can use Wald's equation again to get that $\E{N}\in\bigO{\epsilon^{-1}+\epsilon^{-1}\E{N'}}$, where $N'$ is the number of $X_i$ samples that \Cref{Ber2p} requires.  Observe (from the pseudocode of \Cref{BerAux}) that each $\Ber{\varrho}$-call \Cref{Ber2p} requires only a single $X_i$ sample. From \cite[Theorem~1.1]{MR3506427} we get that \Cref{Ber2p} requires at most $\bigO{\epsilon^{-1}}$ $\Ber{\varrho}$ samples in expectation, which implies that $\E{N'}\in\bigO{\epsilon^{-1}}$. This concludes the proof of the lemma.
\end{proof}

\ifshowAuthors
\section*{Acknowledgments}
We thank Mark Jerrum for very helpful discussions on this topic. Konrad Anand was funded by a studentship from Queen Mary, University of London. Andreas Göbel was funded by the project PAGES (project No. 467516565) of the German Research Foundation (DFG).
Marcus Pappik was funded by the HPI Research School on Data Science and Engineering.  Will Perkins was supported in part by NSF grant  CCF-2309708. 
\fi

\bibliographystyle{abbrv}
\bibliography{References}

\appendix
\section{Measure theory and conditional expectations}

\subsection{Conditional expectation} \label{appendix:measures}
We start with a brief recap of the notation used in the appendix.

Let $(\Omega, \mathcal{A}, \PrSymbol)$ be a probability space.
For an event $A \in \sigmafield$ with $\Pr{A} > 0$, we write $\PrSymbol[A]$ for the probability measure $\Pr{\,\cdot}[][A] = \Pr{\,\cdot}[A]$
on $(\statespace, \sigmafield)$.
Note that for all events $A, B \in \sigmafield$ with $\Pr{A \cap B} > 0$ it holds that $\PrSymbol[A \cap B] = (\PrSymbol[A])_{B}$.
Let $f, g: \statespace \to \R$ be  measurable maps.
We denote by $\E{f}$ the expectation of $f$ under the measure $\PrSymbol$.
For a sub-$\sigma$-field $\mathcal{F} \subseteq \sigmafield$, we write $\E{f}[\mathcal{F}]$ as a placeholder for any version of a conditional expectation of $f$ given $\mathcal{F}$ under the probability measure $\PrSymbol$.
Further, we write $\E{\,\cdot}[f]$ for conditional expectations given $\sigma(f)$, the $\sigma$-field generated by $f$, and $\E{\,\cdot}[f, g]$ for conditional expectations given $\sigma(\sigma(f) \cup \sigma(g))$.
Finally, for an event $A \in \sigmafield$ with $\Pr{A} > 0$ and a sub-$\sigma$-field $\mathcal{F} \subseteq \sigmafield$, we write $\E{f}[][A]$ for the expectation of $f$ under $\PrSymbol[A]$ and $\E{f}[\mathcal{F}][A]$ for the conditional expectation of $f$ given $\mathcal{F}$ under the measure $\PrSymbol[A]$.

The following two statements allow us to relate conditional expectations under different probability distributions.

\begin{lemma} \label{lemma:conditioning_on_event}
	Let $(\Omega, \mathcal{A}, \PrSymbol)$ be a probability space, $X$ be an integrable random variable, let $\mathcal{F} \subseteq \mathcal{A}$ be a sub-$\sigma$-field and let $A \in \mathcal{A}$ with $\Pr{A} > 0$.
	Then 
	\[
		\E{\ind{A} X}[\mathcal{F}] = \E{X}[\mathcal{F}][A] \E{\ind{A}}[\mathcal{F}]  
	\]
	$\PrSymbol$-almost surely. 
	Moreover, it holds that
	\[
	\E{X}[\mathcal{F}][A] = \frac{\E{\ind{A} X}[\mathcal{F}]}{\E{\ind{A}}[\mathcal{F}]}
	\]
	$\PrSymbol[A]$-almost surely.
\end{lemma}

\begin{proof}
	By definition, $\E{X}[\mathcal{F}][A] \E{\ind{A}}[\mathcal{F}]$ is $\mathcal{F}$-measurable.
	Moreover, for any $B \in \mathcal{F}$ it holds that
	\begin{align*}
		\E{\ind{B} \E{X}[\mathcal{F}][A] \E{\ind{A}}[\mathcal{F}]}[]
		&= \E{\ind{B} \ind{A} \E{X}[\mathcal{F}][A]}[] \\
		&= \E{\ind{A}}[] \E{\ind{B} \E{X}[\mathcal{F}][A]}[][A] \\
		&= \E{\ind{A}}[] \E{\ind{B} X}[][A] \\
		&= \E{\ind{B} \ind{A} X}[] \\
		&= \E{\ind{B} \E{\ind{A} X}[\mathcal{F}]}[] .
	\end{align*}
	Thus, it holds that 
	\[
	\E{\ind{A} X}[\mathcal{F}] = \E{X}[\mathcal{F}][A] \E{\ind{A}}[\mathcal{F}]  
	\]
	$\PrSymbol$-almost surely.
	Next, observe that for $A \in \mathcal{A}$ with $\Pr{A} > 0$ it holds that $\E{A}[\mathcal{F}] > 0$ $\PrSymbol[A]$-almost surely.
	Thus, it follows immediately that
	\[
	\E{X}[\mathcal{F}][A] = \frac{\E{\ind{A} X}[\mathcal{F}]}{\E{\ind{A}}[\mathcal{F}]} 
	\]
	$\PrSymbol[A]$-almost surely.
\end{proof}

The following properties can be concluded.
\begin{lemma} \label{lemma:law_of_total_probability}
	Let $(\Omega, \mathcal{A}, \PrSymbol)$ be a probability space, $X$ be an integrable random variable and $\mathcal{F} \subseteq \mathcal{A}$ be a sub-$\sigma$-field.
	Let $A_1, \dots, A_n \in \mathcal{A}$ be disjoint and such that $\Pr{A_i} > 0$ for all $1 \le i \le n$ and $\Pr{\bigcup_{i=1}^{n} A_i} = 1$. 
	If an $\mathcal{F}$-measurable function $f: \Omega \to \R$ is a version of $\E{X}[\mathcal{F}][A_i]$ for all $1 \le i \le n$, then $f$ is also a version of $\E{X}[\mathcal{F}]$.
\end{lemma}

\begin{proof}
	Since the events $A_1, \dots, A_n \in \mathcal{A}$ and satisfy $\Pr{\bigcup_{i=1}^{n} A_i} = 1$, we have
	\[
	X 
	= X \ind{\bigcup_{i=1}^{n} A_i} 
	= \sum_{i=1}^{n} X \ind{A_i} 
	\]
	$\PrSymbol$-almost surely.
	Thus, by linearity of expectation, we have
	\[
	\E{X}[\mathcal{F}] 
	= \sum_{i=1}^{n} \E{X \ind{A_i}}[\mathcal{F}] .
	\]
	Furthermore, since $\Pr{A_i} > 0$ for all $1 \le i \le n$ \Cref{lemma:conditioning_on_event} gives
	\[
	\sum_{i=1}^{n} \E{X \ind{A_i}}[\mathcal{F}] 
	= \sum_{i=1}^{n} \E{X}[\mathcal{F}][A_i] \E{\ind{A_i}}[\mathcal{F}]
	= f \sum_{i=1}^{n} \E{\ind{A_i}}[\mathcal{F}] 
	\]
	$\PrSymbol$-almost surely.
	Finally, observing that
	\[
	\sum_{i=1}^{n} \E{\ind{A_i}}[\mathcal{F}] 
	= \E{\ind{\bigcup_{i=1}^{n} A_i}}[\mathcal{F}]
	= 1
	\]
	concludes the proof.
\end{proof}

\begin{lemma} \label{lemma:adding_conditions}
	Let $(\Omega, \mathcal{A}, \PrSymbol)$ be a probability space, $X$ be an integrable random variable, let $\mathcal{F} \subseteq \mathcal{A}$ be a sub-$\sigma$-field and let $A \in \mathcal{A}$ with $\Pr{A} > 0$.
	\begin{enumerate}[(1)]
		\item\label{lemma:adding_conditions:1} If $A \in \mathcal{F}$ then $\E{X}[\mathcal{F}] = \E{X}[\mathcal{F}][A]$ $\PrSymbol[A]$-almost surely.
		\item\label{lemma:adding_conditions:2} If  $\E{\ind{A}}[\mathcal{F}] = \E{\ind{A}}[\mathcal{F}, \mathcal{G}]$ for a $\sigma$-field $\mathcal{G} \subseteq \mathcal{A}$ with $\sigma(X) \subseteq \mathcal{G}$, then $\E{X}[\mathcal{F}] = \E{X}[\mathcal{F}][A]$ $\PrSymbol[A]$-almost surely. 
	\end{enumerate}
\end{lemma} 

\begin{proof}
	For \ref{lemma:adding_conditions:1}, note that for $A \in \mathcal{F}$ we have $\E{X \ind{A}}[\mathcal{F}] = \ind{A} \E{X}[\mathcal{F}]$.
	Moreover, \Cref{lemma:conditioning_on_event} yields
	\[
	\E{X \ind{A}}[\mathcal{F}] 
	= \E{\ind{A}}[\mathcal{F}] \E{X}[\mathcal{F}][A]
	= \ind{A} \E{X}[\mathcal{F}][A] .
	\]
	Thus, we have $\ind{A} \E{X}[\mathcal{F}] = \ind{A} \E{X}[\mathcal{F}][A]$ and, in particular, $\E{X}[\mathcal{F}] = \E{X}[\mathcal{F}][A]$ $\PrSymbol[A]$-almost surely.
	
	For \ref{lemma:adding_conditions:2}, observe that 
	\begin{align*}
		\E{X \ind{A}}[\mathcal{F}] 
		&= \E{\E{X \ind{A}}[\mathcal{F}, \mathcal{G}]}[\mathcal{F}] \\
		&= \E{X \E{\ind{A}}[\mathcal{F}, \mathcal{G}]}[\mathcal{F}] \\
		&= \E{X \E{\ind{A}}[\mathcal{F}]}[\mathcal{F}] \\
		&= \E{X}[\mathcal{F}] \E{\ind{A}}[\mathcal{F}]
	\end{align*}
	$\PrSymbol$-almost surely, where the second equality follows from $\sigma(X) \subseteq \mathcal{G}$ and the third follows from $\E{\ind{A}}[\mathcal{F}] = \E{\ind{A}}[\mathcal{F}, \mathcal{G}]$.
	The claim now follows from \Cref{lemma:conditioning_on_event}.
	
\end{proof}

\subsection{Regular conditional distributions} \label{appendix:rcd}
Consider a probability space $(\Omega, \mathcal{A}, \PrSymbol)$ with a sub-$\sigma$-field $\mathcal{F} \subseteq \mathcal{A}$, a measure space $(D, \mathcal{D})$ and $(D, \mathcal{D})$-valued random variable $X$.
A map $Q: \Omega \times \mathcal{D} \to [0, 1]$ is called a regular conditional distribution of $X$ given $\mathcal{F}$ if 
\begin{enumerate}[(1)]
	\item $Q(\omega, \cdot)$ is a probability measure on $(D, \mathcal{D})$ for all $\omega \in \Omega$ and
	\item $Q(\cdot, A)$ is a version of $\E{\ind{X \in A}}[\mathcal{F}]$ for all $A \in \mathcal{D}$.
\end{enumerate} 
The following statements makes regular conditional distributions particularly useful.


\begin{theorem}[{\cite{ccinlar2011probability}[Theorem $2.19$]}]\label{thm:integration_RCP}
	Let $(\Omega, \mathcal{A}, \PrSymbol)$ be a probability space, let $X$ be a $(D, \mathcal{D})$-valued random variable and let $Y$ be a $(E, \mathcal{E})$-valued random variable.
	If $Q: \Omega \times \mathcal{D}$ is a regular conditional distribution of $X$ given $\sigma(Y)$, then, for all $\mathcal{D} \otimes \mathcal{E}$-measurable $f: D \times E \to \R_{\ge 0}$ it holds that
	\[
	\E{f(X, Y)}[Y](\omega) = \int_{D} f(x, Y(\omega)) Q(\omega, \intD x) 
	\]
	for $\PrSymbol$-almost all $\omega \in \Omega$.
\end{theorem}\

Moreover, the following lemma helps to identify regular conditional distributions based on a $\pi$-system.
\begin{lemma} \label{lemma:extension_RCD}
	Let $(\Omega, \mathcal{A}, \PrSymbol)$ be a probability space and let $\mathcal{F} \subseteq \mathcal{A}$ be a sub-$\sigma$-field.
	Let $X$ be a $(D, \mathcal{D})$-valued random variable on $(\Omega, \mathcal{A}, \PrSymbol)$ and let $\mathcal{G} \subseteq \mathcal{D}$ be a $\pi$-system that generates $\mathcal{D}$.
	Assume there is a function $Q: \Omega \times \mathcal{D}$ such that 
	\begin{enumerate}[(1)]
		\item $Q(\omega, \cdot)$ is a probability distribution on $(D, \mathcal{D})$ for all $\omega \in \Omega$ and
		\item $Q(\cdot, G)$ is a version of $\E{\ind{X \in G}}[\mathcal{F}]$ for all $G \in \mathcal{G}$ .
	\end{enumerate}
	Then $Q(\cdot, A)$ is a version of $\E{\ind{X \in A}}[\mathcal{F}]$ for all $A \in \mathcal{D}$ and, in particular, $Q$ is a regular conditional distribution for $X$ given $\mathcal{F}$.	
\end{lemma} 

\begin{proof}
	Consider the set of events
	\[
	\mathcal{H} = \set{A \in \mathcal{D}}[Q(\cdot, A) \text{ is a version of } \E{\ind{X \in A}}[\mathcal{F}]] .
	\]
	Our goal is to prove $\mathcal{H} = \mathcal{D}$.
	To this end, note that $\mathcal{G} \subseteq \mathcal{H} \subseteq \mathcal{D}$.
	Thus, if we prove that $\mathcal{H}$ is a Dynkin system, then the $\pi$-$\lambda$ Theorem implies that $\mathcal{D} = \sigma(\mathcal{G}) \subseteq \mathcal{H}$, which proves our claim.
	To show that $\mathcal{H}$ is a Dynkin system, we need to argue that $D \in \mathcal{H}$ and that $\mathcal{H}$ is closed under complements and countable disjoint unions.
	
	To see that $D \in \mathcal{H}$, note that $Q(\omega, D) = 1$ for all $\omega \in \Omega$.
	Thus, $Q(\cdot, \Omega)$ is trivially $\mathcal{F}$-measurable.
	Moreover, for any $B \in \mathcal{F}$, it holds that
	\[
	\E{\ind{B} Q(\cdot, D)}[\mathcal{F}]
	= \E{\ind{B}}[\mathcal{F}]
	= \E{\ind{B} \ind{X \in D}}[\mathcal{F}] ,
	\]
	which shows that $Q(\cdot, D)$ is indeed a version of $\E{\ind{X \in D}}[\mathcal{F}]$.
	
	Next, fix some $G \in \mathcal{H}$ and observe that $Q(\omega, \complementOf{G}) = 1 - Q(\omega, G)$ for all $\omega \in \Omega$.
	Since $Q(\cdot, G)$ is by assumption a version of $\E{\ind{X \in G}}[\mathcal{F}]$ (therefore $\mathcal{F}$-measurable), this shows that $Q(\cdot, \complementOf{G})$ is $\mathcal{F}$-measurable.
	Moreover, for all $B \in \mathcal{F}$, we have
	\begin{align*}
		\E{\ind{B} Q(\cdot, \complementOf{G})}[\mathcal{F}] 
		&= \E{\ind{B}}[\mathcal{F}] - \E{\ind{B} Q(\cdot, G)}[\mathcal{F}] \\
		&= \E{\ind{B} \ind{X \in D}}[\mathcal{F}] - \E{\ind{B} \ind{X \in G}}[\mathcal{F}] \\
		&= \E{\ind{B} \ind{X \in \complementOf{G}}}[\mathcal{F}] ,
	\end{align*}
	which proves that $\complementOf{G} \in \mathcal{H}$.
	
	Finally, consider some sequence of disjoint events $(G_n)_{n \in \N} \in \mathcal{H}^{\N}$ and set $G= \bigcup_{n \in \N} G_n$. 
	Note that $Q(\omega, G) = \sum_{n \in \N} Q(\omega, G_n)$ for all $\omega \in \Omega$.
	Since each function $Q(\cdot, G_n)$ is a version of $\E{\ind{X \in G_n}}[\mathcal{F}]$, this implies that $Q(\cdot, G)$ is $\mathcal{F}$-measurable.
	Moreover, for all $B \in \mathcal{F}$ it holds that
	\[
	\E{\ind{B} Q(\cdot, G)}[\mathcal{F}] 
	= \sum_{n \in \N} \E{\ind{B} Q(\cdot, G_n)}[\mathcal{F}]
	= \sum_{n \in \N} \E{\ind{B} \ind{X \in G_n}}[\mathcal{F}]
	= \E{\ind{B} \ind{X \in G}}[\mathcal{F}],
	\] 
	showing that $G \in \mathcal{H}$.
	Thus, $\mathcal{H}$ is a Dynkin system, which concludes the proof.
\end{proof}

\section{Hitting times and tail bounds}

We frequently make use of the following version of Wald's identity.
\begin{lemma}\label{lemma:walds_equation}
	Let $(\Omega, \mathcal{A}, \PrSymbol)$ be a probability space, let $(X_n)_{n \in N_{\ge 1}}$ be a sequence of random variables on $(\Omega, \mathcal{A}, \PrSymbol)$ with values in $\R_{\ge 0}$ and let $\mathcal{F} \subseteq \mathcal{A}$ be a sub-$\sigma$-field. 
	Suppose there is a $\mathcal{F}$-measurable random variable $M$ such that for all $n \in \N$ it holds that $\E{X_n}[\mathcal{F}] \le M$ almost surely.
	Let $N$ be a random variable in $\N$ such that for all $n \in \N$ it holds that $\E{X_n \ind{N \ge n}}[\mathcal{F}] = \E{X_n}[\mathcal{F}] \E{\ind{N \ge n}}[\mathcal{F}]$ almost surely.
	Then $\E{\sum_{n=1}^{N} X_n}[\mathcal{F}] \le M \E{N}[\mathcal{F}]$ almost surely. 
\end{lemma}

\begin{proof}
	Using monotone convergence we have
	\begin{align*}
		\E{\sum_{n=1}^{N} X_n}[\mathcal{F}]
		&= \sum_{n=1}^{\infty} \E{\ind{N \ge n} X_n}[\mathcal{F}]
		= \sum_{n=1}^{\infty} \E{X_n}[\mathcal{F}] \E{\ind{N \ge n}}[\mathcal{F}] \\
		&\le M \E{\sum_{n=1}^{\infty} \ind{N \ge n}}[\mathcal{F}]
		= M \E{N}[\mathcal{F}] \qedhere
	\end{align*}
\end{proof}

Moreover, we use the following drift theorem to bound the expected number of iterations of our sampling algorithm.
\begin{theorem}[{\cite[Theorem 1]{lengler2020drift}, \cite[Theorem 2]{kotzing2014concentration}}] \label{thm:additive_drift}
	Let $(X_t)_{t \in \N_{0}}$ be an integrable random process over $\R$ that is adapted to a filtration $(\mathcal{F}_t)_{t \in \N_{0}}$ and let $T = \inf\set{t \in \N_{0} \mid X_t \le 0}$.
	Assume
	\begin{enumerate}[a)]
		\item $X_t \ind{T \ge t} \ge 0$ for all $t \in N$ and
		\label{thm:additive_drift:non_negative}
		\item there is some $\varepsilon \in \R_{>0}$ such that $\E{(X_{t} - X_{t+1}) \ind{T > t}}[\mathcal{F}_t] \ge \varepsilon \ind{T > t}$ for all $t \in N$.
		\label{thm:additive_drift:drift}
	\end{enumerate}
	Then $\E{T}[] \le \frac{\E{X_0}[]}{\varepsilon}$.
	Further, suppose that 
	\begin{enumerate}[a)]
		\setcounter{enumi}{2}
		\item $X_0 \le x$ for some $x \in \R_{> 0}$ and 
		\label{thm:additive_drift:start}
		\item there is some $c \in \R_{>0}$ such that $\absolute{X_{t} - X_{t+1}} < c$ for all $t \in \N_{0}$.
		\label{thm:additive_drift:step_size} 
	\end{enumerate}
	Then, for all $s \ge \frac{2 x}{\varepsilon}$, $\Pr{T \ge s} \le \exponential{- \frac{s \varepsilon^2}{16 c^2}}$.
\end{theorem}

Finally, we make use of the following tail bound for Poisson random variables.
\begin{theorem}[{\cite[Theorem 5.4]{mitzenmacher2017probability}}] \label{thm:poisson_tail}
	Let $Y \sim Pois(\rho)$ for some $\rho \in \R_{>0}$.
	For all $y > \rho$ it holds that $\Pr{Y \ge y} \le \eulerE^{- \rho} \left(\frac{\eulerE \rho}{y}\right)^{y}$.
\end{theorem}

In particular, we use the following corollary of the above bound.
\begin{corollary}\label{cor:poisson_tail}
	Let $Y \sim Pois(\rho)$ for some $\rho \in \R_{>0}$.
	For all $\gamma > 1$ all $y \ge \eulerE^{\gamma} \rho$ it holds that $\Pr{Y \ge y} \le \eulerE^{- (\gamma - 1) y}$.
\end{corollary}
\begin{proof}
	Since $y \ge \eulerE^{\gamma} \rho > \rho$, \Cref{thm:poisson_tail} implies that
	\[
	\Pr{Y \ge y} 
	\le \eulerE^{- \rho} \left(\frac{\eulerE \rho}{y}\right)^{y}
	\le \left(\frac{\eulerE \rho}{\eulerE^{\gamma} \rho}\right)^{y}
	= \eulerE^{- (\gamma - 1) y}. \qedhere
	\]
\end{proof}

\section{Gibbs point processes} \label{appendix:gpps}

Here we collect some useful lemmas about Gibbs point processes.

The following technical lemma will come in handy.
\begin{lemma} \label{lemma:pi_system}
	Let $\region \in \borel$ and let $\region_1, \region_2 \subseteq \region$ be a partitioning of $\region$ into measurable sets.
	Let $\mathcal{G}$ denote all events of the form $\{\eta \in \pointsets[\region] \mid \eta \cap \region_1 \in A_1, \eta \cap \region_2 \in A_2\}$ with $A_1 \in \pointsetEvents[\region_1]$ and $A_2 \in \pointsetEvents[\region_2]$.
    Then $\mathcal{G}$ is a $\pi$-system that satisfies $\mathcal{G} \subseteq \pointsetEvents[\region] \subseteq \sigma(\mathcal{G})$.
\end{lemma}

\begin{proof}
    Seeing that $\mathcal{G}$ is a $\pi$-system is trivial.
	For showing that $\mathcal{G} \subseteq \pointsetEvents[\region]$, it suffices to show that for $i \in \{1, 2\}$ the projection $\projection{\region_i}: \eta \mapsto \eta \cap \region_i$, viewed as a map $\pointsets \to \pointsets[\region_i]$, is $\pointsetEvents$-$\pointsetEvents[\region_i]$-measurable.
    To this end, note that for every $k \in \N_0 \cup \{\infty\}$ and $\region' \subseteq \region_i$ the event $E = \{\eta \in \pointsets[\region_i] \mid \size{\eta \cap \region'} = k\} \in \pointsetEvents[\region_i]$ satisfies 
    \[
        \projection{\region_1}^{-1}(E) = \{\eta \in \pointsets \mid \size{\eta \cap \region'} = k\} \in \pointsetEvents.
    \]
    As such events $E$ generate $\pointsetEvents[\region_i]$, the desired measurability of $\projection{\region_i}$ follows, which proves the inclusion $\mathcal{G} \subseteq \pointsetEvents[\region]$.

    For $\pointsetEvents[\region] \subseteq \sigma(\mathcal{G})$, let $E = \{\eta \in \pointsets[\region] \mid \size{\eta \cap \region'} = k\} \in \pointsetEvents[\region]$ for some $k \in \N_0 \cup \{\infty\}$ and $\region' \subseteq \region_i$, and observe that such events generate $\pointsetEvents[\region]$.
    Next, observe that such an event $E$ can be expressed as
    \[
        E = \bigcup_{k_1 + k_2 = k} \{\eta \in \pointsets[\region] \mid \size{\eta \cap \region_1 \cap \region'} = k_1, \size{\eta \cap \region_2 \cap \region'} = k_2\}.
    \]
    As, for each $i \in \{1, 2\}$, it holds that $\{\eta' \in \pointsets[\region_i] \mid \size{\eta' \cap (\region_i \cap \region')} = k_i\} \in \pointsetEvents[\region_i]$, we have thus written $E$ as a countable union of element of $\mathcal{G}$, which concludes the proof.
\end{proof}

In this bounded-range setting, the following lemma might be seen as a version of the spatial Markov property for partition functions.
\begin{lemma} \label{lemma:spatial_markov}
	Let $\potential$ be a repulsive potential with bounded range $\range \in \R_{\ge 0}$.
	Moreover, let $\subregion \subseteq \region$ and let be any activity function $\activityFunction: \R^{\dimensions} \to \R_{\ge 0}$.
	For every two point configurations $\eta_1, \eta_2 \in \pointsets[\region]$ with $\dist{\subregion}{\eta_1 \symmDiff \eta_2} \ge \range$, where $\symmDiff$ denotes the symmetric difference, it holds that $\partitionFunction[\subregion][\eta_1](\activityFunction) = \partitionFunction[\subregion][\eta_2](\activityFunction)$.
\end{lemma}
\begin{proof}
	Since the range of $\potential$ is bounded by $\range$, it holds that $\activityFunction_{\eta_1}(x) = \activityFunction_{\eta_2}(x)$ for all $x \in \region$.
	Therefore, we have
	\begin{align*}
		\partitionFunction[\subregion][\eta_1](\activityFunction)
		= \sum_{k \ge 0} \frac{1}{k!} \int_{\subregion^k} \activityFunction_{\eta_1}^{\vectorize{x}} \eulerE^{-\hamiltonian[\vectorize{x}]} \intD \vectorize{x} 
		= \sum_{k \ge 0} \frac{1}{k!} \int_{\subregion^k} \activityFunction_{\eta_2}^{\vectorize{x}} \eulerE^{-\hamiltonian[\vectorize{x}]} \intD \vectorize{x} 
		= \partitionFunction[\subregion][\eta_2](\activityFunction),
	\end{align*}
	which proves the claim.
\end{proof}

\end{document}